\documentclass[10pt]{article}
\usepackage{latexsym}
\usepackage{theorem}
\usepackage{color,graphicx}
\usepackage{amssymb}
\usepackage{amsmath}
\usepackage{txfonts}
\usepackage{enumitem}
\usepackage{fullpage}

\newenvironment{proof}[1][Proof]
{\par\noindent{\bf #1:} }{\hspace*{\fill}\nolinebreak{$\Box$}\bigskip\par}
\newcommand{\qed}{\hspace*{\fill}\nolinebreak\ensuremath{\Box}}

\newcommand{\complTime}[2]{C(#1,#2)}
\newcommand{\startTime}[2]{s(#1,#2)}

\newtheorem{theorem}{Theorem}
\newtheorem{lemma}{Lemma}[section]
\newtheorem{claim}[lemma]{Claim}
\newtheorem{corollary}[theorem]{Corollary}
\newtheorem{observation}[lemma]{Observation}
\newtheorem{proposition}[theorem]{Proposition}
\theorembodyfont{\upshape}

\newcommand{\cP}{\mathcal{P}}

\newcommand{\reals}{\mathbb{R}}
\newcommand{\st}{\hspace{0.1cm}\bigl|\bigr.\hspace{0.1cm}}
\newcommand{\vect}[1]{\mbox{\boldmath$#1$}}

\newcommand{\jobs}{\mathcal{J}}
\newcommand{\A}{\textup{A}}
\newcommand{\cyclic}[1]{\overset{#1}{\rightsquigarrow}}
\newcommand{\cyclicshift}[1]{\big\langle #1 \big\rangle}

\begin{document}

\title{Structural Properties of an Open Problem in Preemptive Scheduling}

\author{Bo Chen\footnote{Centre for Discrete Mathematics and Its Applications (DIMAP) and Warwick Business School, University of Warwick, Coventry, UK, {bo.chen@wbs.ac.uk}}
 \and Ed Coffman\footnote{Departments of Electrical Engineering and of Computer Science, Columbia University, New York, USA, {coffman@cs.columbia.edu}}
 \and Dariusz Dereniowski\footnote{Faculty of Electronics, Telecommunication and Informatics, Gda{\'n}sk University of Technology, Gda{\'n}sk, Poland, {deren@eti.pg.gda.pl}}
 \and Wies{\l}aw Kubiak\footnote{Faculty of Business Administration, Memorial University, St. John's, Canada, {wkubiak@mun.ca}}
}

\maketitle

\begin{abstract}
Structural properties of optimal preemptive schedules have been studied in a number of recent papers with a primary focus on two structural parameters: the minimum number of preemptions necessary, and a tight lower bound on \emph{shifts}, i.e., the sizes of intervals bounded by the times created by preemptions, job starts, or completions. These two parameters have been investigated for a large class of preemptive scheduling problems, but so far only rough bounds for these parameters have been derived for specific problems. This paper sharpens the bounds on these structural parameters for a well-known open problem in the theory of preemptive scheduling: Instances consist of in-trees of $n$ unit-execution-time jobs with release dates, and the objective is to minimize the total completion time on two processors. This is among the current, tantalizing ``threshold'' problems
of scheduling theory:  Our literature survey reveals that any significant generalization leads to an NP-hard problem, but that any significant simplification leads to tractable problem with a polynomial-time solution.

For the above problem, we show that the number of preemptions necessary for optimality need not exceed $2n-1$; that the number must be of order $\Omega(\log n)$ for some instances;  and that the minimum shift need not be less than $2^{-2n+1}.$   These bounds are obtained by combinatorial analysis of optimal preemptive schedules rather than by the analysis of polytope corners for linear-program formulations of the problem, an approach to be
found in earlier papers. The bounds immediately follow from a fundamental structural property called \emph{normality}, by which minimal shifts of a job are exponentially decreasing functions. In particular, the first interval between a  preempted job's start and its preemption must be a multiple of $1/2$, the second such interval must be a multiple of $1/4$, and in general, the $i$-th preemption must occur at a multiple of $2^{-i}$. We expect the new structural properties to play a prominent role in finally settling a vexing, still-open question of complexity.
\end{abstract}


\section{Introduction}

We study structural properties of optimal preemptive schedules of a classic problem of scheduling UET (Unit Execution Time) jobs with precedence constraints and release dates on two processors.
Optimal nonpreemptive schedules for this and related problems have been well
researched in the literature for various objective functions and restrictions. Fujii, Kasami and Ninomiya \cite{FKN} present a matching-based algorithm, and Coffman and Graham \cite{CG2} devise a  job-labeling algorithm for minimum-makespan nonpreemptive schedules. Garey and Johnson
introduce $O(n^2)$ and $O(n^{2.81})$ time algorithms for minimizing maximum
lateness for jobs, respectively without  release dates \cite{GareyJohnson-deadlines},
and with release dates  \cite{GareyJohnson-2-proc}.
Gabow \cite{Gabow} designed an almost linear-time algorithm for the
minimum-makespan problem.  Leung, Palem, and Pnueli \cite{LeungPalemPnueli} and Carlier, Hanen, and Munier-Kordon \cite{Alix} extend these results to precedence delays. Baptiste
and Timkovsky \cite{BaptisteTimkovsky_shortest_paths} focus on
minimization of total completion time and present an $O(n^9)$ time
shortest-path optimization algorithm for scheduling jobs with release
dates. They also conjecture that there always exist so-called \textit{ideal} schedules that minimize
both maximum completion time and total completion time for jobs with release
dates. This has been known to hold true
for equal release dates without preemptions (Coffman and Graham \cite{CG2})
and with preemptions (Coffman, Sethuraman and Timkovsky  \cite{ideal_preeemptive}). Coffman, Dereniowski and Kubiak \cite{CoffmanDereniowskiKubiak} prove the Baptiste-Timkovsky conjecture and
give an $O(n^3)$ algorithm for the minimization of total completion time for jobs
with release dates -- a major improvement over the $O(n^9)$ time algorithm in \cite{BaptisteTimkovsky_shortest_paths}.

Optimal preemptive schedules have proven more challenging to compute efficiently, especially for jobs with release dates and the total-completion-time criterion. Coffman, Dereniowski, and Kubiak \cite{CoffmanDereniowskiKubiak} prove that these schedules are not ideal, that is, for some instances any schedule minimizing total completion time will be longer than the schedule minimizing maximum completion time. That holds even for in-tree precedence constraints.
This last result serves as a point of departure for this paper, with its focus on in-tree precedence constraints, release dates, and the criterion of total-completion-time, the problem $P2|pmtn,in\textup{-}tree,r_j,p_j=1|\sum C_j$ in the well-known three-field notation. Despite numerous efforts, the computational complexity of the problem remains open: reducing the number of processors to $m=1$ renders the problem polynomially solvable (Baptiste et al.~\cite{BaptisteBKT04}); and so does dropping the precedence constraints
(Herrbach, Lee and Leung \cite{HeL90}); dropping the release dates (Coffman, Sethuraman and Timkovsky  \cite{ideal_preeemptive}); and
assuming out-trees instead of in-trees (Baptiste and Timkovsky \cite{BaptisteT01}). With this background in mind, we focus on key structural properties of optimal preemptive schedules for the problem $P2|pmtn,in\textup{-}tree,r_j,p_j=1|\sum C_j$.

Sauer and Stone \cite{SS87} study the problem with no release dates and maximum-completion-time (makespan) minimization. They show that, for every optimal preemptive schedule, there is an optimal preemptive schedule with at most $n$ preemptions, where preemptions occur at multiples of $1/2$, and go on to define a  \emph{shift} that is the duration of an interval between two consecutive time points, each of which is a job start, a job end, or a job preemption.
The shortest necessary shift in an optimal schedule is then called its \emph{resolution}.  The minimum resolution over all instances of a given preemptive scheduling problem is called the \emph{problem resolution}.  Following \cite{SS87},  the minimum number of preemptions and the minimum resolution necessary for optimal schedules have become two main structural parameters in preemptive scheduling. They have been investigated for a large class of preemptive scheduling problems by Baptiste et al.\ \cite{BaptisteCKQSS11} who give general bounds for these parameters.

Coffman, Ng and Timkovsky \cite{CNT13} provide bounds on the resolutions of various scheduling problems --- we refer the reader to their work for a comprehensive overview.
In particular, they show upper bounds of $m^{-n/(m+1)}$ and $m^{-(n-1)/(m+1)}$ on resolutions for problems $P|pmtn,in\textup{-}tree,r_j,p_j=1|C_{\max}$ and $P|pmtn,in\textup{-}tree,r_j,p_j=1|\sum C_j$, respectively, where $n$ is the number of jobs and $m$ is the number of processors.
Thus, for the problem $P2|pmtn,in\textup{-}tree,r_j,p_j=1|\sum C_j$ studied in this paper one immediately obtains an upper bound of $2^{-(n-1)/3}$ on its resolution.
As for lower bounds, \cite{CNT13} shows that the resolution of $P|pmtn,prec,r_j|\sum w_iC_j$ is at least $(m+n)^{-(2n+1)/2}$ .

The papers of Sauer and Stone \cite{SS87}, Baptiste et al.\ \cite{BaptisteCKQSS11} and Coffman, Ng and Timkovsky \cite{CNT13} obtain their resolution bounds by analyzing the corners of feasibility regions of linear programs designed for specific problems. Our approach is combinatorial and does not make use of the theory of linear programming. It yields a lower bound of $2^{-2n-1}$ on the problem resolution of $P2|pmtn,in\textup{-}tree,r_j,p_j=1|\sum C_j$, which is a significant improvement over the lower bound of $(n+2)^{-(2n+1)/2}$ that can be derived directly from \cite{CNT13}.

We introduce in this paper the concept of \emph{normal} schedules where shifts decrease as a function of time: The first shift is a multiple of $1/2$, the second one is a multiple of $1/4$, and in general, the $i$-th shift is a multiple if $2^{-i}$. We prove that there exist optimal schedules that are normal for in-trees. However, we conjecture that this is no longer the case for arbitrary precedence constraints, i.e., there are instances for which no optimal schedule is normal. The normality of a schedule implies that each shift is a multiple of $2^{-2n+1}$, which is a much stronger claim than the usual requirement that all shifts are no shorter than the problem resolution. Normality also implies that there exists an optimal schedule with a finite number (in particular, a number not exceeding $2n-1$) of events which are times when jobs start, end, or are preempted. Thus, $2n-1$ is an upper bound on the number of preemptions necessary for optimality.  We also observe that a job may be required to preempt at a point which is neither a start nor the end of another job in order to ensure optimality. These preemption events unrelated to job starts or completions seem to be confined to rather contrived instances; they are more the exception than the rule in preemptive scheduling. We also prove that there exists a sequence of problem instances indexed by $n$ for which the number of
preemptions in the corresponding optimal schedules is $\Omega(\log n)$. 
Thus, a tight upper bound on the number of preemptions required for optimality must be at least logarithmic in $n$.

\section{Our approach and results: A general overview}
\label{sec:normal_preemptions}

We first show that an optimal schedule is a concatenation of \emph{blocks}, each with at most three jobs. No job starts or completes inside
a block but there is at least one job start at the beginning of a block,
and/or at least one job completion at the end of a block. This is done in Sections \ref{subsec:events}
and \ref{subsec:opt_properties}. A block is called $l$-\emph{normal} if each job duration in the block is a multiple of $1/2^{l+1}$, and the block length is a multiple of $1/2^{l}$. In a normal schedule the first block must be $1$-normal, the second $2$-normal and so on. These concepts are introduced in Section \ref{subsec:abnormality_points}, where it is verified that, in a normal schedule with $q$ blocks, each preemption occurs at a multiple of $1/2^{q+1},$ where $q\leq 2n-1$. Our goal is to show that there exists an optimal schedule that is normal. Our proof is by contradiction. We begin by assuming an  optimal schedule that is also \emph{maximal} in the sense
that it has a latest possible \emph{abnormality} point $i$, i.e., a latest block $i$ which is not $i$-normal.   We show that such a block must have exactly three jobs. One completes at the end of the block and has an $(i+1)$-normal duration, but the durations of the other two  are not $(i+1)$-normal, as shown by Lemma~\ref{lem:at_least_two_abnormal}. These two jobs then trigger an \emph{alternating} chain of jobs to which they also belong, as shown in Section~\ref{subsec:alternating_chains}. The completion times of the jobs in the chain are not $(i+1)$-normal, which makes it possible under normal-block circumstances to either extend the chain by one
job or prove that the abnormality point must exceed $i$; this is our main result in Proposition~\ref{pro:infinite_chain}. Thus, we get a contradiction in either case since the number of jobs is finite and the schedule is maximal.
The normal-block circumstances here mean that the alternating chain does not end with a certain structure that we call an \emph{A-configuration}, a configuration that prevents us from extending the alternating chain. However, we show that there always exists a maximal schedule that does not include an $\A$-configuration. This is done in
Section~\ref{subsec:A-configurations}, where the key result is Proposition~\ref{pro:no-A-configurations}.
The main result of the paper follows and states that there is a normal schedule that is optimal for $P2|pmtn,in\textup{-}tree,r_j,p_j=1|\sum C_j$. Finally, in Section~\ref{sec:lower_bound} we exhibit sequences of problem instances indexed by $n$ for which the rate at which the number of preemptions
increases is on the order of $\log n$.

\section{Optimal, normal and maximal schedules}

\subsection{Preliminaries}
\label{subsec:preliminaries}

Let $\jobs$ be a set of $n$ unit UET jobs.
The \emph{release date} for job $a,$ denoted by $r(a),$ is the earliest start time for $a$ in any valid schedule  of $\jobs$. We assume that $r(a)$ is an integer for all $a\in\jobs$.

For two jobs $a$ and $b$, we say that $a$ is a \emph{predecessor} of $b$, and that $b$ is a \emph{successor} of $a,$ if all valid schedules require that $b$ not start until $a$ has finished.
We write $a\prec b$ to denote this relation.  In contrast,
$a\nprec b$ means that $b$ can start prior to the completion time of $a$.
Two jobs $a$ and $b$ are said to be \emph{independent} if $a\nprec b$ and $b\nprec a$.
For $B\subseteq\jobs$, we say that the jobs in $B$ are \emph{independent} if each pair of jobs in $B$ is independent.
This work deals with \emph{in-tree} precedence constraints, i.e., for each job $a$ there exists at most one job $b$ such that $a\prec b$.

The symbol $\reals_+$ denotes the set of nonnegative real numbers.
Given a schedule $\cP$ and a job $a\in\jobs$, define $\startTime{\cP}{a}$ and $\complTime{\cP}{a}$ to be the start and completion times of $a$ in $\cP$, respectively. A job is called \emph{release date pinned} in $\cP$ if it starts at its release date in $\cP$.
The \emph{total completion time} of a schedule $\cP$ of $\jobs$ is given by $\sum_{a\in\jobs}\complTime{\cP}{a}$.
We say that a preemptive schedule $\cP$ is \emph{optimal} if the sum of its job completion times is minimum among all preemptive schedules for $\jobs$.

\subsection{Events, partitions and basic schedule transformations} \label{subsec:events}

For a given schedule $\cP$, define a vector $\vect{e}=(e_1,\ldots,e_q)$, where $0=e_1<e_2<\cdots<e_q$, such that
\[\{e_1,\ldots,e_q\}=\{0\}\cup \left\{\startTime{\cP}{a}\st a\in\jobs\right\}\cup\left\{\complTime{\cP}{a}\st a\in\jobs\right\}.\]
The elements of $\vect{e}$ are called the \emph{events} of $\cP$.
The part of $\cP$ in time interval $[e_i,e_{i+1}]$ is called the \emph{$i$-th block} of $\cP$, or simply a \emph{block} of $\cP$, $i\in\{1,\ldots,q-1\}$.
Given $i\in\{1,\ldots,q-1\}$, let $\xi_i\colon\jobs\to\reals_+$ be a function such that for each $a\in\jobs$, $\xi_i(a)$ is the total length of $a$ executed in the $i$-th block of $\cP$.
Then, $(\xi_1,\ldots,\xi_{q-1})$ is called the \emph{partition} of $\cP$.
Denote by $(\cP,\vect{e},\vect{\xi})$ the schedule $\cP$ with events $\vect{e}$ and partition $\vect{\xi}$.
Unless specified otherwise, it is understood that $\vect{e}$ has $q$ components.
For each $a\in\jobs$, $\tau_{\cP}(a)$ is the integer $i\in\{1,\ldots,q-1\}$ such that $\complTime{\cP}{a}=e_{i+1}$.
In other words, the $\tau_{\cP}(a)$-th block is the last block in which job $a$ appears.
Whenever $\cP$ is clear from context we will simply write $\tau(a)$.
For any function $f\colon\jobs\to\reals_+$, let
\[\jobs(f)=\left\{ a\in\jobs\st f(a)\neq 0 \right\}.\]

In the following we will analyze schedules by investigating their events and partitions.
Informally speaking, the events and the partition of a schedule $\cP$ are insufficient to uniquely reconstruct the schedule $\cP$ but they suffice to build a schedule with the same total completion time as $\cP$.
The schedules built from a list of events and a partition may differ in how pieces of jobs are executed within the blocks.
The main advantage of our approach is that in order to construct a block in $[e_i,e_{i+1}]$ one only needs to solve the problem $P2|p_j,pmtn|C_{\max}$ where the execution time of a job $a$ is $\xi_i(a)$; the proof of Lemma~\ref{lem:xi_to_P} gives more details.
We formalize this observation in the next two lemmas.
\begin{lemma} \label{lem:P_to_xi}
Given a schedule $(\cP,\vect{e},\vect{\xi})$, for each $i\in\{1,\ldots,q-1\}$, the following hold:
\begin{enumerate} [label={\normalfont(\roman*)},leftmargin=*]
 \item\label{it:P_to_xi:1} For each $a\in\jobs(\xi_i)$: $r(a)\leq e_i$;
 \item\label{it:P_to_xi:2} For each $a\in\jobs(\xi_i)$: $\xi_i(a)\leq e_{i+1}-e_i$ and
       $\sum_{a\in\jobs(\xi_i)}\xi_i(a)\leq 2(e_{i+1}-e_i)$;
 \item\label{it:P_to_xi:3} For each $a\in\jobs(\xi_i)$ and $b\in\jobs(\xi_j)$, where $i\leq j<q$: $b\nprec a$.
\end{enumerate}
\end{lemma}
\begin{proof}
Condition \ref{it:P_to_xi:1} follows from the fact that no job in $\jobs(\xi_i)$ starts or completes in $(e_i,e_{i+1})$, $i\in\{1,\ldots,q-1\}$.
(Note that $r(a)>e_i$ is not possible for $a\in\jobs(\xi_i)$ because then we would have $\startTime{\cP}{a}\in(e_i,e_{i+1})$ which would contradict $e_i$ and $e_{i+1}$ being two consecutive events of $\cP$.)
Conditions \ref{it:P_to_xi:2} and \ref{it:P_to_xi:3} follow directly from the fact that $\cP$ is a feasible schedule for $\jobs$.
(Note that \ref{it:P_to_xi:3} in particular implies that the jobs in $\jobs(\xi_i)$ are independent.)
\end{proof}

We often rely on rearrangements of the events $\vect{e}$ of a schedule $\cP$ which result in new schedules $\cP'$ with events that differ from those in  $\vect{e}$. The resulting schedule $\cP'$, however, may still be analyzed in the time intervals $[e_i,e_{i+1}]$, $i\in\{1,\ldots,q-1\}$ defined by the original $\vect{e}$.
For this analysis, we need the following lemma, in which vectors of increasing real numbers beginning with 0 are regarded as  sequences of \emph{time points}.

\begin{lemma} \label{lem:xi_to_P}
If there exist $q$ time points $e_1 < \cdots< e_q$ and $q-1$ functions $\xi_i\colon\jobs\to\reals_+$ ($i=1,\ldots,q-1$)
such that for each $a\in\jobs$, $\sum_{i=1}^{q-1}\xi_i(a)=1$
and conditions \ref{it:P_to_xi:1}--\ref{it:P_to_xi:3} in Lemma~\ref{lem:P_to_xi} are satisfied,
then there exists a schedule $\cP$ such that for each $i\in\{1,\ldots,q-1\}$ and for each $a\in\jobs$ the total
length of all pieces of $a$ executed in $[e_i,e_{i+1}]$ equals $\xi_i(a)$.
\end{lemma}
\begin{proof}
For any given $i\in\{1,\ldots,q-1\}$,
it is enough to construct the part of schedule $\cP$, denoted by $\cP_i$, in the time interval $[e_i,e_{i+1}]$.
By \ref{it:P_to_xi:1} and \ref{it:P_to_xi:2}, this is equivalent to solving the problem $P2|p_j,pmtn|C_{\max}$ where the execution time of each job $a$ is $\xi_i(a)$. It is easy to see that
such a schedule $\cP_i$ exists if and only if the duration of $[e_i,e_{i+1}]$ is at least
the larger of the maximum of the execution times $\xi(a)$ and the sum of these times  averaged over
the two processors, i.e.,
\[e_{i+1}-e_i\geq\max\left\{\frac{1}{2}\sum_{a\in\jobs}\xi_i(a),\max\left\{\xi_i(a)\st a\in\jobs\right\}\right\}.\]
Thus, \ref{it:P_to_xi:2} guarantees that $\cP_i$ exists.
Finally note that \ref{it:P_to_xi:3} guarantees that the precedence constraints between jobs in different blocks are met.
\end{proof}

We close this section by introducing two basic transformations of a given schedule $(\cP,\vect{e},\vect{\xi})$: the \emph {cyclic shift} and the \emph{swapping} of two jobs.
Let $\varepsilon>0$ and $j>0$.
Let $B=\{a_1,\ldots,a_j\}\subseteq\jobs$ be $j$ different jobs and $\{i_1,\ldots,i_j\}\subseteq\{1,\ldots,q-1\}$ be $j$ blocks of $\cP$ such that $\xi_{i_k}(a_k)\geq \varepsilon$ and $\xi_{i_{k+1}}(a_k)\leq e_{i_{k+1}+1}-e_{i_{k+1}}-\varepsilon$ for $k\in\{1,\ldots,j\}$, where $i_{j+1}=i_{1}$.
We define a \emph {cyclic shift} of $B$ by $\varepsilon$ on $\{i_1,\ldots,i_j\}$ in $\cP$, or just a cyclic shift if it is clear from context, as follows.
Let
\[(\vect{e}',\vect{\xi}')=\cyclicshift{\vect{e},\vect{\xi},\varepsilon,(i_1\cyclic{a_1} i_2 \cyclic{a_2} \ldots \cyclic{a_{j-1}} i_j\cyclic{a_j} i_1)}\]
be the events and the partition, respectively, obtained by replacing a piece of $a_{k+1}$ of length $\varepsilon$  in block $j_{k+1}$ of $\cP$ with a piece of $a_{k}$ of length $\varepsilon$ for each $k\in\{1,\ldots,j\}$, where $i_{j+1}=i_{1}$.
This transformation may not result in a feasible schedule because the precedence constraints or release dates may be violated.
However, if neither is violated, then the assumptions of Lemma~\ref{lem:xi_to_P} are met for $\cP'$ with the events $\vect{e}'$, and the partition $\vect{\xi}'$ exists. If $\cP'$ exists, then in addition we assume that the blocks of $\cP'$ enforce the following restrictions:
\begin{itemize}
 \item For each $a_k\in B$, if $\complTime{\cP}{a_k}=e_{i_k+1}$ and $i_{k+1}<i_k$ (taking $i_{j+1}=i_{1}$), then $\complTime{\cP'}{a_k}=e_{i_k+1}-\varepsilon$, which reduces the completion time of job $a_k$ by as much as possible with respect to the cyclic shift.
 \item If $\complTime{\cP}{a_k}\leq e_{i_{k+1}}$ (taking $i_{j+1}=i_{1}$), then $\complTime{\cP'}{a_k}=e_{i_{k+1}}+\varepsilon$, which increases the completion time of job $a_k$ by as little as possible with respect to the cyclic shift.
\end{itemize}
Note that, in general, $\vect{e}$ does not consist of the events of $\cP'$, and the number of events of $\cP'$ may be different than the number of events of $\cP$.

Finally, we introduce the notion of swapping of two jobs which is used in Sections~\ref{subsec:abnormality_points} and~\ref{subsec:alternating_chains} to reduce total completion time of a schedule by applying the shortest processing time (SPT) rule to two jobs that complete in consecutive blocks. Let $\cP$ be a schedule with events $\vect{e}$ and partition $\vect{\xi}$.
Let $a$ and $a'$ be two jobs such that $\complTime{\cP}{a'}=e_{\tau(a)}$, $\startTime{\cP}{a}\leq\startTime{\cP}{a'}$ and $a'$ is independent of any job in  $\jobs(\xi_{\tau(a)})$.
We define a transformation of \emph{swapping $a$ and $a'$} that results in a new schedule $\cP'$ as follows (see Figure~\ref{fig:swapping}).
   \begin{figure*}[hbt]
    \begin{center}
    \includegraphics[scale=1.0]{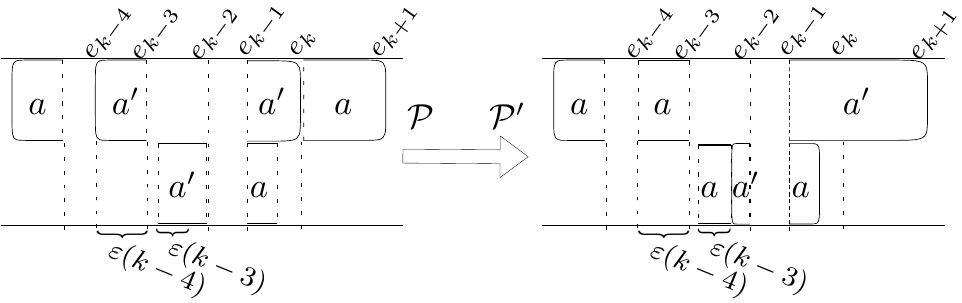}
    \caption{Swapping $a$ and $a'$, where $I=\{k-4,k-3\}$, $k=\tau(a)$, leads in this case to a schedule $\cP'$ with a smaller total completion time}
    \label{fig:swapping}
   \end{center}
    \end{figure*}
Find a set of indices $I\subseteq\{1,\ldots,\tau(a)-1\}$ such that for each $j\in I$,
\[0<\varepsilon(j)\leq\min\{e_{j+1}-e_j-\xi_j(a),\xi_j(a')\},\]
$\varepsilon(\max I)$ is minimum and $\sum_{j\in I}\varepsilon(j)=\xi_{\tau(a)}(a)$.
Such a set $I$ exists because of the constraints imposed on $a$ and $a'$.
The schedule $\cP'$ is obtained by performing the following three steps:
\begin{itemize}
 \item For each $j\in I$, remove a piece of $a'$ of length $\varepsilon(j)$ from the $j$-th block of $\cP$.
 \item Remove the piece of $a$ executing in the $\tau(a)$-th block and add a piece of $a'$ of length $\xi_{\tau(a)}(a)$ to the $\tau(a)$-th block of $\cP$.
 \item Add a piece of $a$ of length $\varepsilon(j)$ to the $j$-th block of $\cP$ for each $j\in I$.
\end{itemize}

\begin{lemma} \label{lem:swapping}
Given schedule $(\cP,\vect{e},\vect{\xi})$,
let $a,a'$ be two jobs such that $\complTime{\cP}{a'}=e_{k}$, $\startTime{\cP}{a}\leq\startTime{\cP}{a'}$ and
$a'$ is independent of any job in  $\jobs(\xi_{k})$, where $k=\tau_{\cP}(a)$.
Then, the schedule $\cP'$ obtained by swapping $a$ and $a'$ in $\cP$ is valid and
$\sum_{a''\in\jobs}\complTime{\cP'}{a''}\leq\sum_{a''\in\jobs}\complTime{\cP}{a''}$ with strict inequality when
$\startTime{\cP}{a}<\startTime{\cP}{a'}$ and $\xi_{k-1}(a)<e_k-e_{k-1}$.
\end{lemma}
\begin{proof}
The fact that $\cP'$ is valid follows directly from its construction.
Suppose that $\startTime{\cP}{a}<\startTime{\cP}{a'}$ and $\xi_{k-1}(a)<e_k-e_{k-1}$.
If $k-1\notin I$, then $\complTime{\cP'}{a}\leq e_{k-1}+\xi_{k-1}(a)<e_k$.
Otherwise, the restriction on taking $\varepsilon(\max I)=\varepsilon(k-1)$ to be minimum implies, due to $\startTime{\cP}{a}<\startTime{\cP}{a'}$, that $\xi_{k-1}(a)+\varepsilon(k-1)<e_k-e_{k-1}$ and hence $\complTime{\cP'}{a}=e_{k-1}+\xi_{k-1}(a)+\varepsilon(k-1)<e_k$.
Thus, the total completion time of $\cP'$ is strictly smaller than that of $\cP$ as required.
\end{proof}

\subsection{Properties of optimal schedules}
\label{subsec:opt_properties}

We now give some key properties of optimal schedules and describe three configurations that are forbidden in optimal
schedules. These results will be used in subsequent sections. The following lemma states that if a job $a$ completes
in the $i$-th block of an optimal schedule $\cP$, i.e., $\tau(a)=i$,
then the part of $a$ that executes in that block spans the block. Such a job $a$ is called a \emph{spanning job in block $i$}.

\begin{lemma} \label{lem:how_jobs_finish}
Given schedule $(\cP,\vect{e},\vect{\xi})$,
each job $a\in\jobs$ is a spanning job in block $\tau(a)$, i.e., $\xi_{\tau(a)}(a)=e_{\tau(a)+1}-e_{\tau(a)}$.
\end{lemma}

\begin{proof}
The proof is by contradiction. There exists $\varepsilon>0$ such that at most one job executes in
$I=[e_{\tau(a)+1}-\varepsilon,e_{\tau(a)+1}]$ on each machine in $\cP$ and $\varepsilon\leq
e_{\tau(a)+1}-e_{\tau(a)}-\xi_{\tau(a)}(a)$. Let $B$ be the set of the jobs that execute in $I$.
Clearly, $a\in B$ and $1\leq |B|\leq 2$.
There exists a job $b'\in\jobs\setminus B$ such that $\xi_{\tau(a)}(b')\neq 0$.
Indeed, otherwise $a$ could be executed in $[e_{\tau(a)},e_{\tau(a)}+\xi_{\tau(a)}(a)]$ without making any other changes
in the schedule.
Since the new schedule completes $a$ earlier (because $\xi_{\tau(a)}(a)<e_{\tau(a)+1}-e_{\tau(a)}$), this would
contradict the optimality of $\cP$. Then $\complTime{\cP}{b'} > \tau(a)$ and we can use some of the space of
$\xi_{\tau(a)}(b')$ for job $a$ to complete $a$ earlier. More formally,
define $\varepsilon'=\min\{\varepsilon,\xi_{\tau(a)}(b')\}$.
Let $e'=e_{\tau(a)+1}-\varepsilon'$ and for each job $c\in\jobs$ let
\[\xi'(c)=
\begin{cases}
 \xi_{\tau(a)}(c),                       & \textup{if } c\notin\{b'\}\cup(B\setminus\{a\}), \\
 \xi_{\tau(a)}(c)-\varepsilon',  & \textup{if } c\in\{b'\}\cup (B\setminus\{a\}),
\end{cases}
\]
and
\[\xi''(c)=
 \begin{cases}
   0,             & \textup{if } c\notin\{b'\}\cup(B\setminus\{a\}), \\
   \varepsilon',  & \textup{if } c\in\{b'\}\cup (B\setminus\{a\}).
 \end{cases}
\]
By Lemma~\ref{lem:xi_to_P}, there exists a schedule $\cP'$ such that for each $t\in\{1,\ldots,q-1\}\setminus\{\tau(a)\}$ the total
length of all pieces of each job $c\in\jobs$ executed in $[e_t,e_{t+1}]$ is $\xi_t(c)$, the total length of all pieces of each
job $c$ executed in $[e_{\tau(a)},e']$ equals $\xi'(c)$, and the total length of all pieces of each job $c$ executed in
$[e',e_{\tau(a)}+1]$ equals $\xi''(c)$.
However, $\complTime{\cP'}{c}=\complTime{\cP}{c}$ for each $c\in\jobs\setminus\{a\}$ and $\complTime{\cP'}{a}<\complTime{\cP}{a}$,
which contradicts the optimality of $\cP$.
\end{proof}

\begin{lemma} \label{lem:no_idle-time_between_preemptions}
Given schedule $(\cP,\vect{e},\vect{\xi})$, if $a\in\jobs$ is not a spanning job in block $i$ ($i\in\{1,\ldots,q-2\}$),
$\startTime{\cP}{a}\leq e_i$ and $\complTime{\cP}{a}\geq e_{i+1}$, then there is no idle time in the $i$-th block of $\cP$.
\end{lemma}

\begin{proof}
Suppose for a contradiction that there is idle time of length $\varepsilon>0$ in $[e_i,e_{i+1}]$ on one of the processors in
$\cP$. We get a contradiction by obtaining another schedule $\cP'$ such that $\complTime{\cP}{b}=\complTime{\cP'}{b}$ for each
$b\in\jobs\setminus\{a\}$ and $\complTime{\cP'}{a}<\complTime{\cP}{a}$.
Namely, take $\varepsilon'=\min\{\varepsilon,\xi_{\tau(a)}(a),e_{i+1}-e_i-\xi_i(a)\}$.
By Lemma~\ref{lem:how_jobs_finish}, $\tau(a)>i$ and hence $\varepsilon'>0$.
By Lemma~\ref{lem:xi_to_P}, the desired schedule $\cP'$ obtained from $\cP$ by moving the piece of $a$ that executes in
$[\complTime{\cP}{a}-\varepsilon',\complTime{\cP}{a}]$ to the $i$-th block of $\cP$ is valid.
\end{proof}

Given schedule $(\cP,\vect{e},\vect{\xi})$, two jobs $a$ and $b$ with $\tau(a)<\tau(b)$ are said to \emph{interlace} if
job $b$ is not spanning in block $\tau(a)$ and there exists $t<\tau(a)$ such that job $a$ is not spanning in block $t$,
$\xi_t(b)>0$, $r(a)<e_{t+1}$ and $a$ is independent of all jobs in $\jobs(\xi_t)\cup\cdots\cup\jobs(\xi_{\tau(a)})$.
Note that, informally speaking, the above constraints imply that a piece of $a$ executed in
$[\complTime{\cP}{a}-\varepsilon,\complTime{\cP}{a}]$, for some $\varepsilon>0$, can be exchanged with a piece of $b$ of
length $\varepsilon$ executing in the $t$-th block of $\cP$.
We formalize this observation in the next lemma.

\begin{lemma} \label{lem:interlace}
If $\cP$ is an optimal schedule, then no two jobs interlace in $\cP$.
\end{lemma}
\begin{proof}
Let $\vect{e}$ and $\vect{\xi}$ be the events and the partition of $\cP$, respectively.
Suppose for a contradiction that two jobs $a$ and $b$ with $\tau(a)<\tau(b)$ interlace and $t$ is the block in the definition.
Let
\begin{eqnarray*}
 \varepsilon= & \min\big\{ \xi_t(b),e_{\tau(a)+1}-e_{\tau(a)}-\xi_{\tau(a)}(b), \\
              & \xi_{\tau(a)}(a),e_{t+1}-e_t-\xi_t(a) \big\}.
\end{eqnarray*}
Note that $\varepsilon>0$.
By Lemma~\ref{lem:xi_to_P}, there exists a schedule $\cP'$ with $\vect{e}'$ and partition $\vect{\xi}'$ such that
\[
(\vect{e}',\vect{\xi}')=\cyclicshift{ \vect{e},\vect{\xi},\varepsilon,(t\cyclic{b} \tau(a)\cyclic{a} t) }.
\]
The schedule $\cP'$ is valid for two reasons.
First, $r(a)<e_{t+1}$ implies that if $\startTime{\cP}{a}<e_{t+1}$, then $r(a)\leq e_t$ and if $\startTime{\cP}{a}\geq e_{t+1}$, then $\startTime{\cP'}{a}\geq e_{t+1}-\varepsilon$ according to the definition of the transformation, which implies that $a$ does not start prior to its release date in $\cP'$.
Second, the fact that $a$ is independent of all jobs in $\jobs(\xi_t)\cup\cdots\cup\jobs(\xi_{\tau(a)})$ implies that $a$ does not violate the precedence constraints in $\cP'$.
For each $c\in\jobs\setminus\{a\}$, $\complTime{\cP}{c}=\complTime{\cP'}{c}$ and $\complTime{\cP}{a}>\complTime{\cP'}{a}$.
This contradicts the optimality of $\cP$.
\end{proof}

\begin{lemma} \label{lem:at_most_two_start}
Let $(\cP,\vect{e},\vect{\xi})$ be an optimal schedule.
Let $I=[x,y]$ be an interval and let $B\subseteq\jobs$ be such that $a\in B$ if and only if the total length of job $a$
executing in $I$ is strictly between $0$ and $y-x$.

If jobs in $B$ are independent, $\complTime{\cP}{a}\geq y$ and $r(a)\leq x$ for each $a\in B$, then $|B|\leq 2$.
\end{lemma}
\begin{proof}
It follows from definition of set $B$ that no job completes in $(x,y)$.
We first argue that
\begin{equation} \label{eq:B-complete-late}
\complTime{\cP}{b}>y \textup{ for each }b\in B.
\end{equation}
Suppose for a contradiction that $\complTime{\cP}{b}=y$ for some job $b\in B$.
Since the total length of $b$ in $I$ is less than $y-x$, there exists a non-empty interval $I'\subseteq I$ such that no part of $b$ executes in $I'$.
We obtain a schedule $\cP'$ by exchanging the part of $\cP$ that executes in $I'$ with the part of $\cP$ that executes in $[y-|I'|,y]$.
Since the release date of each job that executes in $I$ is at most $x$ and the jobs whose parts execute in $I$ are independent, we obtain that $\cP'$ is indeed a feasible  schedule.
Then, $\complTime{\cP'}{b}=y-|I'|<y=\complTime{\cP}{b}$ and $\complTime{\cP'}{a}\leq\complTime{\cP}{a}$ for each $a\in\jobs\setminus\{b\}$, which completes the proof of \eqref{eq:B-complete-late}.

We now prove the lemma.
Suppose for a contradiction that $|B|>2$.
Let $b$ be a job in $B$ with minimum completion time in $\cP$.
Since $|B|>2$, Lemma~\ref{lem:how_jobs_finish} implies that there exists $b'\in B$ such that $\tau(b)<\tau(b')$ and $b'$ is not a spanning job in chunk $\tau(b)$.
Define $\varepsilon=\min\{y-x-p,\xi_{\tau(b)}(b),p',e_{\tau(b)+1}-e_{\tau(b)}-\xi_{\tau(b)}(b')\}$, where $p$ and $p'$ are the total lengths of $b$ and $b'$ respectively executing in $I$.
Due to the choice of $b'$, $\varepsilon>0$.
We obtain a schedule $\cP'$ by first exchanging the pieces of $b'$ of total length $\varepsilon$ executing in $I$ with a piece of $b$ of length $\varepsilon$ executing in chunk $\tau(b)$.
The resulting $\cP'$ may not be feasible in $I$,
however, the McNaughton's rule can readily turn this part into a feasible schedule.
This provides a feasible schedule $\cP'$ because the release date of each job whose part executes in $I$ is at most $x$ and the jobs that execute in $I$ in $\cP$ are independent.
By \eqref{eq:B-complete-late}, $\complTime{\cP'}{b}=\complTime{\cP}{b}-\varepsilon$.
Note that if a job completes at $y$ in $\cP$, then the total length of this job in $I$ equals $y-x$; otherwise the job would belong to $B$ contradicting \eqref{eq:B-complete-late}.
Thus, no job completes later in $\cP'$ than in $\cP$ --- a contradiction with the optimality of $\cP$.
\end{proof}

\begin{lemma} \label{lem:one_dominates}
Let schedule $(\cP,\vect{e},\vect{\xi})$ be optimal.
If $\jobs(\xi_i)\neq\emptyset$ ($i\in\{1,\ldots,q-1\}$), then:
\begin{enumerate} [label={\normalfont(\roman*)},leftmargin=*]
 \item\label{it:domin:1} There exists a job in $\jobs$ that is spanning in block $i$;
 \item\label{it:domin:2} $|\jobs(\xi_i)|\leq 3$ and if $|\jobs(\xi_i)|=3$, then some job in $\jobs(\xi_i)$
 completes at $e_{i+1}$ in $\cP$.
\end{enumerate}
\end{lemma}

\begin{proof}
Note that $|\jobs(\xi_i)|>3$ would lead to a contradiction to Lemma~\ref{lem:at_most_two_start} with $I=[e_i,e_{i+1}]$.
Moreover, if $i=\tau(a)$ for some $a\in\jobs$, then by Lemma~\ref{lem:how_jobs_finish}, $a$ is spanning in block $i$ and
the lemma holds.

Thus, assume that no job finishes in the $i$-th block of $\cP$.
If $|\jobs(\xi_i)|\leq 2$, then it remain to prove \ref{it:domin:1}: if no job $a$ is spanning in block $i$, then by
Lemma~\ref{lem:no_idle-time_between_preemptions}, there is no idle time in the $i$-th block of $\cP$, which would violate
$|\jobs(\xi_i)|\leq 2$.
This completes the proof of case $|\jobs(\xi_i)|\leq 2$.
We prove, by contradiction, that $|\jobs(\xi_i)|=3$ is not possible if no job completes at $e_{i+1}$.
Denote $B=\{a\in\jobs\st 0<\xi_i(a)<e_{i+1}-e_i\}$.
Clearly, $|B|>1$.
On the other hand, $|B|<3$ for otherwise the job in $B$ with smallest completion time interlaces with one of the two other
jobs in $B$, which contradicts Lemma~\ref{lem:interlace}. Thus, $|B|=2$.
Denote $B=\{b,b'\}$ and assume without loss of generality that $\complTime{\cP}{b}\leq\complTime{\cP}{b'}$.
According to Lemma~\ref{lem:how_jobs_finish}, job $b$ is spanning in block $\tau(b)$.
Also job $b'$ is spanning in block $\tau(b)$, since otherwise $b$ and $b'$ interlace, which is not possible according to
Lemma~\ref{lem:interlace}. The only job, call it $c'$, in $\jobs(\xi_i)\setminus\{b,b'\}$ completes in $(e_{i+1},e_{\tau(b)})$
for otherwise this job and $b$ interlace --- again a contradiction with Lemma~\ref{lem:interlace}.
Thus, in particular, $\tau(b)>i+1$. This situation is depicted in Figure~\ref{fig:three-jobs}.

   \begin{figure}[htb]
    \begin{center}
    \includegraphics[scale=1.0]{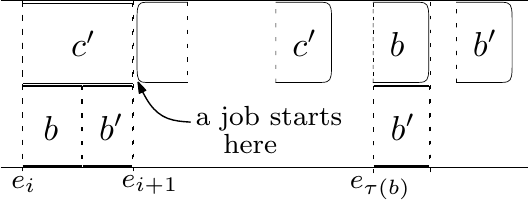}
    \caption{The proof of Lemma~\ref{lem:one_dominates}: the positioning of jobs $b,b'$ and $c'$.}
    \label{fig:three-jobs}
   \end{center}
    \end{figure}

Let $Y=\{a\in\jobs\setminus\jobs(\xi_i)\st e_{i+1}\leq\startTime{\cP}{a}\leq e_{\tau(b)}\}$.
Since $e_{i+1}$ is an event of $\cP$, $Y\neq\emptyset$. By Lemma~\ref{lem:interlace}, if $e_{i+1}\leq\complTime{\cP}{a}\leq e_{\tau(b)}$, then $a\in Y$ or $a=c'$.
If there exists $c\in Y$ such that $\complTime{\cP}{c}=e_{\tau(b)}$, then we obtain a schedule $\cP'$ by swapping $b$ and $c$.
By Lemma~\ref{lem:swapping}, $\cP'$ is feasible.
Moreover, if job $b$ is non-spanning in block $\tau(b)-1$, then the total completion time of $\cP'$ is smaller than that of $\cP$, which completes the proof.
If, on the other hand, job $b$ is spanning in block $\tau(b)-1$, then $\complTime{\cP'}{b}=e_{\tau(b)}$ and
$\xi'_{\tau_{\cP'}(b)-1}(b')=0$, where $\vect{\xi}'$ is the partition of $\cP'$, in which case $b$ and $b'$ interlace in $\cP'$ --- a contradiction with Lemma~\ref{lem:interlace}.
Thus, it remains to consider the situation when no such $c$ exists.
This, since $e_{\tau(b)}$ is an event of $\cP$, implies that $c'$ ends at $e_{\tau(b)}$ in $\cP$. Moreover, $\jobs(\xi_{\tau(c')})\subseteq\{c',b,b'\}$ for otherwise $\cP$ would not be optimal.
Thus, some job $c\in Y$ ends at $e_{\tau(c')}$ because $e_{\tau(c')}$ is an event of $\cP$ and no job in $Y$ can start at $e_{\tau(c')}$.
Therefore, one of jobs $\{c',b\}$ must be non-spanning in block $\tau(c)$.
Swapping this job with $c$ gives, by Lemma~\ref{lem:swapping}, a schedule with smaller total completion time that that of $\cP$, which provides the required contradiction and completes the proof of the lemma.
\end{proof}

The following two lemmas describe additional configurations that cannot be present in an optimal schedule.
The first situation is depicted in Figure~\ref{fig:forbidden-situations}(a), while the statement of Lemma~\ref{lem:between_preemptions} is shown in Figure~\ref{fig:forbidden-situations}(b).

   \begin{figure}[htb]
    \begin{center}
    \includegraphics[scale=1]{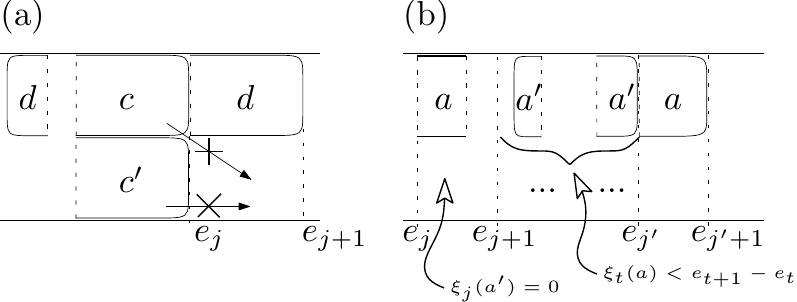}
    \caption{(a) the illustration of Lemma~\ref{lem:ccd}; (b) the illustration of Lemma~\ref{lem:between_preemptions}}
    \label{fig:forbidden-situations}
   \end{center}
    \end{figure}

\begin{lemma}\label{lem:ccd}
Given schedule $(\cP,\vect{e},\vect{\xi})$, let $e_j$ be an event in $\cP$ and jobs $c$, $c'$, and $d$ be such that
\begin{enumerate} [label={\normalfont(\roman*)},leftmargin=*]
 \item\label{it:c1} $\complTime{\cP}{c}=\complTime{\cP}{c'}=e_j$;
 \item\label{it:c2} $\complTime{\cP}{d}=e_{j+1}$ and $\startTime{\cP}{d}<e_j$;
 \item\label{it:c3} Jobs in $\{c,c'\}\cup \jobs(\xi_j)$ are independent.
\end{enumerate}
Then, $\cP$ is not optimal.
\end{lemma}

\begin{proof}
By Lemma~\ref{lem:one_dominates} (ii), one of jobs $c$ or $c'$, say $c$, satisfies $\xi_t(c)=0$, where $e_t=\startTime{\cP}{ d}$.
We then have $\startTime{\cP}{ d}<\startTime{\cP}{c}$ for otherwise $c$ and $d$ would interlace, thus we get a contradiction by
Lemma~\ref{lem:interlace}. Therefore, we can swap jobs $c$ and $d$.
By Lemma~\ref{lem:swapping}, the resulting schedule is feasible and has smaller total completion time than $\cP$, as required.
\end{proof}

\begin{lemma} \label{lem:between_preemptions}
Let schedule $(\cP,\vect{e},\vect{\xi})$ be optimal and jobs $a$ and $a'$ be such that
\begin{itemize}
 \item $\xi_j(a)>0$, $\xi_{j'}(a)>0$, $j<j'-1$ and job $a$ is spanning in block $t$ for each $t\in\{j+1,\ldots,j'-1\}$;
 \item $\xi_{j}(a')=0$ and $\complTime{\cP}{a'}=e_{j'}$;
 \item No successor of $a'$ starts at $e_{j'}$.
\end{itemize}
Then, $\startTime{\cP}{a'}\geq e_{j+1}$ and $\tau(a)>j'$.
\end{lemma}

\begin{proof}
If $\startTime{\cP}{a'}<e_{j+1}$, then due to $\xi_j(a')=0$, $\startTime{\cP}{a'}<e_j$.
But then, $a$ and $a'$ would interlace, which is not possible in an optimal schedule according to Lemma~\ref{lem:interlace}.

By assumption, $a'$ is independent of any job in $\jobs(\xi_{j'})$.
Also, $\startTime{\cP}{a}\leq e_{j+1}\leq\startTime{\cP}{a'}$.
Then, $\tau(a)>j'$ follows from an observation that otherwise swapping $a$ and $a'$ in $\cP$ would produce, by Lemma~\ref{lem:swapping}, a schedule with smaller total completion time than that of $\cP$.
\end{proof}

\subsection{Abnormality points and maximal schedules}
\label{subsec:abnormality_points}

We now define normal schedules, abnormality points and maximal schedules. In particular Lemma~\ref{lem:at_least_two_abnormal}
gives key necessary conditions for an abnormality point.

For any $x\in\reals_+$ and nonnegative integer $l$, we say that $x$ is \emph{$l$-normal} if $x=l'/2^l$ for some integer $l'$.
We say that a block of a schedule $\cP$ is \emph{$l$-normal} if the length of the block is $l$-normal and the total execution time
of each job in the block is $(l+1)$-normal. A preemptive schedule $\cP$ with $q$ events is \emph{normal} if the $i$-th block of
$\cP$ is $i$-normal for each $i\in\{1,\ldots,q-1\}$.
If a schedule $\cP$ with $q$ events $\vect{e}$ and partition $\vect{\xi}$ is not normal, then the minimum index
$i\in\{1,\ldots,q-1\}$ such that the $i$-th block of $\cP$ is not $i$-normal is called the \emph{abnormality point of} $\cP$.
If a schedule is normal, then its abnormality point is denoted by $\infty$ for convenience. We have the following simple
observations.

\begin{observation} \label{obs:normal_numbers}
If $x$ is $l$-normal, then $x$ is $l'$-normal for each $l'\geq l$. \qed
\end{observation}

\begin{observation} \label{obs:abnormality_point}
If $i\neq\infty$ is the abnormality point of a schedule $\cP$ with events $\vect{e}$, then $e_i$ is $(i-1)$-normal.
\qed
\end{observation}

According to our definition, if an $i$-th block of a schedule $\cP$ is $i$-normal, then $\xi_i(a)$ is $(i+1)$-normal for each
$a\in\jobs$, however, this does not necessarily imply that job preemptions occur at $(i+1)$-normal time points in the $i$-th block
of $\cP$. Such job preemptions can possibly take place only strictly between $e_i$ and $e_{i+1}$ since both $e_i$ and $e_{i+1}$
are $i$-normal by assumption. By the next observation, we may assume without loss of generality that $i$-normal blocks have job
preemptions only at $(i+1)$-normal time points.

\begin{observation} \label{obs:normal_schedule}
If the $i$-th block of a schedule $\cP$ is $i$-normal, then there exists a schedule $\cP'$ with the same events, partition and total completion time as that of $\cP$, in which each preemption, resumption, job start and job completion in the $i$-th block occurs at $(i+1)$-normal time point.
\end{observation}

\begin{proof}
It follows from the McNaughton's algorithm.
\end{proof}

Let us introduce a partial order, denoted by $\trianglelefteq$, to the set of all schedules. For schedules $\cP$ and $\cP'$,
we write $\cP\trianglelefteq\cP'$ if and only if one of the following holds:
\begin{itemize}
 \item $\cP = \cP'$;
 \item $\cP'$ is optimal, while $\cP$ is not;
 \item Both $\cP$ and $\cP'$ are optimal and, additionally, $\cP'$ is normal while $\cP$ is not;
 \item Both $\cP$ and $\cP'$ are optimal, but neither is normal. Additionally $i\leq i'$, where $i$ and $i'$ are the
       abnormality points of $\cP$ and $\cP'$, respectively.
\end{itemize}
Any element in $\jobs$ that is maximal under the partial order is called a \emph{maximal} schedule.

\begin{lemma} \label{lem:remaining_part_i-normal}
Let schedule $(\cP,\vect{e},\vect{\xi})$ have abnormality point $i\neq\infty$.
For each $a\in\jobs$ and for each $i'\leq i$, $\sum_{j=i'}^{q-1}\xi_j(a)$ is $i'$-normal.
\end{lemma}

\begin{proof}
Let $a\in\jobs$ be selected arbitrarily. Note that
\[
\sum_{j=i'}^{q-1}\xi_j(a)=1-\sum_{j=1}^{i'-1}\xi_j(a).
\]
Since $i\geq i'$ is the abnormality point of $\cP$, Observation~\ref{obs:normal_numbers} implies that $\xi_j(a)$ is $i$-normal
for each $j\in\{1,\ldots,i'-1\}$.
\end{proof}

The next lemma, informally speaking, allows us to further consider only those maximal schedules with abnormality point
$i\neq\infty$ in which the abnormality of the $i$-th block is due to the length of the jobs in this block, and not due to the
length of this block.

\begin{lemma} \label{lem:all_e_i's_are_normal}
Let $\cP$ be a maximal schedule with the events $e_1,\ldots,e_q$.
If $i\neq\infty$ is the abnormality point of $\cP$, then there exists a maximal schedule $\cP'$ with abnormality point $i$ such that $e_1,\ldots,e_i,e_{i+1}',\ldots,e_{q'}'$ are its events and $e_{i+1}'-e_i$ is $i$-normal.
\end{lemma}

\begin{proof}
If $i=\tau(a)$ for some $a\in\jobs$, then by Lemma~\ref{lem:how_jobs_finish}, $e_{i+1}=e_i+\xi_i(a)$.
By Lemma~\ref{lem:remaining_part_i-normal}, $\xi_i(a)$ is $i$-normal. Thus, $\cP'=\cP$ is the required schedule,
which proves the lemma.
Hence, $i\neq\tau(a)$ for each $a\in\jobs$, i.e., no job completes at $e_{i+1}$.
Note that $e_{i+1}=\startTime{\cP}{a}$ for some $a\in\jobs$.
Suppose for a contradiction that $e_{i+1}-e_i$ is not $i$-normal.
Thus, $e_{i+1}$ is not $i$-normal.
Lemma~\ref{lem:no_idle-time_between_preemptions} implies that there is no idle time in the $i$-th block of $\cP$.
Thus, $|\jobs(\xi_i)|\geq 2$ and therefore there exists $d\in\jobs(\xi_i)$ that is non-spanning in block $i+1$ because
$a$ starts at $e_{i+1}$.

\paragraph{Case 1:}  There is an $i$-normal number in $(e_i,e_{i+1}]$.
  Let $x$ be the maximal $i$-normal number in $(e_i,e_{i+1}]$.
  Then, $r(a)\leq x$ because there is no $i$-normal number in $(x,e_{i+1}]$ and $r(a)\leq e_{i+1}$ is $i$-normal by
  Observation~\ref{obs:normal_numbers}. Let
  \[
  0<\varepsilon\leq\min\{\xi_{i+1}(a),e_{i+1}-x,\xi_i(d),e_{i+2}-e_{i+1}-\xi_{i+1}(d)\}.
  \]
  No job completes at $e_{i+1}$ and therefore the jobs in $\jobs(\xi_i)\cup\jobs(\xi_{i+1})$ are independent.
  Thus, by Lemma~\ref{lem:xi_to_P}, there exists a schedule $\cP'$ with events $\vect{e}'$ and partition $\vect{\xi}'$, where
  \[
  (\vect{e}',\vect{\xi}')=\cyclicshift{ \vect{e},\vect{\xi},\varepsilon,(i\cyclic{d} i+1\cyclic{a} i) }.
  \]
  Moreover, due to the McNaughton's rule, one can assume that $\startTime{\cP'}{a}=x$.
  By Observation~\ref{obs:abnormality_point}, $e_i$ is $(i-1)$-normal and hence, by Observation~\ref{obs:normal_numbers},
  $x-e_i$ is $i$-normal.   Since the first $i-1$ blocks are identical in $\cP$ and $\cP'$, and
  $e'_{i+1}=x$, $\cP'$ is the desired schedule, which completes the proof in this case.

\paragraph{Case 2:} There is no $i$-normal number in $(e_i,e_{i+1}]$.
  By Observation~\ref{obs:normal_numbers}, there is no $(i-1)$-normal number in $(e_i,e_{i+1}]$.
  Let $x>e_i$ be the minimum $(i-1)$-normal number.
  Since $i+1<q$, more than one block intersects $(e_i,x)$.

  Suppose first that exactly two blocks intersect $(e_i,x)$, and there exists a job $b$ such that $\xi_i(b)+\xi_{i+1}(b)=e_{i+2}-e_i$.
  One of the two blocks is of length at least $(x-e_i)/2$.
  By Observation~\ref{obs:abnormality_point}, $e_i+(x-e_i)/2$ is $i$-normal.
  Hence, due to the condition in Case~2, this must be the $(i+1)$-st block.
  However, the schedule with events $(e_1,\ldots,e_{i},e_i+e_{i+2}-e_{i+1},e_{i+2},\ldots,e_q)$ and partition
  $(\xi_1,\ldots,\xi_i,\xi_{i+2},\xi_{i+1},\xi_{i+3},\ldots,\xi_{q-1})$ would satisfy the assumption in Case~1.
  This allows us to construct the desired schedule $\cP'$ as in Case~1.

  Suppose now that exactly two blocks intersect $(e_i,x)$ and there exists no job $b$ such that $\xi_i(b)+\xi_{i+1}(b)=e_{i+2}-e_i$.
  Lemma~\ref{lem:at_most_two_start} applied to $I=[e_i,e_{i+2}]$ gives a contradiction.
  Observe that the corresponding set $\{a,b,d\}\subseteq B$ in Lemma~\ref{lem:at_most_two_start} is of size at least $3$.
  Moreover, since no job completes in $(e_i,e_{i+2})$, $B$ contains only independent jobs.

  Finally, suppose that more than two blocks intersect $(e_i,x)$. Thus, the job $a$ does not complete before $x$.
  Moreover, no job completes at $e_{i+2}$ because otherwise either $\cP$ is not optimal or $e_{i+2}$ is $i$-normal by
  Lemma~\ref{lem:remaining_part_i-normal}.
  Since $e_{i+2}<x$ we get a contradiction in either case.
  Therefore, there is a job $a'$ that starts at $e_{i+2}$.
  Clearly, $a'$ does not complete before $x$.
  Thus, Lemma~\ref{lem:at_most_two_start} for $I=[e_i,\min\{x,e_{i+3}\}]$ again gives a contradiction.
  Observe that the corresponding set $\{a,a',d\}\subseteq B$ in Lemma~\ref{lem:at_most_two_start} is of size at least $3$.
  Moreover, since no job completes in $(e_i,e_{i+3})$, $B$ contains only independent jobs.
\end{proof}

Given schedule $(\cP, \vect{e}, \vect{\xi})$, for $i\in\{1,\ldots,q-1\}$ define
\[
A_i(\cP)=\left\{a\in\jobs \st \xi_i(a) \textup{ is not $(i+1)$-normal}\right\}.
\]

\begin{lemma} \label{lem:at_least_two_abnormal}
Let $\cP$ be a maximal schedule.
If $i\neq\infty$ is the abnormality point of $\cP$, then $|A_i(\cP)|=2$ and $|\jobs(\xi_i)|= 3$.
\end{lemma}

\begin{proof}
Let $\vect{e}$ and $\vect{\xi}$ be the events and the partition of $\cP$, respectively.
By Lemma~\ref{lem:all_e_i's_are_normal}, $e_{i+1}-e_i$ is $i$-normal.
We have $|\jobs(\xi_i)|>2$ because otherwise by Lemmas~\ref{lem:no_idle-time_between_preemptions} and \ref{lem:how_jobs_finish},
$\xi_i(a)=e_{i+1}-e_i$ for each $a\in\jobs(\xi_i)$, which would contradict the fact that $e_{i+1}-e_i$ is $i$-normal.
Lemma~\ref{lem:one_dominates} implies that $|\jobs(\xi_i)|=3$ and there exists $a\in\jobs(\xi_i)$ such that $\xi_i(a)=e_{i+1}-e_i$.
Thus, $a\notin A_i(\cP)$, and we have that $|A_i(\cP)|\leq 2$ because $A_i(\cP)\subseteq\jobs(\xi_i)$.
Also, $|A_i(\cP)|>1$ by Lemma~\ref{lem:no_idle-time_between_preemptions}.
\end{proof}

We finish this section with a remark that directly follows from the definition of normality.
The remark allows us to conclude that the abnormality point of a schedule does not decrease after a certain type of schedule modifications.

\begin{lemma} \label{lem:abnormality_preserved}
Let $\cP$ be a schedule and $\varepsilon=l'/2^{l}$ for some integers $l$ and $l'$.
If $\cP'$ is a schedule that is obtained from $\cP$ by a sequence of modifications, each modification being a removal of a piece
of length that is a multiple of $\varepsilon$ from a $j$-th block and an insertion of this piece into a $j'$-th block, where
$\min\{j,j'\}\geq l-1$, then the abnormality point of $\cP'$ is not smaller than that of $\cP$. \qed
\end{lemma}

\section{$\A$-configurations}
\label{subsec:A-configurations}

In this section we first define a particular structure that may appear in a schedule; we refer to this structure as an \emph{$\A$-configuration}.
Our proof that there exists a normal optimal schedule for each $\jobs$, given in Section~\ref{subsec:alternating_chains}, relies on the key assumption that there exists a maximal schedule without $\A$-configurations, or \emph{$\A$-free} maximal schedule, for each set of jobs $\jobs$. Therefore,
the main goal of this section is to prove that an $\A$-free maximal schedule exists for each $\jobs$. Our proof is by contradiction: informally speaking, we take a maximal schedule having an $\A$-configuration as early as possible, and, after some schedule transformations, we either obtain a new schedule with smaller total completion time or with an earlier $\A$-configuration. In the former case we clearly obtain a contradiction.
In the latter case, a contradiction occurs only if the new schedule is maximal, i.e., its abnormality point is not smaller than that of the initial schedule.
For this reason, while performing the initial schedule transformations we must ensure that they do not change the abnormality point in the latter case.
The proof works for in-trees, however, it does not for general precedence constraints. The question whether there is an $\A$-free maximal schedule for each $\jobs$ and general precedence constraints remains open.

Let $(\cP,\vect{e},\vect{\xi})$ be a schedule. We say that $\cP$ has an \emph{$\A$-configuration of length $\ell$ ($\ell>0$)
starting at $e_j$} if there exist two jobs $a$ and $b$ such that
\begin{itemize}
\item $\complTime{\cP}{a}=e_{j}$ and $\complTime{\cP}{b}=e_{j'}$ for some $j'>j$;
\item $[e_{j}-\ell, e_{j}]$ is a maximal interval where $a$ executes non-preemptively;
\item $b$ executes non-preemptively in $[e_j,e_{j'}]$, and $b$ does not execute in $[e_{j}-\ell, e_{j}]$;
\item $\startTime{\cP}{b}<e_{j}-\ell$;
\item $a$ and each job in $\jobs(\xi_j)\cup\cdots\cup\jobs(\xi_{j'})$ are independent.
\end{itemize}
We also say that the jobs $a$ and $b$ \emph{form} the $\A$-configuration.
See Figure~\ref{fig:def-A-conf} for an exemplary $\A$-configuration.

   \begin{figure}[htb]
    \begin{center}
    \includegraphics[scale=1.0]{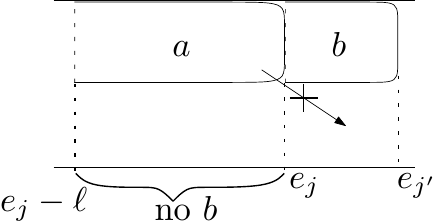}
    \caption{an example of $\A$-configuration}
    \label{fig:def-A-conf}
   \end{center}
    \end{figure}

If no pair of jobs form an $\A$-configuration in $\cP$, then $\cP$ is called \emph{$\A$-free}.
For any time interval $I$, if for any $x\in I\cap\{e_i:i=1,\ldots,q\}$ there is no $\A$-configuration
at $x$ in $\cP$, then $\cP$ is \emph{$\A$-free in $I$}. The main result of this section is the following
proposition.

\begin{proposition} \label{pro:no-A-configurations}
If any maximal schedule has an abnormality point $i\neq\infty$, then there exists an $\A$-free maximal schedule.
\end{proposition}

We first provide several technical lemmas before presenting our proof of Proposition~\ref{pro:no-A-configurations}.
A schedule $\cP$ of abnormality point $i$ is said to be \emph{$\A$-maximal} if it is maximal and, unless $i=\infty$, one of
the following two statements is true:
\begin{itemize}
 \item $\cP$ is $\A$-free;
 \item any maximal schedule is not $\A$-free, and $\cP$ has the earliest starting $\A$-configuration among maximal schedules.
\end{itemize}

\begin{lemma} \label{lem:start-of-a-and-b}
Let $\cP$ be $\A$-maximal. If $a$ and $b$ form an $\A$-configuration in $\cP$ with $\complTime{\cP}{a}<\complTime{\cP}{b}$,
then $\startTime{\cP}{a}\leq\startTime{\cP}{b}$.
\end{lemma}

\begin{proof}
Suppose for a contradiction that $\startTime{\cP}{a}>\startTime{\cP}{b}$. Then, swapping jobs $a$ and $b$ in $\cP$ produces, by Lemma~\ref{lem:swapping}, a schedule with smaller total completion time than that of $\cP$--- a contradiction.
\end{proof}

The first of the following two lemmas describes a situation that guarantees an $\A$-configuration, while the second
a situation that cannot happen in an $\A$-maximal schedule with an $\A$-configuration.

\begin{lemma}\label{lem:acc1d}
Given maximal schedule $(\cP,\vect{e},\vect{\xi})$, let $e_j$ be an event in $\cP$ and jobs $a$, $c$, $c_1$, and $d$
be such that
\begin{enumerate} [label={\normalfont(\roman*)},leftmargin=*]
 \item\label{it:a1} $\complTime{\cP}{c_1}=e_{j}$, $\complTime{\cP}{c}=e_{j+1}$, and $\complTime{\cP}{d}=e_{j+2}$;
 \item\label{it:a2}  $\jobs(\xi_{j-1})=\{c,c_1\}$, $\jobs(\xi_{j})=\{a,c\}$ and $\jobs(\xi_{j+1})=\{a,d\}$;
 \item\label{it:a3} $\startTime{\cP}{d}<e_{j-1}$.
\end{enumerate}
Then, jobs $c$ and $d$ form an $\A$-configuration. (See Figure~\ref{fig:A-conf-lemmas}(a) for an illustration.)
\end{lemma}

   \begin{figure}[htb]
    \begin{center}
    \includegraphics[scale=1.0]{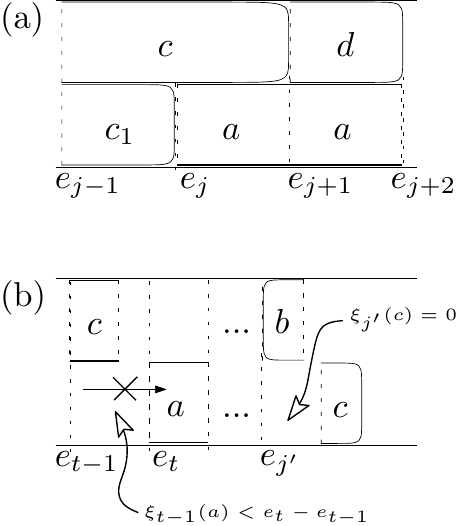}
    \caption{(a) Illustration of Lemma~\ref{lem:acc1d}. (b) Illustration of Lemma~\ref{lem:advance:a}}
    \label{fig:A-conf-lemmas}
   \end{center}
    \end{figure}

\begin{proof}
Let $k<j-1$ be the maximum index such that in block $k-1$ job $c$ is non-spanning but spanning in block $t$ for each
$t\in\{k,\ldots,j-1\}$.
Note that by \ref{it:a1}, \ref{it:a2} and Lemma~\ref{lem:no_idle-time_between_preemptions}, $k$ is well defined.
We prove, by induction on $t\in\{1,\ldots,j-k\}$, that
\begin{equation} \label{eq:i1}
  \begin{split}
\xi_{j-t}(c)=\xi_{j-t}(c_t)=e_{j-t+1}-e_{j-t}\ \land\ \complTime{\cP}{c_t}=e_{j-t+1} \\\textup{ for some }
c_t\in\jobs\setminus\{a,d\} \ \land\ \startTime{\cP}{c}<e_{j-t},
  \end{split}
\end{equation}
which immediately follows from \ref{it:a1}, \ref{it:a2} and Lemma~\ref{lem:no_idle-time_between_preemptions} for $t=1$.
So assume inductively that the claim holds for some $t-1\geq 1$ and we prove it for $t$.
It suffices to argue that some job $c_t$ completes at $e_{j-t+1}$ since then Lemma~\ref{lem:how_jobs_finish} implies
$(c_t)$ is spanning in block $j-t$.
By the induction hypothesis and the fact that all jobs have the same execution time, neither $c_{t-1}$ nor $c$ start at $e_{j-t+1}$.
Since $e_{j-t+1}$ is an event of $\cP$, some job $c_{t}$ completes at $e_{j-t+1}$ as required.
We have $\startTime{\cP}{c}<e_{j-t}$ for otherwise we can swap jobs $c$ and $d$ in $[\startTime{\cP}{c},\complTime{\cP}{d}]$.
The resulting schedule is feasible and has smaller total completion time than $\cP$.
Thus $\cP$ is not optimal --- contradiction.
This proves (\ref{eq:i1}).

If $0<\xi_{k-1}(c)<e_{k}-e_{k-1}$, then by modifying the schedule in block $k-1$ we may without loss of generality assume that $c$ resumes at $e_k$.
Thus, (\ref{eq:i1}) implies that $c$ and $d$ form an $\A$-configuration of length $e_{j+1}-e_k$ at $e_{j+1}$.
\end{proof}

\begin{lemma} \label{lem:advance:a}
Let $(\cP,\vect{e},\vect{\xi})$ be an $\A$-maximal schedule. Suppose that jobs $a$ and $b$ form an $\A$-configuration at $e_j$
in $\cP$ with $\complTime{\cP}{a}<\complTime{\cP}{b}$. Then there exists no $e_t\leq\startTime{\cP}{b}$ such that
(see Figure~\ref{fig:A-conf-lemmas}(b) for an illustration, where it is possible
that job $c$ ends at the start of $b$):
\begin{enumerate} [label={\normalfont(\roman*)},leftmargin=*]
 \item\label{it:advance1}\label{it:advance:first} $r(a)<e_t$, $\xi_t(a)=e_{t+1}-e_t$, job $a$ is non-spanning in block $t-1$;
 \item\label{it:advance2} Jobs in $\jobs(\xi_{t-1})\cup \{a\}$ are independent;
 \item\label{it:advance3} Some job $c$ in $\jobs(\xi_{t-1})$ satisfies $\complTime{\cP}{c}\geq e_{j'}$, $\xi_{j'}(c)=0$, and if
      $\complTime{\cP}{c}=e_{j'}$, then the jobs in $(\jobs(\xi_{j'})\setminus\{b\})\cup\{c\}$ are independent, where
      $e_{j'}=\startTime{\cP}{b}$;
 \item\label{it:advance4}\label{it:advance:last} The abnormality point of $\cP$ is not in $\{t,\ldots,j'\}$.
\end{enumerate}
\end{lemma}

\begin{proof}
Suppose for a contradiction that such an $e_t$ exists.
Let $\ell>0$ be the length of the $\A$-configuration formed by $a$ and $b$.
Define
\[\varepsilon=\min\left\{\xi_{t-1}(c),e_t-r(a),e_{t}-e_{t-1}-\xi_{t-1}(a),\xi_{j'}(b),\ell\right\}.\]
By \ref{it:advance1} and \ref{it:advance3}, we have $\varepsilon>0$.
Let $\cP'$ be a schedule obtained by moving a piece of $c$ of length $\varepsilon$ from the $(t-1)$-st block to the $j'$-th block, a piece of $b$ of length $\varepsilon$ from the $j'$-th block to $[\complTime{\cP}{a}-\varepsilon,\complTime{\cP}{a}]$, and a piece of $a$ from $[\complTime{\cP}{a}-\varepsilon,\complTime{\cP}{a}]$ to the $(t-1)$-st block.
By \ref{it:advance1}, \ref{it:advance2}, \ref{it:advance3} and Lemma~\ref{lem:xi_to_P}, the schedule $\cP'$ is feasible.
This transformation is shown in Figure~\ref{fig:A-conf-contr} when $\varepsilon=e_{t}-e_{t-1}-\xi_{t-1}(a)$ and $j'=t$.
Clearly, $\complTime{\cP}{d}=\complTime{\cP'}{d}$ for each $d\in\jobs\setminus\{a,c\}$ and, by \ref{it:advance3}, $\complTime{\cP'}{c}\leq\complTime{\cP}{c}+\varepsilon$.

    \begin{figure}[htb]
    \begin{center}
    \includegraphics[scale=1.0]{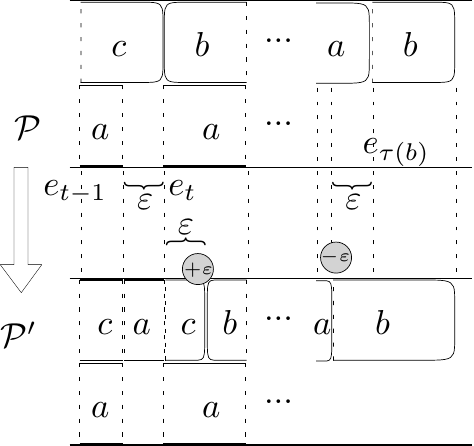}
    \caption{Schedule transformation in the proof of Lemma~\ref{lem:advance:a}}
    \label{fig:A-conf-contr}
    \end{center}
    \end{figure}

If $\varepsilon=\ell$, then the total completion time of $\cP'$ is strictly smaller than that of $\cP$, because $a$ resumes at $\complTime{\cP}{a}-\ell$ in $\cP$, i.e., $\complTime{\cP'}{a}<\complTime{\cP}{a}-\ell$.
We get a contradiction since $\cP$ is optimal.

Otherwise, if $\varepsilon<\ell$, then $a$ and $b$ form an $\A$-configuration in $\cP'$ at $\complTime{\cP}{a}-\varepsilon$.
Also, $\complTime{\cP'}{a}=\complTime{\cP}{a}-\varepsilon$.
Let $i$ be the abnormality point of $\cP$.
If $i\leq t-1$, then the fact that $\cP$ and $\cP'$ are the same in $[0,e_{t-1}]$, we obtain that the abnormality point of $\cP'$ equals $i$ and $\cP'$ is $\A$-maximal.
If $i>t-1$, then by \ref{it:advance4}, $i>j'$ and hence $\varepsilon$ is $t$-normal and, by Lemma~\ref{lem:abnormality_preserved}, $\cP'$ is $\A$-maximal.
Therefore, we obtain a contradiction in both cases, which proves the lemma.
\end{proof}

The following lemma describes how two jobs that form an $\A$-configuration start in an $\A$-maximal schedule. See
Figure~\ref{fig:A-start} for an illustration of the two possible cases: the two jobs can start at different time points,
or at the same time.

    \begin{figure}[htb]
    \begin{center}
    \includegraphics[scale=1.0]{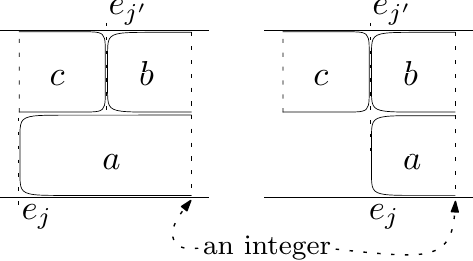}
    \caption{An illustration of Lemma~\ref{lem:at-start-of-b}}
    \label{fig:A-start}
    \end{center}
    \end{figure}

\begin{lemma} \label{lem:at-start-of-b}
Suppose that each $\A$-maximal schedule has an $\A$-configuration. There exists an $\A$-maximal schedule $(\cP,\vect{e},\vect{\xi})$
such that the earliest $\A$-configuration formed by $a$ and $b$ with $\complTime{\cP}{a}<\complTime{\cP}{b}$ satisfies the following
properties:
\begin{enumerate} [label={\normalfont(\roman*)},leftmargin=*]
 \item\label{it:start1} $\jobs(\xi_{j'})=\{a,b\}$, $|\jobs(\xi_j)|=2$ and some job completes at $e_{j'}$, where $e_{j'}=\startTime{\cP}{b}$ and $e_j=\startTime{\cP}{a}$,
 \item\label{it:start2} $0\leq j'-j\leq 1$, and
 \item\label{it:start3} $e_{j'+1}$ is an integer.
\end{enumerate}
\end{lemma}

\begin{proof}
Let $\cP$ be $\A$-maximal. Let $a$ and $b$ form the earliest $\A$-configuration in $\cP$. By Lemma~\ref{lem:start-of-a-and-b}, $\startTime{\cP}{a}\leq\startTime{\cP}{b}=e_{j'}$.
Without loss of generality assume that $\startTime{\cP}{b}$ is as late as possible.

Then job $a$ is spanning in block $j'$ since otherwise $a$ and $b$ interlace and we get a contradiction by
Lemma~\ref{lem:interlace}. Moreover,
\begin{equation} \label{eq:a-with-b}
\jobs(\xi_{j'})=\{a,b\},
\end{equation}
since otherwise, by Lemma~\ref{lem:one_dominates}, $\jobs(\xi_{j'})=\{a,b,c\}$ and $\complTime{\cP}{c}=e_{j'+1}$. The latter implies, by Lemma~\ref{lem:how_jobs_finish}, that job $c$ is spanning in block $j'$, which contradicts that job $a$ is
spanning in block $j'$ and proves (\ref{eq:a-with-b}).

We prove \ref{it:start3} first.
Suppose for a contradiction that $e_{j'+1}$ is not an integer and let $h$ be the greatest integer smaller than $e_{j'+1}$.
Since $e_{j'+1}$ is an event and, by definition of $\A$-configuration, none of the jobs $a$ and $b$ ends at $e_{j'+1}$, \eqref{eq:a-with-b} implies that some job $c$ starts at $e_{j'+1}$.

We show that $\xi_{j'+1}(b)=0$, which will make our first transformation in (\ref{eq:tr_1}) feasible. This holds for $j'+1=\tau(a)$, since by definition of $\A$-configuration $b\notin \jobs(\xi_{\tau(a)})$.
For $j'+1<\tau(a)$, we have $\xi_{j'+1}(b)=0$ or job $a$ is spanning in block $j'+1$ for otherwise $a$ and $b$ interlace and we get a contradiction by Lemma~\ref{lem:interlace}. However, $\xi_{j'+1}(b)>0$ and job $a$ is spanning in block $j'+1$, which imply
$\jobs(\xi_{j'+1})=\{a,b,c\}$.
Thus, by Lemma~\ref{lem:one_dominates}, some job must finish at $e_{j'+1}$ and since $j'+1<\tau(b)<\tau(b)$, this job must be $c$.
By Lemma~\ref{lem:how_jobs_finish}, job $c$ is spanning in block $j'+1$ --- a contradiction. Therefore,
\begin{equation}\label{eq:no-b}
\xi_{j'+1}(b)=0.
\end{equation}

Since $\tau(a)>j'+1$, we obtain by \eqref{eq:a-with-b} and definition of $\A$-configuration that no job ends at $e_{j'+2}$ and hence Lemma~\ref{lem:one_dominates} implies that job $c$ is spanning in block $j'+1$. Now take
\[
\varepsilon=\min\left\{\xi_{j'+1}(c),e_{j'+1}-\max\{h,e_{j'}\}\right\}
\]
and let $\cP'$ be a schedule with events $\vect{e'}$ and partition $\vect{\xi'}$, where
\begin{equation}\label{eq:tr_1}
(\vect{e'},\vect{\xi'})=\cyclicshift{ \vect{e},\vect{\xi},\varepsilon, (j' \cyclic{b} j'+1\cyclic{c} j') }.
\end{equation}

Figure~\ref{fig:A-conf-init}(a) illustrates the transformation from $\cP$ to $\cP'$ for $\varepsilon=e_{j'+1}-h$, when $h>e_{j'}$.
Observe that \eqref{eq:no-b} and $\xi_{j'}(b)\geq\varepsilon$ ensure the feasibility of $\cP'$.
Also, by \eqref{eq:a-with-b} and Lemma \ref{lem:at_least_two_abnormal}, we have $i\neq j'$, where $i$ is the abnormality point of $\cP$ possibly equal $\infty$.
Clearly, if $i<j'$, then $\cP$ and $\cP'$ have the same abnormality point $i$ since the two schedules are identical in $[0,e_{j'}]$.
If $i>j'$, then by Lemma~\ref{lem:all_e_i's_are_normal}, $\xi_{j'+1}(c)$ is $(j'+1)$-normal and Lemma~\ref{lem:abnormality_preserved} implies the same abnormality point $i$ for both $\cP$ and $\cP'$.
Finally, the $\A$-maximality of $\cP$ implies that $a$ and $b$ form an $\A$-configuration in $\cP'$.
Therefore, if $h\leq e_{j'}$, then $\cP'$ is $\A$-maximal and it can be ensured that $\startTime{\cP'}{b}>\startTime{\cP}{b}$, which contradicts our assumption about $\cP$.
If $h>e_{j'}$, then $\cP'$ is $\A$-maximal and satisfies \ref{it:start3} as required. To simplify notation we set $\cP:=\cP'$ in the reminder of the proof.

    \begin{figure*}[htb]
    \begin{center}
    \includegraphics[scale=1.0]{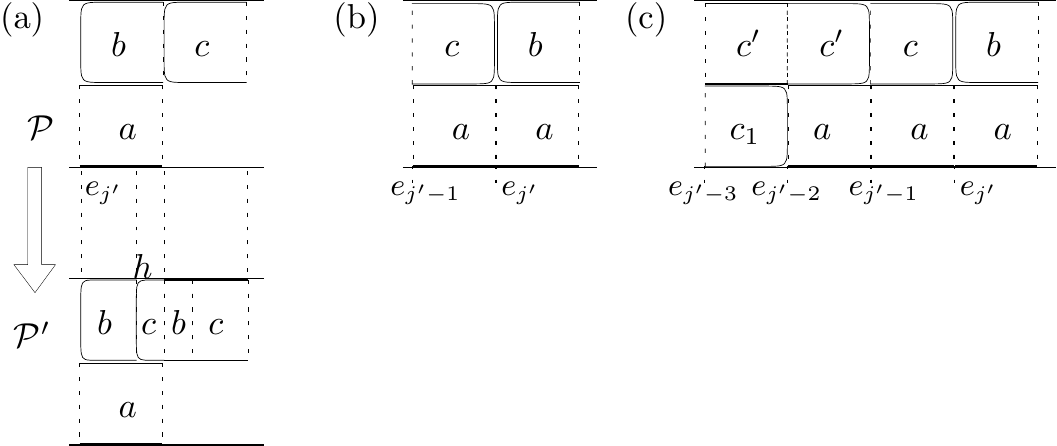}
    \caption{Schedule transformations in the proof of Lemma~\ref{lem:at-start-of-b}}
    \label{fig:A-conf-init}
    \end{center}
    \end{figure*}

We now prove \ref{it:start1} and \ref{it:start2}.  Observe that by \ref{it:start3}, $\startTime{\cP}{b}$ is not an integer and thus
\begin{equation} \label{eq:two}
|\jobs(\xi_{j'-1})|\geq 2
\end{equation}
for otherwise $\cP$ would not be optimal --- a contradiction.

Suppose first that $\startTime{\cP}{a}=\startTime{\cP}{b}=e_{j'}$. If a job $a'$ in $\jobs(\xi_{j'-1})$ does not complete at $e_{j'}$, i.e., $a'$ is preempted at $e_{j'}$, then $\complTime{\cP}{a'}> e_{\tau(b)}$ for otherwise,
by Lemma~\ref{lem:how_jobs_finish}, at most one of jobs $\{a,b\}$ can be spanning in block $\tau(a')$ and thus the other
job in $\{a,b\}$ and $a'$ would interlace, which contradicts Lemma~\ref{lem:interlace}.
However, if $\complTime{\cP}{a'}>e_{\tau(b)}$, then job $a'$ is spanning in block $\tau(b)$ for otherwise $a'$ and $b$ interlace, which again contradicts Lemma~\ref{lem:interlace}. Thus, $|\jobs(\xi_{\tau(b)})|=2$ by Lemma \ref{lem:how_jobs_finish}.
Therefore, a job in $\jobs(\xi_{j'-1})\setminus\{a'\}$ completes at $e_{j'}$.
The other conditions of the lemma trivially follow when $\startTime{\cP}{a}=\startTime{\cP}{b}$.

Now let $\startTime{\cP}{a}\neq\startTime{\cP}{b}$. By assumption, $\startTime{\cP}{a}<\startTime{\cP}{b}$.
Informally, the proof is divided into two stages.
In the first stage we consider block $j'-1$ and we prove that $\jobs(\xi_{j'-1})=\{a,c\}$ and that $\tau(c)=j'-1$ --- see Equations~\eqref{eq:b:a}, \eqref{eq:b:size_of_j-1} and \eqref{eq:b:c_ends} and Figure~\ref{fig:A-conf-init}(b).
In the second stage we prove that $a$ starts at $e_{j'-1}$.
The proof of the latter is done by contradiction, i.e., we suppose that $a$ starts before $e_{j'-1}$.
This assumption implies that $\cP$ looks  as shown in Figure~\ref{fig:A-conf-init}(c) in the interval $[e_{j'-3},e_{j'+1}]$, which allows us to get the desired contradiction thanks to Lemma~\ref{lem:acc1d}.

First we prove by contradiction that
\begin{equation} \label{eq:b:a}
\xi_{j'-1}(a)=e_{j'}-e_{j'-1}.
\end{equation}
By \eqref{eq:two}, $\jobs(\xi_{j'-1})\setminus\{a\}\neq\emptyset$.
Take any $c\in\jobs(\xi_{j'-1})\setminus\{a\}$.
By \eqref{eq:a-with-b}, $\xi_{j'}(c)=0$.
Since job $a$ is non-spanning in block $j'-1$, the conditions \ref{it:advance:first}-\ref{it:advance:last} of  Lemma~\ref{lem:advance:a} are all satisfied by jobs $a$ and $c$, and $t=j'$.
(Condition \ref{it:advance4} holds as $j'$ is not the abnormality point of $\cP$ by (\ref{eq:a-with-b}) and Lemma~\ref{lem:at_least_two_abnormal}.)
Therefore we get a contradictions, and (\ref{eq:b:a}) holds.

Next, we show that
\begin{equation} \label{eq:b:size_of_j-1}
\jobs(\xi_{j'-1})=\{a,c\} \textup{ for some }c\in\jobs.
\end{equation}
If no job completes at $e_{j'}$, then Lemma~\ref{lem:one_dominates} and (\ref{eq:two}) immediately imply \eqref{eq:b:size_of_j-1}.
If some job, say $c$, completes at $e_{j'}$, then Lemma~\ref{lem:how_jobs_finish} implies that job $c$ is spanning in block
$j'-1$. Since $a$ completes after $e_{j'}$, $a\neq c$. This and \eqref{eq:b:a} imply \eqref{eq:b:size_of_j-1}.

Finally to complete the first stage, we prove by contradiction that
\begin{equation} \label{eq:b:c_ends}
\complTime{\cP}{c}=e_{j'}.
\end{equation}
To that end take $\varepsilon=\min\{\xi_{j'-1}(c),\xi_{j'}(b)\}$ and let $\cP'$ be a schedule with events $\vect{e'}$ and partition $\vect{\xi'}$, where
\[
(\vect{e'},\vect{\xi'})=\langle \vect{e},\vect{\xi},\varepsilon, (j'-1\cyclic{c} j'\cyclic{b} j'-1 \rangle.
\]
Note that
\begin{eqnarray*}
 \startTime{\cP'}{b} & \geq & e_{j'-1}=e_{j'+1}-(e_{j'+1}-e_{j'-1})\\
                     &  =   & e_{j'+1}-\xi_{j'-1}(a)-\xi_{j'}(a)\geq e_{j'+1}-1,
\end{eqnarray*}
which, by \ref{it:start3}, implies that $\startTime{\cP'}{b}\geq r(b)$.
Thus, $\cP'$ is feasible and, by assumption, optimal.
Also, by \eqref{eq:a-with-b}, \eqref{eq:b:size_of_j-1} and Lemma \ref{lem:at_least_two_abnormal}, we have $i\notin\{j'-1,j'\}$, where $i$ is the abnormality point of $\cP$.
Thus, as before, $i$ is the abnormality point of $\cP'$.
Indeed, it follows from the fact that the two schedules are identical in $[0,e_{j'}]$ (which covers the case when $i<j'$), and from Lemma~\ref{lem:abnormality_preserved} (that covers the case when $i>j'$).
Moreover, $\cP'$ contains a block that ends at $e_{j'+1}$ and contains the jobs $a,b$ and $c$, none of which completes at $e_{j'+1}$ --- a contradiction with Lemma~\ref{lem:one_dominates}. Therefore, (\ref{eq:b:c_ends}) holds, and thus due
to Equations~\eqref{eq:b:a}, \eqref{eq:b:size_of_j-1} and \eqref{eq:b:c_ends}, the schedule in the interval $[e_{j'-1},e_{j'+1}]$ looks like in Figure~\ref{fig:A-conf-init}(b).

In the second stage we argue that
\begin{equation} \label{eq:b:start-a}
\startTime{\cP}{a}=e_{j'-1}.
\end{equation}
Suppose for a contradiction that this is not the case.
By \eqref{eq:b:a}, $c$ does not start at $e_{j'-1}$.
Since  $a$ does not starts at $e_{j'-1}$ either, there is a job, say $c'$ that ends at $e_{j'-1}$, otherwise $e_{j'-1}$ would not be an event.
By Lemma~\ref{lem:how_jobs_finish},
\begin{equation} \label{eq:j'-2}
\xi_{j'-2}(c')=e_{j'-1}-e_{j'-2},
\end{equation}
which implies
\begin{equation} \label{eq:a_j'-2}
\xi_{j'-2}(a)=e_{j'-1}-e_{j'-2}
\end{equation}
as follows: First we observe that there is no job $d\neq c'$ that completes at $e_{j'-1}$.
Indeed, otherwise Lemma~\ref{lem:ccd} applied to $c=d$, $c'$, $d=c$, and $e_j=e_{j'-1}$ gives the required contradiction.
Now, if $c\in \jobs(\xi_{j'-2})$, then the conditions \ref{it:advance:first}-\ref{it:advance:last} of Lemma~\ref{lem:advance:a} are all satisfied by jobs $a$, $b$ and $c$, and $t=j'-1$ --- a contradiction.
(Condition \ref{it:advance4} holds as neither $j'-1$ nor $j'$ is the abnormality point of $\cP$ by (\ref{eq:a-with-b}), (\ref{eq:b:size_of_j-1}), and Lemma~\ref{lem:at_least_two_abnormal}.)
Therefore, $\xi_{j'-2}(c)=0$.
Thus, $\jobs(\xi_{j'-2})\subseteq \{a,c'\}$, because if a job different than $a$ and $c'$ that does not complete at $e_{j'-1}$ is present in $\jobs(\xi_{j'-2})$ then, by \eqref{eq:b:size_of_j-1} and \eqref{eq:b:c_ends}, this job interlaces with $c$ that contradicts Lemma~\ref{lem:interlace}.
This implies \eqref{eq:a_j'-2} as required.

If job $c'$ is non-spanning in block $j'-3$, then by \eqref{eq:b:size_of_j-1}, \eqref{eq:a_j'-2} and $\complTime{\cP}{c'}=e_{j'-1}$, $\startTime{\cP}{c'}<e_{j'-2}$, which implies that $c$ and $c'$ form an $\A$-configuration of length $e_{j'-1}-e_{j'-2}$ at $e_{j'-1}$, which leads to a contradiction with $\A$-maximality of $\cP$.
Thus we have
\begin{equation}\label{eq:j'-3}
\xi_{j'-3}(c')=e_{j'-2}-e_{j'-3}.
\end{equation}
We prove that
\begin{equation} \label{eq:b:inductiona1}
 \begin{split}
  \xi_{j'-3}(c')=\xi_{j'-3}(c_1)=e_{j'-2}-e_{j'-3}\textup{ and} \\
  \complTime{\cP}{c_1}=e_{j'-2}\textup{ for some }c_1\in\jobs\setminus\{a,c\},
 \end{split}
\end{equation}
i.e., we prove that $\cP$ in the interval $[e_{j'-3},e_{j'+1}]$ is as shown in Figure~\ref{fig:A-conf-init}(c).
First, we have $\xi_{j'-3}(c_1)>0$ for some $c_1\notin \{a,c,c'\}$.
Otherwise $\jobs(\xi_{j'-3})\subseteq \{a,c,c'\}$, and since $e_{j'-2}$ is an event, $\startTime{\cP}{a}=e_{j'-2}$.
Then, however, conditions \ref{it:advance:first}-\ref{it:advance:last} of  Lemma~\ref{lem:advance:a} are all satisfied by jobs $a$, $b$, $c$, and $t=j'-2$ --- a contradiction (observe that $h-\startTime{\cP}{a}<1$, thus \ref{it:advance1} is satisfied; condition \ref{it:advance4} holds as none of $j'-2$, $j'-1$, $j'$ is the abnormality point of $\cP$ by (\ref{eq:a-with-b}), (\ref{eq:b:size_of_j-1}), (\ref{eq:j'-2}), (\ref{eq:a_j'-2}), and Lemma~\ref{lem:at_least_two_abnormal}). Second, each such $c_1$ completes at $e_{j'-2}$
for otherwise, by \eqref{eq:b:size_of_j-1}, \eqref{eq:a_j'-2} and \eqref{eq:j'-3}, $c_1$ and $c$ interlace --- a contradiction by Lemma~\ref{lem:interlace}.
Thus, by Lemma \ref{lem:how_jobs_finish}, job $c_1$ is spanning in block $j'-3$.
This, and \eqref{eq:j'-3} imply \eqref{eq:b:inductiona1}.
Thus, $\cP$ looks in the interval $[e_{j'-3},e_{j'+1}]$ as shown in Figure~\ref{fig:A-conf-init}(c). Finally, by Lemma \ref{lem:acc1d} applied to $c=c'$, $c_1$, $d=c$, $a$, and $e_j=e_{j'-2}$, we obtain that $c$ and $c'$ form
an $\A$-configuration at $e_{j'-1}$.
Thus, again, we get a contradiction since $\cP$ is $\A$-maximal.
Hence, \eqref{eq:b:start-a} holds.
Therefore the lemma follows by (\ref{eq:a-with-b}), (\ref{eq:b:size_of_j-1}), and (\ref{eq:b:start-a}).
\end{proof}

Given schedule $(\cP, \vect{e},\vect{\xi})$, $l\geq 1$ and $\{a_1,\ldots,a_l\}\subseteq \jobs$, job sequence $(a_1,\ldots, a_l)$
is called a \emph{sub-chain starting at} $t$ in $\cP$ if:
\begin{enumerate} [label={\normalfont(S\arabic*)},leftmargin=*]
 \item\label{it:C1} For each $j\in\{1,\ldots,l-1\}$, $a_j\preceq a_{j+1}$;
 \item\label{it:C2} For each $j\in\{1,\ldots,l-1\}$, $\complTime{\cP}{a_j}=\startTime{\cP}{a_{j+1}}$;
 \item\label{it:C3'} Job $a_1$ executes non-preemptively in $[t,\complTime{\cP}{a_1}]$.
\end{enumerate}
Moreover, job sequence $(a_1,\ldots,a_l)$ is a \emph{chain} in $\cP$ if it satisfies conditions \ref{it:C1}, \ref{it:C2}
and additionally
\begin{enumerate} [label={\normalfont(S\arabic*)},leftmargin=*] \setcounter{enumi}{3}
 \item\label{it:C3} Time $t$ is the earliest moment such that $a_1$ executes with no preemption in $[t,\complTime{\cP}{a_1}]$;
 \item\label{it:C4} No predecessor of $a_1$ ends at $t$.
\end{enumerate}

Suppose that jobs $a$ and $b$ form an $\A$-configuration in $\cP$ with $\complTime{\cP}{a}<\complTime{\cP}{b}$.
For $\varepsilon\geq 0$, we define an operation of \emph{$\varepsilon$-exchanging} $a$ and $b$ in an interval $[e_k,\complTime{\cP}{b}]$, $k<\tau(b)$, as follows.
First, all pieces of $a$ and $b$ are removed from the blocks $k,\ldots,\tau(b)$.
Note that the total length of all removed pieces of $a$ and $b$ is $\sum_{t=k}^{\tau(b)}\xi_t(a)$ and $\sum_{t=k}^{\tau(b)}\xi_t(b)$, respectively.
Then, the empty gaps are filled out with the total length $\sum_{t=k}^{\tau(b)}\xi_t(a)-\varepsilon$ of $a$ and the total length of $\sum_{t=k}^{\tau(b)}\xi_t(b)+\varepsilon$ of $b$ in such a way that $b$ completes as early as possible.
Note that the new schedule is valid only if $\varepsilon=0$.
Whenever the transformation of $\varepsilon$-exchanging will be used with $\varepsilon>0$, then some other appropriate changes in the schedule will be made to ensure feasibility.

For $\varepsilon>0$, we extend the operation of $\varepsilon$-exchanging of two jobs to sub-chains as follows.
Let $A=(a_1,\ldots,a_l=a)$ and $B=(b_1,\ldots,b_{l'}=b)$ be two sub-chains in $\cP$ that start at $t$, and such that $a$ and $b$ form an $\A$-configuration in $\cP$, where $\startTime{\cP}{a}\leq\startTime{\cP}{b}$.
(Note that we either have $l=l'$ or $l=l'-1$.)
Let $d$ be any job that executes in $[t-\varepsilon,t]$.
The operation of \emph{$(\varepsilon,d)$-exchanging} of $A$ and $B$ in $\cP$ leads to a schedule $\cP'$ obtained by making the following changes to $\cP$:
\begin{itemize}
 \item For each $j\in\{1,\ldots,l\}$, $a_j$ is executed in $[t_j-\varepsilon,t_{j+1}-\varepsilon]$ in $\cP'$, where take $t_1=t$, $t_j=\startTime{\cP}{a_j}$ for $j\in\{2,\ldots,l\}$ and $t_{l+1}=e_{j'+1}$ such that $e_{j'}=\startTime{\cP}{b}$;
 \item For each $j\in\{1,\ldots,l'\}$, $b_j$ is executed in $[u_j+\varepsilon,u_{j+1}+\varepsilon]$ in $\cP'$, where take $u_1=t$, $u_j=\startTime{\cP}{b_j}$ for $j\in\{2,\ldots,l'\}$ and $u_{l'+1}=e_{j'+1}$;
 \item A piece of $d$ executing in $[t-\varepsilon,t]$ is placed in $[t,t+\varepsilon]$ in $\cP'$ (the ``room'' for this job is made
     by postponing $b_1$);
 \item In the interval $[e_{j'+1},\complTime{\cP}{b}]$ $\varepsilon$-exchanging of $a$ and $b$ is made.
\end{itemize}
The transformation is illustrated in Figure~\ref{fig:exchanging-chains} for $d=b_1$. Note that in this particular case the total
completion times of $\cP$ and $\cP'$ are equal.

    \begin{figure*}[htb]
    \begin{center}
    \includegraphics[scale=1.0]{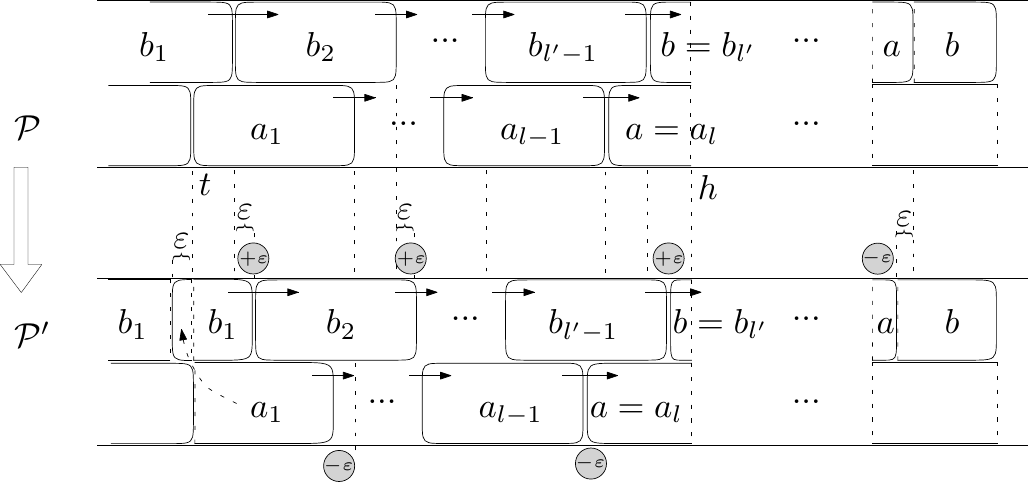}
    \caption{$(\varepsilon,d)$-exchanging of $(a_1,\ldots,a_l)$ and $(b_1,\ldots,b_l)$ when $l'=l+1$, $d=b_1$ and $\startTime{\cP}{a}=t$}
    \label{fig:exchanging-chains}
    \end{center}
    \end{figure*}

The new schedule $\cP'$ is valid under certain conditions. First, the value of $\varepsilon$ must be selected in such a way that
$\varepsilon$-exchanging of $a$ and $b$ is possible in the above-mentioned interval. Second, $d$ should not be a predecessor
of $a_1$. Also, the release dates of jobs $a_1,\ldots,a_l$ need to be respected and $a_1$ must be non-preemptively executed in
$[t,t+\varepsilon]$. We summarize those conditions in the following lemma.

\begin{lemma} \label{lem:swapping_chains}
Let $(a_1,\ldots,a_l=a)$ and $(b_1,\ldots,b_l=b)$, starting at $t$, be two sub-chains in $\cP$ such that $a$ and $b$ form an $\A$-configuration of length $\ell$ in $\cP$.
If $\varepsilon\leq\ell$, $r(a_1)\leq t-\varepsilon$ and $r(a_j)\leq\startTime{\cP}{a_j}-\varepsilon$ for each $j\in\{2,\ldots,l\}$, $a_1$ executes non-preemptively in $[t,t+\varepsilon]$, and some job $d$ that is not a predecessor of $a_1$ executes non-preemptively in $[t-\varepsilon,t]$, then the schedule $\cP'$ obtained by $(\varepsilon,d)$-exchanging of the two sub-chains in $\cP$ is valid.
\qed
\end{lemma}

\subsubsection*{Proof of Proposition~\ref{pro:no-A-configurations}}

Let $(\cP, \vect{e}, \vect{\xi})$ be a maximal schedule that satisfies the properties in Lemma~\ref{lem:at-start-of-b}.
Let $(a_1,\ldots,a_l=a)$ be the chain in $\cP$ that starts at $t_a$ and let $(b_1,\ldots,b_{l'}=b)$ be the chain in $\cP$
that starts at $t_b$. By definition of chains and Lemma~\ref{lem:at-start-of-b} we have
\begin{equation}\label{eg:ta1}
l=1\quad\Rightarrow\quad t_a=\startTime{\cP}{a},
\end{equation}
\begin{equation}\label{eg:ta2}
l\geq 2\quad\Rightarrow\quad \startTime{\cP}{a}-(l-1)\leq t_a<\startTime{\cP}{a}-(l-2),
\end{equation}
\begin{equation}\label{eg:tb1}
l'=1\quad\Rightarrow\quad t_b=\startTime{\cP}{b}, \textup{ and}
\end{equation}
\begin{equation}\label{eg:tb2}
l'\geq 2\quad\Rightarrow\quad \startTime{\cP}{b}-(l'-1)\leq t_b<\startTime{\cP}{b}-(l'-2).
\end{equation}

\paragraph{Case 1: $t_a\geq t_b$.}

 In this case we perform a transformation shown in Figure~\ref{fig:case1} as described below.
 By Lemma~\ref{lem:at-start-of-b}, there exists an integer $h$ such that both jobs $a$ and $b$ execute non-preemptively in $[\startTime{\cP}{a},h]$ and $[\startTime{\cP}{b},h]$, respectively.
 We have $e_p=t_a$ for some event $e_p$.

    \begin{figure*}[htb]
    \begin{center}
    \includegraphics[scale=1.0]{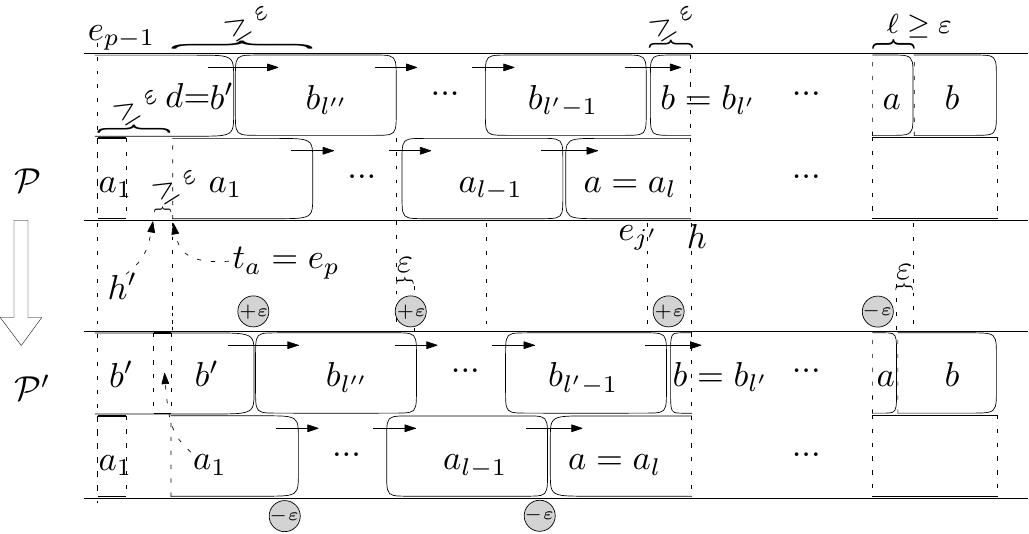}
    \caption{Transformation from $\cP$ to $\cP'$ ($\varepsilon=t_a-h'$, $d=b'$) in Case 1 in proof of Proposition~\ref{pro:no-A-configurations}}
    \label{fig:case1}
    \end{center}
    \end{figure*}

 Let $e_{j'}=\startTime{\cP}{b}$.
 By Lemma~\ref{lem:at-start-of-b}, $h-e_{j'}=\xi_{j'}(b)$.
 Let $\ell$ be the length of the $\A$-configuration formed by $a$ and $b$.
 Clearly $\complTime{\cP}{a}-\ell > h$ by definition of $\A$-configuration and Lemma~\ref{lem:at-start-of-b}.

 Let $A=(a_1,\ldots,a_l)$ and let $B$ be the sub-chain of the chain $(b_1,\ldots,b_{l'}=b)$ that starts at $t_a$ with a job $b'$ and ends with the job $b$.
 By definition of $t_a$, $\xi_{p-1}(a_1)<e_p-e_{p-1}$.
 Thus, $|\jobs(\xi_{p-1})\setminus\{a_1\}|\geq 2$.
 Let $d\in \jobs( \xi_{p-1})\setminus\{a_1\}$ be a job that does not complete at $e_p$ (possibly $b'$), if such a job exists.
 Otherwise, let $d$ be any job in $\jobs( \xi_{p-1})\setminus\{a_1\}$.
 Take
  \[
  \varepsilon=\min\left\{\xi_{p-1}(d),e_p-e_{p-1}-\xi_{p-1}(a_1),y,\xi_{j'}(b),\ell,t_a-h'\right\},
  \]
 where $h'$ is the greatest integer smaller than $t_a$ and
 \[
  y = \begin{cases}
        h-\startTime{\cP}{a}, & \textup{when }l=1, \\
        (\complTime{\cP}{a_1}-e_p)/2, & \textup{when $l>1$ and $\complTime{\cP}{d}=t_a$}, \\
        \complTime{\cP}{a_1}-e_p, & \textup{otherwise.}
      \end{cases}
 \]
 The latter ensures that $d$, if it completes at $t_a$ in $\cP$, does not complete after $\startTime{\cP'}{a_2}$ in $\cP'$.
 Note that, by definition of $t_a$, no predecessor of $a_1$ ends at $t_a$ and $\xi_{p-1}(a_1)<e_p-e_{p-1}$.
 Hence, in particular, $\varepsilon>0$.
 Let $\cP'$ be the schedule obtained by $(\varepsilon,d)$-exchanging of $A$ and $B$ in $\cP$.
 By Lemma~\ref{lem:swapping_chains}, $\cP'$ is feasible.
 If $\varepsilon=\ell$, then the total completion time of $\cP'$ is strictly smaller than that of $\cP$ and we get a contradiction since $\cP$ is optimal.

 Thus, consider $\varepsilon<\ell$.
 Then, the total completion time of $\cP'$ does not exceed that of $\cP$.
 To see that we observe that by (\ref{eg:ta1}) and (\ref{eg:ta2}) we have $\startTime{\cP}{b}-t_a<l$.
 Also, if two jobs in $\jobs(\xi_{p-1})\setminus\{a_1\}$ complete at $t_a$, then either at least one of them is a predecessor of $b'$, which implies that $\startTime{\cP}{b'}=t_a$, or otherwise we obtain from Lemma~\ref{lem:ccd} that $\startTime{\cP}{b'}=t_a$.
 Therefore, no more than $l$ jobs in $\{d,b_1,\ldots,b_{l'}\}$ complete in $[t_a,h]$ in $\cP$.
 Thus, no more than $l$ jobs get delayed by $\varepsilon$ each as a result of the exchange, however, each job in the chain  $(a_1,\ldots,a_l=a)$ completes by $\varepsilon$ earlier at the same time.

 Finally, we show that $\cP$ and $\cP'$ have the same abnormality point.
 Clearly, this holds if $i<p-1$.
 Also, if $i>j'$, then $\varepsilon$ is $p$-normal.
 To see this we observe that $e_p$, $e_{p-1}$, $\xi_{p-1}(a_1)$ and $t_a$ are clearly all $p$-normal.
 By Lemma~\ref{lem:remaining_part_i-normal}, $\complTime{\cP}{a_1}-t_a$ is $p$-normal.
 Also $e_{j'}=t_a+(\complTime{\cP}{b'}-t_a)+k-2$, where $k$ is the number of jobs in $B$, is $p$-normal.
 If $l=1$, then $\startTime{\cP}{a}$ and $h$ are $p$-normal, which implies $p$-normality of $y$.
 For $l>1$, we argue that $y$ is also $p$-normal and for that we need only consider $y=(\complTime{\cP}{a_1}-e_p)/2$.
  Then, $b'$ is not present in the $(p-1)$-st block for otherwise $b'$ would be selected as $d$.
  The length of $(p-1)$-st block, $e_p-e_{p-1}$, is by definition $(p-1)$-normal.
  By Lemma~\ref{lem:how_jobs_finish}, $\xi_{p-1}(d)=e_{p}-e_{p-1}$.
  By Lemma~\ref{lem:interlace}, each job in $\jobs(\xi_{p-1})\setminus\{a_1\}$ must complete at $e_{p}$.
  This proves, again by Lemma~\ref{lem:how_jobs_finish}, that $|\jobs(\xi_{p-1})|=2$.
  By Lemma~\ref{lem:remaining_part_i-normal}, $\xi_{p-1}(a_1)+\xi_p(a_1)+\xi_{p+1}(a_1)=\xi_{p-1}(a_1)+\complTime{\cP}{a_1}-t_a$ is $(p-1)$-normal.
  Since $\xi_{p-1}(a_1)\in\{0,e_p-e_{p-1}\}$, we obtain that $\xi_{p-1}(a_1)$ is $(p-1)$-normal.
  Thus, $(\complTime{\cP}{a_1}-t_a)/2$ is $p$-normal as required.
 Therefore, $\varepsilon$ is $p$-normal and, by Lemma~\ref{lem:abnormality_preserved}, $\cP$ and $\cP'$ have the same abnormality point for $i>j'$.
 Also, by Lemma \ref{lem:at-start-of-b}, and chain definition, we have $|\jobs(\xi_k)|=2$ for each $k\in\{p,\ldots,j'\}$.
 Thus, by Lemma \ref{lem:at_least_two_abnormal}, $i\notin \{p,\ldots,j'\}$.
 Finally, consider $i=p-1$.
 Then, if $i$ is no longer the abnormality point $i'$ of $\cP'$, then $i'>i$ --- a contradiction since $\cP$ is $\A$-maximal.
 Therefore, $i$ is the abnormality point of $\cP'$, and hence we proved that $\cP$ and $\cP'$ have the same abnormality point.
 To complete the case proof we note that $a$ and $b$ form an $\A$-configuration in $\cP'$ at $\complTime{\cP'}{a}=\complTime{\cP}{a}-\varepsilon$, which contradicts  the $\A$-maximality of $\cP$.

\paragraph{Case 2: $t_a<t_b$.}
We first define
\[t_a'=\max\left\{ t< t_b \st t=t_a\textup{ or }t\in\{\startTime{\cP}{a_1},\ldots,\startTime{\cP}{a_l}\} \right\}\]
and $a'$ to be the job from the chain $(a_{1},\ldots,a_{l})$ that starts or resumes at $t'_a$.
 By (\ref{eg:ta1}-\ref{eg:tb2}), it holds $a'=a_{l-l'+1}$ or $a'=a_{l-l'+2}$, and only one job, namely $a'$, from the chain $(a_1,\ldots,a_l)$ is executed in $(t'_a,t_b)$.
By definition $t'_a<t_b$, also we have $t'_a=e_{p}$ for some event $e_p$.

 We first prove that exactly one job that is not in the chain $(a_{1},\ldots,a_{l})$, call it $d$, executes in $[t'_a,t_b]$ and completes at $t_b$.
  Indeed, if $l'=1$, then this follows from Lemma~\ref{lem:at-start-of-b}.
  If $l'>1$, then any job not in the chain $(a_{1},\ldots,a_{l})$ that executes in $[t'_a,t_b]$ completes in $[t'_a,t_b]$, otherwise this job interlaces with $b_1$ --- a contradiction with Lemma~\ref{lem:interlace}.
  Finally,  we show that two or more jobs not in the chain $(a_{1},\ldots,a_{l})$ cannot complete in $[t'_a,t_b]$.
  If there are at least three such jobs, then the last two of them form an $\A$-configuration, which contradicts the $\A$-maximality of $\cP$.
  For exactly two, $c'$ and $c$ completing in this order, by Claim \ref{claim:acc1d1}, $c$ and $c'$ form an $\A$-configuration at $e_{p+1}$ --- a contradiction since $\cP$ is $\A$-maximal.
  Also, observe that for $l'=l+1$, job $b_1$ resumes at $t_b$ and thus $b_1$ and $d$  form an $\A$-configuration at $e_{p+1}$ by Claim \ref{claim:acc1d1} --- a contradiction since $\cP$ is $\A$-maximal.
  Thus, let $l\geq l'$ in the reminder of the lemma.

  Now we prove that our schedule $\cP$ satisfies the following claim that we have used above (see Figure~\ref{fig:claim}(a) for illustration of Claim~\ref{claim:acc1d1}):

\begin{claim}~\label{claim:acc1d1}
Suppose that $t'_a=e_p$ is an event in $\cP$ and that there exist jobs $a'$, $c$, and $d$ such that
\begin{enumerate}[label={\normalfont(\roman*)},leftmargin=*]
 \item \label{it:a11}~$\complTime{\cP}{c}=e_{p+1}$, and $\complTime{\cP}{d}=e_{p+2}$;
 \item \label{it:a21}~$\jobs(\xi_{p})=\{a',c\}$ and $\jobs(\xi_{p+1})=\{a',d\}$;
 \item \label{it:a31}~$\startTime{\cP}{d}<e_{p}$;
 \item \label{it:a41}~if $l'=l+1$, then $d=b_1$; otherwise $\complTime{\cP}{d}=t_b$.
\end{enumerate}
Then jobs $c$ and $d$ form an $\A$-configuration at $e_{p+1}$.
\end{claim}

\begin{proof}
If $\xi_{p-1}(c)<e_p-e_{p-1}$, then the jobs $c$ and $d$ form an $\A$-configuration of length $e_{p+1}-e_{p}$ at $e_{p+1}$ --- the lemma holds.
Thus,
\begin{equation} \label{eq:p-1}
  \xi_{p-1}(c)=e_p-e_{p-1}.
\end{equation}
  We now prove that in interval $[e_{p-1},e_{p+2}]$ the schedule $\cP$ is as in Figure~\ref{fig:claim}, i.e., there exists a job $c_1$ such that
  \begin{equation} \label{eq:b:induction1}
  \xi_{p-1}(c_1)=e_{p}-e_{p-1}\ \land\ \complTime{\cP}{c_1}=e_{p} \ \land\ \startTime{\cP}{d}<e_{p-1}.
  \end{equation}

   \begin{figure}[htb]
    \begin{center}
    \includegraphics[scale=1.0]{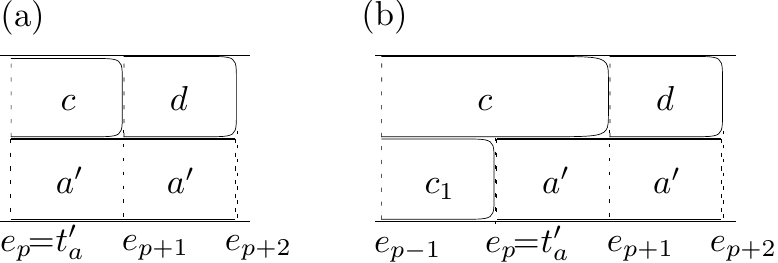}
    \caption{(a) Illustration of Claim~\ref{claim:acc1d1}; (b) Proof of Claim~\ref{claim:acc1d1}}
    \label{fig:claim}
    \end{center}
   \end{figure}

  First, we show that $\xi_{p-1}(c_1)>0$ for some $c_1\notin \{a',c,d\}$.
  Otherwise, by \eqref{eq:p-1} and \ref{it:a21}, $\startTime{\cP}{a'}=e_p$ since $e_p$ is an event.
  Thus, $\jobs(\xi_{p-1})=\{c,d\}$.
  Now, take
  \[
  \varepsilon=\min\big\{\xi_{p-1}(d),\xi_{j'}(b),\ell,t'_a-h'\big\},
  \]
 where $h'$ is the greatest integer smaller than $t'_a$.
 Observe that $\varepsilon>0$.
 Let $A=(a_1,\ldots,a_l)$ and let $B=(d,b_1,\ldots,b_{l'}=b)$ be the sub-chain that starts at $e_{p+1}$.
 Perform the $(\varepsilon,d)$-exchanging of $A$ and $B$ in $\cP$ as in Case 1 (the completion time of $c$ does not change in this transformation when $d\neq b_1$ because $a'$ from the $\tau(a')$-th block is placed in the $(p-1)$-st block and $d$ from $(p-1)$-st block is placed in the $(\tau(d)+1)$-st block) to get a contradiction.
 Observe that, by \ref{it:a41}, $d=b_1$ for $l'=l+1$ and thus the $(\varepsilon,d)$-exchanging of $A$ and $B$ indeed produces schedule $\cP'$ with total completion time that does not exceed that of schedule $\cP$.
 Also, by Lemma~\ref{lem:no_idle-time_between_preemptions}, there is no idle time in the $(p-1)$-st block of $\cP$.
 This implies, by Lemma~\ref{lem:remaining_part_i-normal}, that $\xi_{p-1}(d)=e_{p}-e_{p-1}$ is $(p-1)$-normal, which by arguments in Case~1 implies that the abnormality points of $\cP$ and $\cP'$ are the same.

 Second, by Lemma~\ref{lem:interlace}, $c_1$ and $d$ cannot interlace, which implies $\complTime{\cP}{c_1}=e_{p}$.
 By Lemma~\ref{lem:how_jobs_finish}, job $c_1$ is spanning in block $p-1$.
 Thus, by (\ref{eq:p-1}) we have $\{c_1,c\}=\jobs(\xi_{p-1})$.
 Finally,  $\startTime{\cP}{d}<e_{p-1}$ is due to \ref{it:a31} and $\jobs(\xi_{p-1})=\{c_1,c\}$.
 This completes the proof of \eqref{eq:b:induction1}.
 Equation~\eqref{eq:b:induction1} allows us to apply Lemma~\ref{lem:acc1d} to $c$, $c_1$, $d$, $a=a'$, and $e_j=e_{p}$ to conclude that $c$ and $d$ form an $\A$-configuration at $e_{p+1}$.
 This contradicts the $\A$-maximality of $\cP$ and completes the proof of Claim~\ref{claim:acc1d1}.
\end{proof}

    \begin{figure}[htb]
    \begin{center}
    \includegraphics[scale=1.0]{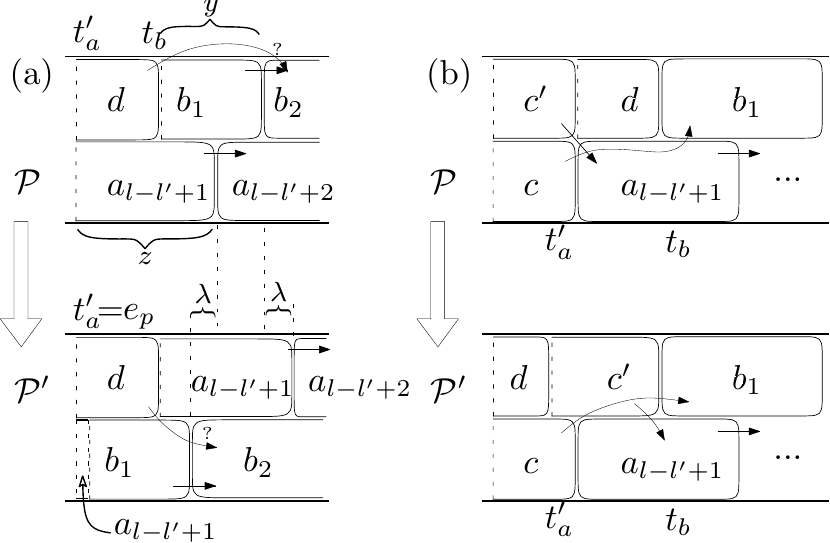}
    \caption{Transformations from $\cP$ to $\cP'$ in Case 2}
    \label{fig:A-conf2a}
    \end{center}
    \end{figure}

  Let $e_{j'}=\startTime{\cP}{b}$. Now, let $y=\sum_{t\geq p}\xi_t(b_1)$ and $z=\sum_{t\geq p}\xi_t(a_{l-l'+1})$.
  First we prove that $z\leq y$.
  This holds due to Lemma~\ref{lem:at-start-of-b} when $l'=1$ and hence let $l\geq l'>1$.
  If $z>y$, then swap $a$ and $b$ and then do the $(\varepsilon,a_{l-l'+1})$-exchanging of $(b_1,\ldots,b_{l'})$ and $(a_{l-l'+1},\ldots,a_l)$ (note the order of the chain, which is important) both starting at $t_b$, where $\varepsilon=\startTime{\cP}{b}-\startTime{\cP}{a} + \lambda$ and $0<\lambda<\min\{\ell,z-y,\xi_{j'}(b)\}$.
  This transformation is shown in Figure~\ref{fig:A-conf2a}(a).
  Let the resulting schedule be $\cP'$.
  The swapping increases the total completion time by $\startTime{\cP}{b}-\startTime{\cP}{a}$ and the $(\varepsilon,a_{l-l'+1})$-exchanging decreases it by $((l'-1)-l')\varepsilon$ --- observe that after the swapping of $a$ and $b$ the completion time of job $a=a_l$ does not change in the exchange.
  Therefore, the overall change equals $-\lambda$ and thus to get a contradiction it suffices to prove that $\cP'$ is feasible.

  Observe that $\startTime{\cP}{b}-\startTime{\cP}{a}=\complTime{\cP}{b_1}-\complTime{\cP}{a_{l-l'+1}}$.
  By Lemma~\ref{lem:at-start-of-b}, $e_{j'+1}$ is an integer.
  Thus, $r(b)<e_{j'+1}$ implies $r(b)\leq e_{j'+1}-1$.
  Moreover, $\startTime{\cP'}{b}\geq e_{j'+1}-1$.
  Therefore, by the definition of a sub-chain, all jobs $b_{2},\ldots,b_{l'}$ respect their release dates in $\cP'$.
  Since $z>y$, we have that $\startTime{\cP}{b_1}<t_a'$ and hence $b_1$ respects its release date in $\cP'$.
  Thus, $z\leq y$ for the reminder of the proof. We consider the following three subcases.

  \begin{description}
  \item{Case 2a:}\label{it:case2a} $t'_a=\startTime{\cP}{a_{l-l'+1}}$.
    (Schedule transformation performed in this case is shown in Figure~\ref{fig:A-conf2a}(b).)
    Then, $z=1$. Since $z\leq y$, we have $y=1$.
    If some job in $\jobs(\xi_{p-1})$ does not complete at $t_a'$, then this case reduces to Case~1.
    Otherwise, two jobs in $\jobs(\xi_{p-1})$ complete at $t_a'$.
    Thus, by Lemma~\ref{lem:one_dominates}, for at least one job in $\jobs(\xi_{p-1})$, say job $c'$, we have $\xi_t(c')=0$, where $e_t=\startTime{\cP}{d}$.
    Therefore, $c'$ and $d$ interlace if $\startTime{\cP}{c'}<\startTime{\cP}{d}$ --- a contradiction by Lemma \ref{lem:interlace}, or we can swap jobs $c'$ and $d$ if $\startTime{\cP}{c'}>\startTime{\cP}{d}$.
    In the latter case the resulting schedule (see Figure~\ref{fig:A-conf2a}(b)) reduces the total completion time of $\cP$ by Lemma \ref{lem:swapping}.
    This schedule is not feasible when $c'\prec a_{l-l'+1}$ and we restore feasibility by applying a $0$-exchanging of $b$ and $a$ in $[e_{j'+1},\complTime{\cP}{b}]$ followed by $(t_b-t_a',a_{l-l'+1})$-exchanging of $(b_1,\ldots,b_{l'})$ and $(a_{l-l'+1},\ldots,a_l)$ both starting at $t_b$.
    The new schedule $\cP'$ is feasible since $c'\nprec b_1$ for in-trees, and since, by Lemma~\ref{lem:at-start-of-b}, $\startTime{\cP'}{b}\geq e_{j'+1}-1\geq r(b)$, which shows that all $b_1,\ldots,b_{l'}$ respect their release dates in $\cP'$. Thus, we get a contradiction since $\cP$ is optimal.

    \item{Case 2b:}\label{it:case2b} $t'_a=\startTime{\cP}{a_{l-l'+2}}$.
    Since $l\geq l'$, $\xi_{p-1}(a_{l-l'+1})=e_p-e_{p-1}$.
    Also, $t_a'<t_b$ implies that $b_1$ resumes at $t_b=e_{p+1}$.
   Thus, by Claim \ref{claim:acc1d1}, $d$ and $b_1$ form an $\A$-configuration at $e_{p+1}$, which contradicts the $\A$-maximality of $\cP$.

    \item{Case 2c:}\label{it:case2c} $t'_a\neq\startTime{\cP}{a_{l-l'+1}}$ and $t'_a\neq\startTime{\cP}{a_{l-l'+2}}$.
    By definition of $t_a'$, we have $t'_a=t_a$ and $a_1$ resumes at $t_a$.
    Since $t_a$ is an event of $\cP$ some job, say $c$, completes at $e_p$.
    By Lemma~\ref{lem:how_jobs_finish}, job $c$ is spanning in block $p-1$.
    If another job completes at $e_p$, then we get a contradiction by Lemma~\ref{lem:ccd}.
    Hence, by Lemma~\ref{lem:interlace}, $\jobs(\xi_{p-1})\subseteq\{c,d,b_1,a_1\}$.
    Note that $z\leq y$, $l\geq l'$ and $t_a<t_b$ imply that $l'=l$.
    By definition of $t_a$, job $a_1$ is non-spanning in block $p-1$.
    By Lemma~\ref{lem:interlace}, $d$ and $b_1$ do not interlace, which implies $\xi_{p-1}(b_1)=0$.
    Therefore, $\xi_{p-1}(d)>0$.
    This allows us to obtain a contradiction by performing an analogous transformation as in Claim~\ref{claim:acc1d1}.
  \end{description}

    Observe that for the proof of Proposition~\ref{pro:no-A-configurations} it is crucial to show that $\cP$ and $\cP'$ have the same abnormality point. This needs to be proven in Case~1, Claim~\ref{claim:acc1d1}, \ref{it:case2a} and \ref{it:case2c}. In \ref{it:case2a} and \ref{it:case2c} the proof reduces to the proof for Case~1 and Claim~\ref{claim:acc1d1}. In Claim~\ref{claim:acc1d1} the proof also reduces to the proof for Case~1 but the $\xi_{p-1}(d)$ is new in definition of $\varepsilon$ as compared to Case~1 so we provide an appropriate comment about this $\xi_{p-1}(d)$ in Claim~\ref{claim:acc1d1}. Finally, in Case~1 we explicitly prove that $\cP$ and $\cP'$ have the same abnormality point.

\section{Alternating chains}
\label{subsec:alternating_chains}

In this section we prove that there always exists a normal schedule that is optimal for $P2|pmtn,in\textup{-}tree,r_j,p_j|\sum C_j$.
Our proof is by contradiction. We show that an abnormality point $i\neq \infty$ in a maximal schedule
implies that there is an alternating chain, see Section~\ref{sec:alt_chains:properties} for its definition, in the schedule.
Each job in that chain completes at the moment which is not $i$-normal.
This fact allows us to either make the alternating chain longer, which is shown in Section~\ref{sec:alt_chains:extending}, or find an
optimal schedule with an abnormality point higher than $i$.
Thus in either case we get a contradiction, in the former, since the number of jobs is finite and we cannot extend the chain
\emph{ad infinitum}, in the latter since the initial schedule is maximal.

\subsection{Basic definitions and properties}
\label{sec:alt_chains:properties}

Given schedule $(\cP,\vect{e},\vect{\xi})$, let $I=\{j,\ldots,j'\}\subseteq\{1,\ldots,q-1\}$. For two jobs $a$ and $a'$,
We say that $a$ \emph{covers $a'$ in $I$} if for each $t\in I$, $\xi_t(a')>0$ implies $\xi_t(a)=e_{t+1}-e_t$.
Let $\cP$ be a maximal schedule of abnormality point $i\neq\infty$.
Job sequence $(d_1,\ldots,d_l)$ is called an \emph{alternating chain} in $\cP$ if $d_1\in A_i(\cP)$ and $d_1$ executes
non-preemptively in $[e_{i+1},\complTime{\cP}{d_1}]$ and additionally, unless if $l=1$, the following conditions are satisfied:
\begin{enumerate} [label={\normalfont(C\arabic*)},leftmargin=*,itemindent=1em]
   \item\label{it:chain0} $A_i(\cP)=\{d_1,d_2\}$ and $\tau(d_1)=i+1$,
   \item\label{it:chain1} $\complTime{\cP}{d_j}<\complTime{\cP}{d_{j+1}}$ for each $j\in\{1,\ldots,l-1\}$.
   \item\label{it:chain2} the job $d_j$ executes non-preemptively in the interval  $[\complTime{\cP}{d_{j-2}},\complTime{\cP}{d_j}]$ for each $j\in\{2,\ldots,l\}$, where $\complTime{\cP}{d_0}=e_{i+1}$.
\end{enumerate}
If $(d_1,\ldots,d_l)$ ($l>2$) satisfies \ref{it:chain0}, \ref{it:chain2} and
\begin{enumerate} [label={\normalfont(C\arabic*')},leftmargin=*,itemindent=1em] \setcounter{enumi}{1}
   \item\label{it:chain1'} $\complTime{\cP}{d_j}<\complTime{\cP}{d_{j+1}}$ for each $j\in\{1,\ldots,l-2\}$ and $\complTime{\cP}{d_{l-1}}\geq\complTime{\cP}{d_l}$,
\end{enumerate}
then $(d_1,\ldots,d_l)$ is said to be \emph{almost alternating}.

\begin{lemma} \label{lem:alt_chain:no-almost}
If $\cP$ is a maximal schedule that is $\A$-free in $[t,\infty)$, then $\cP$ has no almost alternating chain $(d_1,\ldots,d_l)$
such that $l>2$ and $\complTime{\cP}{d_{l-1}}=t$.
\end{lemma}

\begin{proof}
Suppose that $(d_1,\ldots,d_l)$, $l>2$, is an almost alternating chain and $\complTime{\cP}{d_{l-1}}=t$.
By \ref{it:chain0}, it holds $\startTime{\cP}{d_1}<e_{i+1}$ and $\startTime{\cP}{d_2}<e_{i+1}$.
This, \ref{it:chain2} and \ref{it:chain1'} imply that $\startTime{\cP}{d_{l-1}}<e_{i+1}$ and $\startTime{\cP}{d_l}<e_{i+1}$.
Thus, $d_{l-2}$ and $d_l$ form an $\A$-configuration at $t$ in $\cP$, which contradicts that $\cP$ is $\A$-free in $[t,\infty)$.
\end{proof}

The following lemma excludes almost alternating chains with simultaneous completions of $d_{l-1}$ and $d_l$ from maximal schedules. The lemma does not require that the maximal schedules are $\A$-free.

\begin{lemma} \label{lem:alt_chain:normal}
Let $\cP$ be a maximal schedule of abnormality point $i\neq\infty$.
There exists no alternating chain with at least two jobs in which the last two jobs of the chain complete at the same time.
\end{lemma}

\begin{proof}
Suppose for a contradiction that $(d_1,\ldots,d_l)$ is an alternating chain with $\complTime{\cP}{d_{l-1}}=\complTime{\cP}{d_l}$.
Let first $l>3$ and let $U$ be the set of odd indices in $\{3,\ldots,l\}$. Denote by $\vect{\xi}$ the partition of $\cP$.
We fist prove that the total length of $d_2$ executed in $[\complTime{\cP}{d_1},\complTime{\cP}{d_l}]$, namely $\xi_{i+2}(d_2)$,
is $i$-normal. From the definition of alternating chain we know that there is no idle time in this interval and only the jobs from the
chain execute in it. Hence,
\begin{eqnarray*}
\xi_{i+2}(d_2) & = & 2(\complTime{\cP}{d_l}-\complTime{\cP}{d_1}) \\
               &   & -\sum_{j=3}^l(\complTime{\cP}{d_j}-\complTime{\cP}{d_{j-2}}) \\
               & = & 2\sum_{j\in U}\sum_{j'\geq i}\xi_{j'}(d_j) - \sum_{j=3}^l\sum_{j'\geq i}\xi_{j'}(d_j).
\end{eqnarray*}
By Lemma~\ref{lem:remaining_part_i-normal}, $\sum_{j'=1}^i\xi_{j'}(d_j)$ is $i$-normal for each $j\in\{3,\ldots,l\}$.
Therefore, $\xi_{i+2}(d_2)$ is $i$-normal.
This implies, by Lemmas~\ref{lem:remaining_part_i-normal} and~\ref{lem:all_e_i's_are_normal}, that the three following numbers are $i$-normal:
\[
\xi_i(d_1)+\xi_i(d_2),\quad \xi_i(d_1)+\xi_{i+1}(d_1),\quad \xi_i(d_2)+\xi_{i+1}(d_1)
\]
since $\xi_{i+1}(d_1)=\xi_{i+1}(d_2)$. Therefore, $\xi_i(d_1)$ and $\xi_i(d_2)$ are $(i+1)$-normal, and since $e_{i+1}$ is
$i$-normal due to Lemma~\ref{lem:how_jobs_finish}, $i$ is not the abnormality point of $\cP$ --- a contradiction.

Now consider $l=2$. Let $\vect{e}$ be the events of $\cP$. Denote $x=\xi_{i+1}(d_1)=\xi_{i+1}(d_2)$.
By assumption and by definition of alternating chain, $d_1$ and $d_2$ complete at $e_{i+2}$ and hence $\xi_i(d_1)+x=s_1/2^i$ and
$\xi_i(d_2)+x=s_2/2^i$ for some integers $s_1$ and $s_2$.
By definition of alternating chain, $A_i(\cP)=\{d_1,d_2\}$ and hence $x$ is not $(i+1)$-normal.
Thus, $x=s'/2^{i+1}+\varepsilon$ for some $0<\varepsilon<1/2^{i+1}$.
Then, $\xi_i(d_1)=(s_1-s')/2^i-\varepsilon$ and $\xi_i(d_2)=(s_2-s')/2^i-\varepsilon$.
By Lemma~\ref{lem:all_e_i's_are_normal}, $e_{i+1}-e_{i}$ is $i$-normal. By Lemma~\ref{lem:at_least_two_abnormal},
\[
e_{i+1}-e_{i}=\xi_i(d_1)+\xi_i(d_2)=(s_1+s_2-2s')/2^i-2\varepsilon.
\]
Thus, $2\varepsilon$ is $i$-normal, which implies that $\varepsilon$ is $(i+1)$-normal --- a contradiction with
$0<\varepsilon<1/2^{i+1}$.
\end{proof}

\begin{lemma} \label{lem:alt_chain:not_normal}
Let $\cP$ be a maximal schedule of abnormality point $i\neq\infty$. If $(d_1,\ldots,d_l)$ ($l\geq 1$) is an alternating
chain in $\cP$, then $\complTime{\cP}{d_j}$ is not $(i+1)$-normal for each $j\in\{1,\ldots,l\}$.
\end{lemma}

\begin{proof}
Let first $1\leq j\leq\min\{2,l\}$.
Then, $\complTime{\cP}{d_j} = e_{i+1}+1-\xi_i(d_j)-\sum_{t<i}\xi_t(d_j)$.
By Lemma~\ref{lem:remaining_part_i-normal}, the latter sum is $i$-normal.
By Observations~\ref{obs:normal_numbers} and~\ref{obs:abnormality_point} and by Lemma~\ref{lem:all_e_i's_are_normal},
$e_{i+1}$ is $i$-normal. However, $\xi_i(d_j)$ is not $(i+1)$-normal because $d_j\in A_i(\cP)$ and hence, again by
Observation~\ref{obs:normal_numbers}, $\complTime{\cP}{d_j}$ is not $(i+1)$-normal.

For $j>2$, if $j$ is even (respectively, odd), then let $U$ be the set of even (respectively, odd) integers in $\{1,\ldots,j-1\}$.
Let $u=1$ if $j$ is odd and let $u=2$ if $j$ is even. Then,
\begin{eqnarray*}
 \complTime{\cP}{d_j}&=e_{i+1} + \sum_{j'\in U}\left(1-\xi_i(d_{j'})-\sum_{t<i}\xi_t(d_{j'})\right) \\
                     &=e_{i+1} -\xi_i(d_u) + \sum_{j'\in U}\left(1-\sum_{t<i}\xi_t(d_{j'})\right).
\end{eqnarray*}
Again, by Lemma~\ref{lem:remaining_part_i-normal}, $\xi_i(d_u)$ is the only term in the above expression that is not $(i+1)$-normal.
Thus, $\complTime{\cP}{d_j}$ is not $(i+1)$-normal.
\end{proof}

\subsection{Transformations using alternating chains}
\label{sec:alt_chains:trans}

Consider a schedule $\cP$ of abnormality point $i$ and an alternating chain $(d_1,\ldots,d_l)$ ($l>1$) and
$\jobs(\xi_{\tau (d_l)})=\{x,y,d_l\}$. Let $u=2$ if $l$ is odd, and $u=1$ if $l$ is even.
Let $\varepsilon>0$ be the largest $\varepsilon$ such that
\begin{equation} \label{eq:epsilon}
 \begin{split}
\varepsilon\leq \alpha=\xi_i(d_u)\ \land \varepsilon\leq \beta= \frac{1}{2}\min\left\{e_{\tau(d_{j})+1}-e_{\tau(d_j)}\st j\in U\right\}\\
 \land\ \varepsilon\leq\gamma=\min\left\{e_{\tau(d_{j-1})}-r(d_{j})\st j\in U\setminus\{1,2\}\right\},
 \end{split}
\end{equation}
where $U$ is the set of the indices in $\{1,\ldots,l\}$ having the same parity as $l$, and
\begin{equation} \label{eq:epsilon1}
\varepsilon\leq \min\{\xi_{\tau (d_l)}(x), \xi_{\tau (d_l)}(y)\}.
\end{equation}

We define a transformation of \emph{$\varepsilon$-pushing of $d_l$} that produces
a schedule $\cP'$ as follows (see Figure~\ref{fig:altchain-pushing} for an illustration):
\begin{itemize}
 \item the schedules $\cP$ and $\cP'$ are identical in time intervals $[0,e_j]$ and $[\complTime{\cP}{d_l},\infty)$.
 \item
   To obtain the part of $\cP'$ in $[e_i,e_{i+1}]$, we increase (with respect to $\cP$) the amount of $d_{3-u}$ by $\varepsilon$ and decrease the amount of $d_u$ by $\varepsilon$.
   Then, a part of job $d_{3-u}$ executes in $[e_{i+1},\complTime{\cP}{d_{3-u}}-\varepsilon]$ and a part of job $d_u$ executes in $[e_{i+1},\complTime{\cP}{d_u}+\varepsilon]$. This in particular characterizes the execution of $d_1$ and $d_2$ in $\cP'$.
 \item
  For each $j\in U\setminus\{1,2\}$, the part of $d_j$ that executes in $[\complTime{\cP}{d_{j-2}},\complTime{\cP}{d_j}]$ in $\cP$ is executed in $[\complTime{\cP}{d_{j-2}}-\varepsilon,\complTime{\cP}{d_j}-\varepsilon]$ in $\cP'$. In this way we ensure that each job $d_j$, $j\in U$, completes $\varepsilon$ earlier in $\cP'$ than in $\cP$.
 \item
  For each $j\in \{3,\ldots,l\}\setminus U$, the part of $d_j$ that executes in $[\complTime{\cP}{d_{j-2}},\complTime{\cP}{d_j}]$ in $\cP$ is executed in $[\complTime{\cP}{d_{j-2}}+\varepsilon,\complTime{\cP}{d_j}+\varepsilon]$ in $\cP'$. In this way we ensure that each job $d_j$, $j\notin U$, completes $\varepsilon$ later in $\cP'$ than in $\cP$.
 \item
  Finally, the two jobs $x$ and $y$ are executed in the remaining free space in $[\complTime{\cP'}{d_{l-1}},\complTime{\cP}{d_l}]$ on one machine and in $[\complTime{\cP'}{d_l},\complTime{\cP}{d_l}]$ on the other machine.
\end{itemize}

    \begin{figure*}[htb]
    \begin{center}
    \includegraphics[scale=1.0]{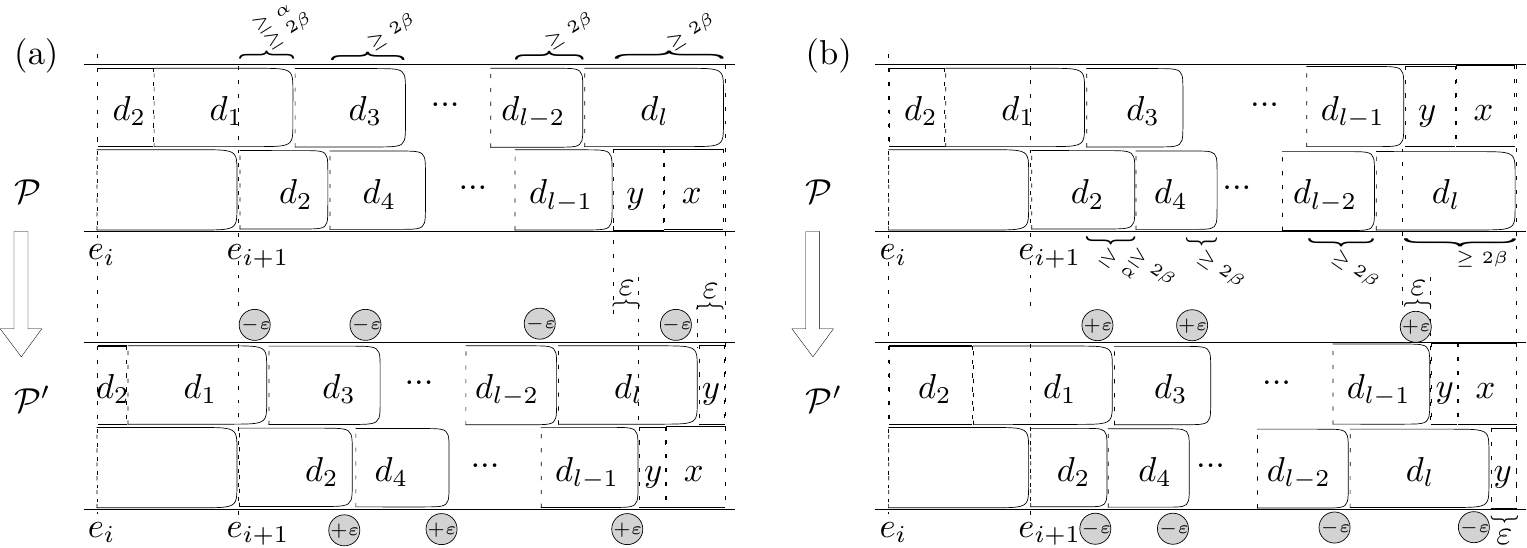}
    \caption{$\varepsilon$-pushing of $d_l$ when: (a) $l$ is odd; (b) $l$ is even}
    \label{fig:altchain-pushing}
    \end{center}
    \end{figure*}

The transformation of $\varepsilon$-pushing of $d_l$ will be a key transformation used to extend an alternating chain of a maximal schedule $\cP$ in the proof of Proposition~\ref{pro:infinite_chain} --- the main result of the next section. The extension, as alluded earlier, requires that $\cP$ is $\A$-free. Actually, it suffices that $\cP$ is $\A$-free in the interval that starts with the completion of $d_l$, the last job of the chain. However, since the $\varepsilon$-pushing of $d_l$ may change $\cP$ itself we need to prove that the resulting schedule is $\A$-free in the interval that starts with the completion of $d_l$, which the transformation may have changed, in order to unable further extensions of the chain. Thus, we need the following lemma.

\begin{lemma} \label{lem:partially-free}
Suppose that $(d_1,\ldots,d_l)$ (where $l>1$) is an alternating chain in a maximal schedule $\cP$ that is $\A$-free in
$[\complTime{\cP}{d_l},\infty)$, and $\jobs(\xi_{\tau (d_l)})=\{x,y,d_l\}$.
If $\varepsilon$ is selected as in \eqref{eq:epsilon} and \eqref{eq:epsilon1}, then the schedule $\cP'$ obtained from $\cP$ by $\varepsilon$-pushing of $d_l$ is maximal, $\A$-free in $[\complTime{\cP'}{d_l},\infty)$, and $(d_1,\ldots,d_l)$ is an alternating chain in $\cP'$.
\end{lemma}

\begin{proof}
An  \emph{$\varepsilon$}-pushing of $d_l$ results in a feasible schedule $\cP'$ (note that at most one of $d_1$ and $d_2$ can have release date in $[e_i,e_{i+1}]$) with the total completion time not exceeding that of $\cP$.
Thus, $\cP'$ is optimal.
Note that an odd $l$ would results in $\cP'$ having smaller total completion time than that of $\cP$.
Thus, $l$ is even.
If the $\varepsilon$ makes at least one of the tree inequalities in (\ref{eq:epsilon}) an equality, then the $i$-th block becomes $i$-normal and we get a contradiction in case of a maximal $\cP$.
On the other hand, if an $\varepsilon$ makes all three inequalities in (\ref{eq:epsilon}) holding strict, then $(d_1,\ldots,d_l)$, $l>1$, is an alternating chain in $\cP'$.

We prove, by contradiction, that the schedule $\cP'$ is $\A$-free in $[\complTime{\cP'}{d_l},\infty)$.
Suppose that some jobs $a$ and $b$ form an $\A$-configuration at a point $t\geq\complTime{\cP'}{d_l}$ in $\cP'$.
Note that $t\geq \complTime{\cP}{d_l}$ is not possible because $\cP$ and $\cP'$ are identical from $\complTime{\cP}{d_l}$ on and $\cP$ is $\A$-free in $[\complTime{\cP}{d_l},\infty)$ by assumption.
Thus, $t=\complTime{\cP'}{d_l}$.
Therefore, $a=d_{l}$ and $b\in\{x,y\}$.

Let $\lambda\in [0,\complTime{\cP'}{b}-\complTime{\cP'}{d_l}]$ be a maximal real number such that for each $\varepsilon'\in[0,\lambda)$ there exists a feasible schedule $\cP_{\varepsilon'}$ such that $(d_1,\ldots,d_l)$ is an alternating chain in $\cP_{\varepsilon'}$ and $\varepsilon'$-pushing of $d_l$ in $\cP_{\varepsilon'}$ results in $\cP$.
Since $l$ is even, the total completion time of $\cP_{\varepsilon'}$ is the same as the total completion time of $\cP'$ for each $\varepsilon'\in[0,\lambda)$.
Informally speaking, $\cP_{\varepsilon'}$ is obtained by performing a modification that is `opposite' to pushing of $d_l$.
By definition of $\A$-configuration, $d_l$ is independent of any job that executes in the interval $(\complTime{\cP'}{d_l},\complTime{\cP'}{b})$ in $\cP'$.
Thus, the maximality of $\lambda$ implies that taking $\varepsilon'=\lambda$ would result in a schedule $\cP_{\varepsilon'}$ in which one of the following holds:
\begin{itemize}
 \item Job sequence $(d_1,\ldots,d_l)$ is not an alternating chain in $\cP_{\varepsilon'}$.
       Then we have two possibilities. The first possibility is that $t=\complTime{\cP_{\varepsilon'}}{d_j}=\complTime{\cP_{\varepsilon'}}{d_{j+1}}$ for some $j\in\{1,\ldots,l-2\}$.
       Then, $(d_1,\ldots,d_{j+1})$ is an alternating chain in which the two last jobs complete at the same time --- a contradiction with Lemma~\ref{lem:alt_chain:normal}.
       The second possibility is that either $d_1$ or $d_2$ is not present in the $i$-th block of $\cP_{\varepsilon'}$.
       Then, the abnormality point of $\cP_{\varepsilon'}$ is greater than $i$ --- a contradiction with the maximality of $\cP$.
 \item $\complTime{\cP_{\varepsilon'}}{d_l}=\complTime{\cP}{b}$.
       This would imply that $\complTime{\cP_{\varepsilon'}}{b}<\complTime{\cP}{b}$ and this is not possible due to the optimality of $\cP$.
 \item $\startTime{\cP_{\varepsilon'}}{d_j}=r(d_j)$ for some $j\in\{3,\ldots,l-1\}$.
       In this case a contradiction follows from Lemma~\ref{lem:alt_chain:not_normal}.
 \item $\startTime{\cP_{\varepsilon'}}{b}=r(b)$.
       Since $\complTime{\cP_{\varepsilon'}}{d_{l-1}}=\startTime{\cP_{\varepsilon'}}{b}$, we again obtain a contradiction with
       Lemma~\ref{lem:alt_chain:not_normal}.
\end{itemize}
Therefore, the lemma is proved.
\end{proof}

Finally, we observe that the \emph{$\varepsilon$}-pushing of $d_l$ can readily be extended to the case where one of the jobs $x$ or $y$ starts in $(\complTime{\cP}{d_{l-1}},\complTime{\cP}{d_{l}})$ but neither of them completes in that interval.

\subsection{Extending an alternating chain}
\label{sec:alt_chains:extending}

We first prove that a single-job alternating chain is present in each maximal (and thus in $\A$-maximal) schedule
of abnormality point $i\neq\infty$.

\begin{lemma} \label{lem:b+c_form_chain}
If $\cP$ is a maximal schedule of abnormality point $i\neq\infty$, then a job in $A_i(\cP)$ with minimum completion time forms an
alternating chain in $\cP$.
\end{lemma}

\begin{proof}
Suppose that schedule $(\cP,\vect{e},\vect{\xi})$ is maximal. According to Observations~\ref{obs:normal_numbers} and
\ref{obs:abnormality_point} and Lemma~\ref{lem:all_e_i's_are_normal}, $e_{i+1}$ is $i$-normal.
By Lemma~\ref{lem:at_least_two_abnormal}, we have $|A_i(\cP)|=2$.
Let $A_i(\cP)=\{b,c\}$, where $\complTime{\cP}{b}\leq\complTime{\cP}{c}$.
Denote $I=\{i+1,\ldots,\tau(b)\}$.
Note that, by Lemma~\ref{lem:remaining_part_i-normal}, $I\neq\emptyset$.

We prove the lemma by contradiction.
More precisely, the assumption that $b$ does not form an alternating chain in $\cP$ allows us to conclude that $\cP$ is not
maximal. We may assume without loss of generality that $c$ covers $b$ in $I$.
Indeed, if this is not the case, then we transform $\cP$ as follows.
Let $t\in I$ be such that $\xi_t(b)>0$ and job $c$ is non-spanning in block $t$.
Take $\varepsilon=\min\{\xi_t(b),\xi_i(c),e_{t+1}-e_t-\xi_t(c)\}$.
Note that $\varepsilon>0$ and, by Lemmas~\ref{lem:no_idle-time_between_preemptions} and
\ref{lem:one_dominates}, $\xi_i(c)=e_{i+1}-e_i-\xi_i(b)$.
The schedule obtained by a transformation
\[
(\vect{e}',\vect{\xi}')=\cyclicshift{ \vect{e},\vect{\xi},\varepsilon,(t\cyclic{b} i\cyclic{c} t) }
\]
has the same total completion time and the same events as $\cP$ and either:
  $\xi_i'(b)=e_{i+1}-e_i$ (which happens when $\varepsilon=\xi_i(c)$); or
  $\xi_t'(b)=0$ (which happens when $\varepsilon=\xi_t(b)$); or
  $\xi_t'(c)=e_{t+1}-e_t$ (which happens when $\varepsilon=e_{t+1}-e_t-\xi_t(c)$).
In the former case we would obtain a schedule with abnormality point greater than $i$, which is not possible due to the maximality of $\cP$.
After repeating this transformation as long as $c$ does not cover $b$ in $I$ we obtain the desired schedule.

Now we prove that $\tau(b)=i+1$. Suppose for a contradiction that $\tau(b)>i+1$. Thus, since $c$ covers $b$ in $I$,
and $e_{\tau(b)}$ is an event of $\cP$, we have $\xi_{\tau(b)-1}(b)=\xi_{\tau(b)-1}(c)=0$ and there exists
$a\in\jobs\setminus\{b,c\}$ such that $\complTime{\cP}{a}=e_{\tau(b)}$.
Find the maximum $j$, $j<\tau(b)-1$, such that $\xi_j(b)\neq 0$. Note that $j\geq i$.
Since job $c$ covers $b$ in $I$, $c$ is spanning in block $j$. By Lemma~\ref{lem:one_dominates}, $a\notin\jobs(\xi_j)$.
Lemma~\ref{lem:between_preemptions} applied to $a=b$, $a'=a$, $j$ and $j'=\tau(b)$ gives $\tau(b)>j'=\tau(b)$ --- a contradiction.
This proves that $(b)$ is an alternating chain in $\cP$.
\end{proof}

\begin{lemma} \label{lem:shifting}
If $(d_1,\ldots,d_l)$ is an alternating chain in a maximal schedule $\cP$, then there is no idle time in block $(\tau(d_{l-1})+1)$
of $\cP$.
\end{lemma}

\begin{proof}
Let $\vect{e}$ be the $q$ events of $\cP$ and let $\vect{\xi}$ be its partition.
Let $i$ be the abnormality point of $\cP$.
Since $\cP$ has an alternating chain, $i\neq\infty$.
Suppose for a contradiction that there is idle time in the $(\tau(d_{l-1})+1)$-st block of $\cP$.
By Lemma~\ref{lem:how_jobs_finish}, at most one job completes in the $(\tau(d_{l-1})+1)$-st block of $\cP$.
By Lemma~\ref{lem:no_idle-time_between_preemptions}, no job that does not complete in the $(\tau(d_{l-1})+1)$-st block can be present in this block.
Thus, $d_l$ is the only job in block $(\tau(d_{l-1})+1)$ and the total length of the idle time is $x=e_{\tau(d_{l-1})+2}-e_{\tau(d_{l-1})+1}$.
Construct a schedule $\cP'$ by performing an $\varepsilon$-pushing of $d_l$ in $\cP$ with
  \[\varepsilon=\min\left\{\alpha, \beta, \gamma,  x/2\right\}.\]
Denote the resulting schedule by $\cP'$.
To complete the proof we observe that $\complTime{\cP'}{d_j}=r(a)$ for some job $a$ and some $j\in\{3,\ldots,l\}$ (when $\varepsilon=\gamma$) or $\complTime{\cP'}{d_{j-1}}=\complTime{\cP'}{d_j}$ for some $j\in\{2,\ldots,l\}$ (when $\varepsilon\in\{\beta,x/2\}$) or there is no $d_j$ in the $i$-th block of $\cP'$ for some $j\in\{1,2\}$ (when $\varepsilon=\alpha$).
Therefore, the choice of $\varepsilon$ always results in $\cP'$ that has all blocks $j$, $j\in\{1,\ldots,i\}$, being $j$-normal --- a contradiction with the lemma assumption that $\cP$ is maximal.
\end{proof}

The next lemma states that if a maximal schedule $\cP$ with an alternating chain $(d_1,\ldots,d_l)$ has no $\A$-configuration
at $t\geq \complTime{\cP}{d_{l-1}}$, then there exists another maximal schedule $\cP'$ with longer alternating chain
$(d_1,\ldots,d_l,d_{l+1})$ with no $\A$-configuration at $t\geq\complTime{\cP'}{d_l}$.

\begin{proposition} \label{pro:infinite_chain}
Let $\cP$ be a maximal schedule of abnormality point $i\neq\infty$.
If $(d_1,\ldots,d_l)$, $l\geq 1$, is an alternating chain in $\cP$ and $\cP$ is $\A$-free in $[\complTime{\cP}{d_{l-1}},\infty)$, where $\complTime{\cP}{d_0}$ is the $(i+1)$-st event of $\cP$, then there exists a job $d_{l+1}$ such that $(d_1,\ldots,d_l,d_{l+1})$ is an alternating chain in some maximal schedule $\cP'$ that is $\A$-free in $[\complTime{\cP'}{d_{l}},\infty)$.
\end{proposition}

We leave the proof of the proposition to the end of the section.
Proposition~\ref{pro:no-A-configurations} guarantees that a maximal $\A$-free schedule exists for in-trees.
If this schedule is not normal, then we have a single-job alternating chain in it by Lemma~\ref{lem:b+c_form_chain}.
Proposition~\ref{pro:infinite_chain} guarantees that the alternating chain can be always extended by one job. However,
the process of extending the alternating chain may result in a schedule which is not $\A$-free in general ---
the source of this lies in Lemma~\ref{lem:partially-free}. Luckily, we do not need the resulting schedule to be $\A$-free ---
it suffices that it has no $\A$-configuration at a completion of $d_l$, the last job from the alternating chain, or later ---
see the assumptions of Proposition~\ref{pro:infinite_chain}. The above gives a sketch of the proof of the following theorem.

\begin{theorem} \label{thm:normal_exist}
There exists a normal optimal schedule for each instance of problem $P2|pmtn,in\textup{-}tree,r_j,p_j|\sum C_j$.
\end{theorem}

\begin{proof}
Take a maximal schedule $\cP$ and suppose for a contradiction that $\cP$ is not normal. Let $i\neq\infty$ be its abnormality point.
By Lemma~\ref{lem:b+c_form_chain}, $A_i(\cP)=\{d_1,d_2\}$, and $(d_1)$ is an alternating chain in $\cP$.
Next, by Proposition~\ref{pro:no-A-configurations}, there is a maximal schedule $\cP^1$ with its abnormality point $i$ and alternating chain $(d_1)$ which is also $\A$-free.
Thus, in particular, $\cP^1$ is $\A$-free in $[x,\infty)$, where $x$ is the $(i+1)$-st event of $\cP^1$.
Finally, by Proposition~\ref{pro:infinite_chain} and a simple inductive argument, there exists a schedule $\cP^{n+1}$ with an alternating chain of $l=n+1$ jobs, contradicting the fact that the number of jobs equals $n$.
\end{proof}

\begin{corollary} \label{cor:main}
For the given set of $n$ jobs, there exists an optimal schedule for $P2|pmtn,in\textup{-}tree,r_j,p_j|\sum C_j$ such that each job
start, preemption or completion occurs at a time point that is a multiple of $1/2^{2n}$. \qed
\end{corollary}

\subsubsection*{Proof of Proposition~\ref{pro:infinite_chain}}

By Lemma~\ref{lem:at_least_two_abnormal}, $|A_i(\cP)|=2$.
If $l=1$, then take $d_{l+1}$ to be the job in $A_i(\cP)\setminus\{d_1\}$.
If $l>1$, then by Lemma~\ref{lem:shifting}, there is no idle time in $(\tau(d_{l-1})+1)$-st block of $\cP$.
Thus, $|\jobs(\xi_{\tau(d_{l-1})+1})|>1$.
Take $d_{l+1}$ to be a job in $\jobs(\xi_{\tau(d_{l-1})+1})\setminus\{d_l\}$ that starts or resumes at $\complTime{\cP}{d_{l-1}}$.

We will perform several schedule modifications leading to some maximal and $(\A,[\complTime{\cP'}{d_l},\infty))$-free schedule $\cP'$ with an alternating chain $(d_1,\ldots,d_{l+1})$.
We point out that some steps of the proof redefine the job $d_{l+1}$ selected above.
By assumption, $d_1\in A_i(\cP)$, and by the choice of $d_2$, $d_2\in A_i(\cP)$.
Thus, \ref{it:chain0} follows for $(d_1,\ldots,d_{l+1})$.

We now prove that
\begin{equation} \label{eq:extend_chain1}
\complTime{\cP}{d_{l+1}}>\complTime{\cP}{d_l}.
\end{equation}
Note that if $l=1$, then \eqref{eq:extend_chain1} follows from Lemma~\ref{lem:b+c_form_chain} and from the choice of $d_1$ and $d_2$. Thus, $l\geq 2$ from now on.

We begin by proving, by contradiction, that $d_{l+1}$ executes non-preemptively in $[\complTime{\cP}{d_{l-1}},\complTime{\cP}{d_{l}}]$.
First, we observe that no job, except for $d_{l}$, completes in the interval $(\complTime{\cP}{d_{l-1}},\complTime{\cP}{d_{l}}]$ for otherwise job $d_{l+1}$
must complete in $(\complTime{\cP}{d_{l-1}},\complTime{\cP}{d_{l}}]$ and thus jobs $d_{l-1}$ and $d_{l+1}$ form an $\A$-configuration in $\cP$ --- a contradiction since $\cP$ is $\A$-free in $[\complTime{\cP}{d_{l-1}},\infty)$.

Second, $|J\setminus\{d_l,d_{l+1}\}|\leq 1$, where $J$ is the set of jobs executed in $(\complTime{\cP}{d_{l-1}},\complTime{\cP}{d_{l}}]$.
Otherwise, there is a pair of jobs in $\{x,y,d_l\}$, where $\{x,y\} \subseteq J\setminus\{d_l,d_{l+1}\}$, that interlace --- contradiction by Lemma \ref{lem:interlace} and by the fact that $d_l$ is preempted in $(\complTime{\cP}{d_{l-1}},\complTime{\cP}{d_{l}}]$.
Indeed, this pair of job consists of a job $z\in\{x,y,d_l\}$ with minimum completion time among those three jobs, and a job in $\{x,y,d_l\}\setminus\{z\}$ that is non-spanning in block $\tau(z)$.

Finally,  we show that without loss of generality $J=\{d_l,d_{l+1}\}$.
Suppose otherwise, i.e., $J=\{x,d_l,d_{l+1}\}$. The \emph{$\varepsilon$}-pushing of $d_l$ with
\begin{eqnarray*}
 \varepsilon & = &\min\big\{\alpha, \beta, \gamma, \\
             &   &\complTime{\cP}{d_{l}}-\complTime{\cP}{d_{l-1}}-\xi_{\tau(d_{l-1})+1}(d_{l+1}), \xi_{\tau(d_{l-1})+1}(d_{l+1})\big\}
\end{eqnarray*}
results in a schedule $\cP'$ that either has all blocks $j$, $j\in\{1,\ldots,i\}$,  being $j$-normal (this happens when $\varepsilon\in\{\alpha,\beta,\gamma\}$) --- a contradiction with the lemma assumption that $\cP$ is maximal; or it has $d_l$ and $d_{l+1}$ as the only two jobs executed in
$[\complTime{\cP'}{d_{l-1}},\complTime{\cP'}{d_{l}}]$; or it has $d_l$ and $x$ as the only two jobs executed in $[\complTime{\cP'}{d_{l-1}},\complTime{\cP'}{d_{l}}]$.
(See Figure~\ref{fig:alt-chains1}(a) for this transformation when $\varepsilon=\complTime{\cP}{d_{l}}-\complTime{\cP}{d_{l-1}}-\xi_{\tau(d_{l-1})+1}(d_{l+1})$.)
In the latter case, i.e., when $\varepsilon=\xi_{\tau(d_{l-1})+1}(d_{l+1})$, we take $x$ as $d_{l+1}$ from now on.
The schedule $\cP'$ is maximal, by Lemma~\ref{lem:partially-free} it is $\A$-free in $[\complTime{\cP}{d_{l}},\infty)$, and $(d_1,\ldots,d_{l})$ is an alternating chain in $\cP'$.
Thus, without loss of generality we can take $\cP$ as being $\cP'$ from now on.
Then, we have that $d_{l+1}$ executes non-preemptively in $[\complTime{\cP}{d_{l-1}},\complTime{\cP}{d_{l}}]$ and, by Lemma~\ref{lem:alt_chain:normal}, \eqref{eq:extend_chain1} holds.
Thus, \ref{it:chain1} holds for $(d_1,\ldots,d_{l+1})$ in $\cP$.

    \begin{figure}[htb]
    \begin{center}
    \includegraphics[scale=1.0]{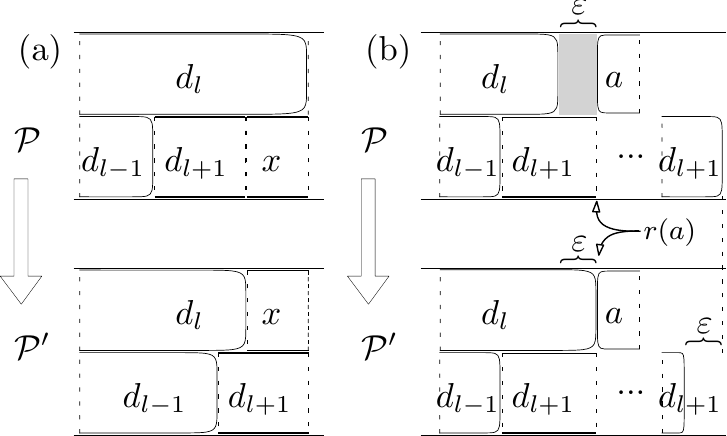}
    \caption{Schedule transformations in the proof of Proposition~\ref{pro:infinite_chain}}
    \label{fig:alt-chains1}
    \end{center}
    \end{figure}

We argue that there is no idle time in the $(\tau(d_l)+1)$-st block of $\cP$.
Our argument is by contradiction.
If there is idle time in the block, then by (\ref{eq:extend_chain1}), $d_{l+1}$ is the only job there.
Thus, $d_{l+1}$ completes at $e_{\tau(d_l)+2}$ or there is a release date pinned job $a$ that starts at $e_{\tau(d_l)+2}$, i.e., $r(a)=e_{\tau(d_l)+2}$.
In the former case we have an alternating chain $(d_1,\ldots,d_{l+1})$ with idle time in block $\tau(d_l)+1$ which contradicts Lemma~\ref{lem:shifting} and completes the proof.
In the latter case we use an extended \emph{$\varepsilon$}-pushing of $d_{l+1}$ (the operation of the $\varepsilon$-pushing can be generalized in a straightforward way to the case when $d_{l+1}$ is preempted in $[\complTime{\cP}{d_{l}},\complTime{\cP}{d_{l+1}}]$, see also Figure~\ref{fig:alt-chains1}(b) --- we omit a formal definition of the extended pushing) with
\[
\varepsilon=\min\left\{\alpha, \beta, \gamma, r(a)-\complTime{\cP}{d_{l}}\right\},
\]
which results in a schedule $\cP'$ that has all blocks $j$, $j\in\{1,\ldots,i\}$, being $j$-normal --- a contradiction with the lemma assumption that $\cP$ is maximal.  (See Figure~\ref{fig:alt-chains1}(b) for an illustration of this schedule transformation when $\varepsilon=r(a)-\complTime{\cP}{d_l}$.)
Therefore, without loss of generality we may assume there is no idle time block $\tau(d_l)+1$, and thus some job $a\neq d_{l+1}$ starts or resumes at $e_{\tau(d_l)+1}$.

\medskip
We next describe a finite iterative process that starts with $\cP$, produces a schedule $\cP_u$ in its $u$-th iteration, $u\geq 1$, and stops after $T\geq 1$ iterations.
We then show that if $T=1$, then the schedule $\cP_1$ satisfies the conditions of the lemma. However, if $T>1$, then we prove that there is a pair of jobs $x$ and $y$ that allows the iterative process to construct maximal schedules $\cP_1$, ..., $\cP_{T-1}$ each with alternating chain $(d_1,\ldots,d_l)$, yet at the same time the pair prevents $d_{l+1}$ from satisfying \ref{it:chain2} in $\cP_1$, \ldots, $\cP_{T-1}$.
However, an exit after $T>1$ iterations is only possible through a schedule $\cP_T$ such that the total completion time of $\cP_T$ is smaller than that of $\cP$ or both schedules have the same total completion times but the abnormality point of $\cP_T$ is greater than $i$, which contradicts the maximality of $\cP$. Therefore the iterative process must exit after exactly one iteration producing the desired schedule $\cP_1$ -- more than one  iteration leads to a contradiction.

In order to describe the iterative process formally, we introduce a key definition and related notation.
We say that a schedule $(\cP_u, \vect{\xi}^u, \vect{e}^u)$ with $\tau_u=\tau_{\cP_{u}}$ is \emph{$d_{l+1}$-preempted} if there
exists a pair of jobs $x$ and $y$ such that the following conditions are satisfied:
\begin{enumerate} [itemsep=0pt,topsep=0pt,label={\normalfont(I\arabic*)},leftmargin=*]
 \item\label{it:I1}\label{it:I:first}
      $\cP_u$ is maximal, $(d_1,\ldots,d_l)$ is an alternating chain in $\cP_u$, $i$ is the abnormality point of $\cP_u$, and
      $d_l$ and $d_{l+1}$ are the only two jobs executed in $[\complTime{\cP_u}{d_{l-1}},\complTime{\cP_u}{d_l}]$, and
      $\complTime{\cP_u}{d_{l+1}}>\complTime{\cP_u}{d_l}$;
 \item\label{it:I2}
      Some job $a_u$ starts or resumes at $\complTime{\cP_u}{d_l}$;
 \item\label{it:I3}
      $d_{l+1}$ covers each job in $\{a_1,\ldots,a_u\}$ in $I_u=\{\tau_u(d_l)+1,\ldots,k_u\}$, where
      $k_u=\min\{\tau_u(a_u),\tau_{u}(d_{l+1})\}$;
 \item\label{it:I4}
      $\complTime{\cP_u}{a_u}>\complTime{\cP_u}{d_{l+1}}$;
 \item\label{it:I5}\label{it:I:last}
      There exists $j_u<\tau_{u}(d_{l+1})$ such that $\jobs(\xi^u_{j_u})=\{x,y\}$ and $\min\{r(x),r(y)\}>\complTime{\cP_u}{d_l}$.
\end{enumerate}

\medskip
In the following we prove that all schedules $\cP_1,\ldots,\cP_{T-1}$ are $d_{l+1}$-preempted when $T>1$.
Let initially $u=1$, and the iterative process is as follows. ($\cP_0$ refers to $\cP$.)

\paragraph{Step 1: Moving $a_{u-1}$ from block $\tau_{u-1}(d_l)+1$ to block $k_{u-1}$.}

If $u=1$, then let $\cP_{u}'=\cP$ and go to Step~2.
If $u>1$, then we construct $\cP_{u}'$ by an extended $\varepsilon$-pushing of $d_{l+1}$ in $\cP_{u-1}$ and moving a piece of $a_{u-1}$ of length $\varepsilon$ from block $\tau_{u-1}(d_l)+1$ to block $k_{u-1}$, where
\begin{eqnarray*}
 \varepsilon = & \min\big\{ & \alpha, \beta, \gamma, \lceil \complTime{\cP_{u-1}}{d_l} \rceil -\complTime{\cP_{u-1}}{d_l},\\
            &  & \xi^{u-1}_{\tau_{u-1}(d_l)+1}(a_{u-1}), e^{u-1}_{k_{u-1}+1}-e^{u-1}_{k_{u-1}}-\xi^{u-1}_{k_{u-1}}(a_{u-1}) \big\}.
\end{eqnarray*}
Denote by $\vect{e}'$ and $\vect{\xi}'$ the events and the partition of $\cP_{u}'$, respectively.
Let for brevity $\tau_{\cP_{u}'}=\tau'$.
Figure~\ref{fig:pushing} depicts the transition from $\cP_{u-1}$ to $\cP_{u}'$ for $u>1$. Note that $\cP_{u}'$ and $\cP_{u-1}$
have the same total completion times.

    \begin{figure}[htb]
    \begin{center}
    \includegraphics[scale=0.95]{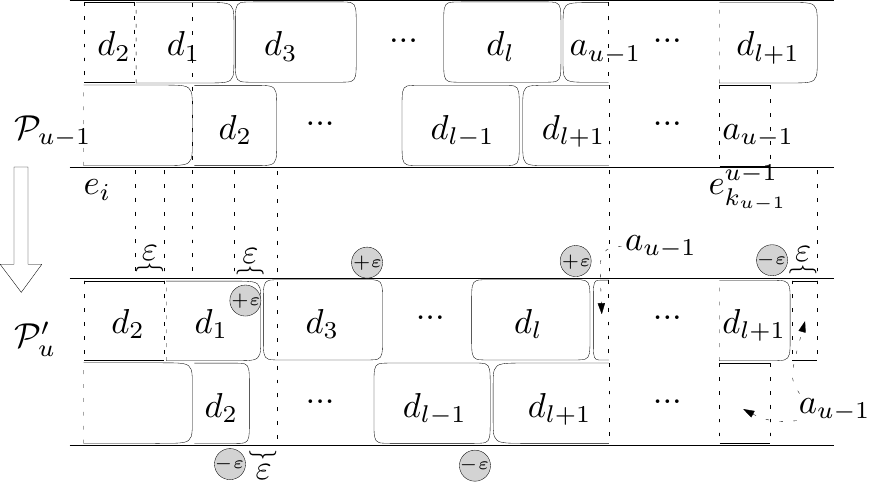}
    \caption{The transformation from $\cP_{u-1}$ to $\cP_{u}'$}
    \label{fig:pushing}
    \end{center}
    \end{figure}

If $\varepsilon\in\{\alpha, \beta, \gamma, \lceil \complTime{\cP'_u}{d_l} \rceil -\complTime{\cP'_u}{d_l}\}$, then the abnormality point of $\cP_u'$ is greater than $i$ and, having the required contradiction, we stop the iterative process with $T=u$.

For the two remaining values we have that the number of blocks in $[\complTime{\cP_{u}'}{d_l},\complTime{\cP_{u}'}{d_{l+1}}]$ in $\cP_{u}'$ is one less than the number of blocks in $[\complTime{\cP_{u-1}}{d_l},\complTime{\cP_{u-1}}{d_{l+1}}]$ in $\cP_{u-1}$. (We give an appropriate argument at the end of the proof of the lemma.)

Also, if $\varepsilon=e^{u-1}_{k_{u-1}+1}-e^{u-1}_{k_{u-1}}-\xi^{u-1}_{k_{u-1}}(a_{u-1})$, then the total completion time of $\cP'_u$ is smaller than the total completion time of $\cP_{u-1}$ provided that ($\xi^{u-1}_{k_{u-1}}(a_{u-1})=0$ and $\xi^{u-1}_{k_{u-1}-1}(d_{l+1})<e^{u-1}_{k_{u-1}}-e^{u-1}_{k_{u-1}-1}$) because $d_{l+1}$ completes in $\cP_u'$ strictly prior to $e_{\tau_u(d_{l+1})}=\complTime{\cP_u}{d_{l+1}}-\varepsilon$.
Then, we stop the iterative process with $T=u$.
Otherwise, set $a_u:=a_{u-1}$ and go to Step~2.

If $\varepsilon=\xi^{u-1}_{\tau_{u-1}(d_l)+1}(a_{u-1})$, then there is a job $a_u$ that starts or resumes at $\complTime{\cP_{u}'}{d_l}$.
Go to Step 2.

\paragraph{Step 2: Making $d_{l+1}$ cover $a_u$.}

Denote for brevity $\tau'=\tau_{\cP_u'}$ and let $\vect{e}'$ and $\vect{\xi}'$ be the events and the partition of $\cP_u'$, respectively.
If $d_{l+1}$ covers $a_u$ in
\[I_u'=\{\tau'(d_l)+1,\ldots,\min\{\tau'(d_{l+1}),\tau'(a_u)\}\},\]
then set $\cP_{u}=\cP_{u}'$
and go to Step~3. Otherwise we obtain $\cP_{u}$ from $\cP'_{u}$ as follows.
Find $t\in I_u'$ such that $\xi_t'(a_u)>0$ and $\xi_t'(d_{l+1})<e_{t+1}'-e_t'$.
By Lemma~\ref{lem:interlace}, $t<\tau'(a_u)$, and by Lemma~\ref{lem:how_jobs_finish}, $t<\tau'(d_{l+1})$. Let
\begin{eqnarray*}
 \varepsilon' &= \min\big\{ & e'_{\tau'(d_l)+1}-r(a_u)-\xi'_{\tau'(d_{l})}(a_u),\\
   & & \xi'_t(a_u),e'_{t+1}-e'_t-\xi'_t(d_{l+1}), \\
   & & (e'_{\tau'(d_l)+1}-e'_{\tau'(d_l)})/2-\xi'_{\tau'(d_{l})}(a_u)\big\}.
\end{eqnarray*}
We construct $\cP_{u}$ with events $\vect{e}^u$ and partition $\vect{\xi}^u$, where
\[
(\vect{e}^u,\vect{\xi}^u)=\cyclicshift{ \vect{e}',\vect{\xi}',\varepsilon', (t\cyclic{a_u}\tau'(d_l)\cyclic{d_{l+1}} t) },
\]
and then:
\begin{enumerate}[itemsep=0pt,label={$\bullet$},topsep=0pt]
 \item If $\varepsilon'\in\{e'_{\tau'(d_l)+1}-r(a_u)-\xi'_{\tau'(d_{l})}(a_u),(e'_{\tau'(d_l)+1}-e'_{\tau'(d_l)})/2-\xi'_{\tau'(d_{l})}(a_u)\}$, then we do \emph{$\varepsilon''$}-pushing of $d_l$ in $\cP'_{u}$ with
     \[\varepsilon''=\min\left\{\alpha, \beta, \gamma, \varepsilon'\right\}\]
     to get a schedule $\cP_{u}$ that has all blocks $j$, $j\in\{1,\ldots,i\}$, being $j$-normal --- in such case we stop the iterative process with $T=u$; (This $\varepsilon''$-pushing is shown in Figure~\ref{fig:alt-chains1}(a) with $x=a_u$ and $\varepsilon'=\varepsilon''$.)
 \item If $\varepsilon'=\xi'_t(a_u)$, then $\xi^{u}_t(a_u)=0$, i.e., $a_u$ is no longer in block $t$ as required;
 \item If $\varepsilon'=e'_{t+1}-e'_t-\xi'_t(d_{l+1})$, then $\xi^{u}_t(d_{l+1})=e^{u}_{t+1}-e^{u}_t$ as required.
\end{enumerate}

\smallskip
If $d_{l+1}$ does not cover $a_u$ in $I_u$ and $\varepsilon'>0$, then repeat Step~2 for $\cP_u$.
Thus, either in the resulting schedule $\cP_{u}$, $d_{l+1}$ covers $a_u$ in $I_u$, in which case go to Step 3,
or  $d_{l+1}$ does not cover $a_u$ in $I_u$ and $\varepsilon'=0$ (then $\xi^u_{\tau_u(d_l)}(a_u)=\xi^u_{\tau_u(d_l)}(d_{l+1})$
or $e^u_{\tau_u(d_l)+1}=r(a_u)$), in which case the abnormality point of $\cP_u$ is greater than $i$ and, having the required
contradiction, we stop the iterative process with $T=u$.

\paragraph{Step 3: Pushing $a_u$ out of  $[\complTime{\cP_u}{d_{l-1}},\complTime{\cP_u}{d_l}]$.}

If a part of $a_u$ executes in $[\complTime{\cP_u}{d_{l-1}},\complTime{\cP_u}{d_l}]$ in $\cP_u$, then perform an
\emph{$\varepsilon''$}-pushing of $d_l$ in $\cP_{u}$ as in Figure~\ref{fig:alt-chains1}(a) with $x=a_u$, where
\[
\varepsilon''=\min\left\{\alpha, \beta, \gamma, \lceil \complTime{\cP_{u}}{d_l}
\rceil -\complTime{\cP_{u}}{d_l},\xi^u_{\tau_u(d_l)}(a_u)\right\}.
\]
If $\varepsilon''=\min\left\{\alpha, \beta, \gamma, \lceil \complTime{\cP_{u}}{d_l} \rceil -\complTime{\cP_{u}}{d_l}\right\}$, then the abnormality point of the resulting schedule is greater than $i$ and, having the required contradiction, we stop the iterative process with $T=u$.
If $(d_1,\ldots,d_{l+1})$ is an alternating chain in $\cP_u$, then also stop with $T=u$.
Otherwise, go to Step~4.

\paragraph{Step 4: Moving to the next iteration.}
Set $u:=u+1$ and return to Step~1.

\medskip
We now briefly sketch the reminder of the proof. For the time being let us assume that the iteration process ends after $T\geq 1$ iterations, and that $\varepsilon> 0$ in Step~1 for $u>1$. We prove these two assumptions at the end of the proof.
In the following $\cP_u'$,  $\cP_u''$, and $\cP_u$ refer to the schedules obtained at the end of Steps 1, 2 and 3, respectively, $u\geq 1$.

Let $u=1$. Note that $\cP_1'=\cP$.
Then either $d_{l+1}$ executes without preemption in $[\complTime{\cP_1''}{d_l},\complTime{\cP_1''}{d_{l+1}}]$ or not.
In the former case, Step~3 ensures that $d_{l+1}$ executes in $\cP_1$ without preemption in $[\complTime{\cP_1}{d_{l-1}},\complTime{\cP_1}{d_{l+1}}]$.
The latter implies that \ref{it:chain2} holds for $\cP_1$ and by Claim~\ref{claim:A-free} below that $\cP_1$ is $\A$-free, which proves the lemma. We then have $T=1$.
If $d_{l+1}$ is preempted in  $[\complTime{\cP_1''}{d_l},\complTime{\cP_1''}{d_{l+1}}]$, then $T>1$ and it suffices to show that in this case we get a contradiction.
To that end we show that, if $T>1$, then $\cP_1$ is $d_{l+1}$-preempted (see Claims \ref{claim:I1-I2-I3}, \ref{claim:I4}, and \ref{claim:cP_1-ok} below), and if $\cP_{u-1}$ is $d_{l+1}$-preempted, then $\cP_u$ is $d_{l+1}$-preempted as well, $u\in\{2,\ldots,T-1\}$ (see Claims \ref{claim:I1-I2-I3}, \ref{claim:I4}, and \ref{claim:I5} below).
This process of generating $d_{l+1}$-preempted schedules cannot continue \emph{ad infinitum} since the process exits after $T$ iterations.
However, any exit schedule $\cP_T$ certifies that $\cP$ is not maximal which gives the required contradiction.

\medskip
We now proceed with details. Note that, for each $u\in\{1,\ldots,T\}$, \ref{it:chain0} and \ref{it:chain1} hold for $(d_1,\ldots,d_{l+1})$ in $\cP_{u}$ and
\begin{equation} \label{eq:a_d_finish0}
\xi^u_{\tau_u(d_l)}(a_u)=0 \land \xi^u_{\tau_u(d_l)+1}(a_u)>0,
\end{equation}
and that \ref{it:I1}, \ref{it:I2} and \ref{it:I3} follow directly from the definition of the iterative process above:
\begin{claim} \label{claim:I1-I2-I3}
Let $T>1$.
If $u=1$, or $u\in\{2,\ldots,T-1\}$ and $\cP_{u-1}$ is $d_{l+1}$-preempted, then $\cP_{u}$ satisfies conditions \ref{it:I1}, \ref{it:I2} and \ref{it:I3}.
\end{claim}

\begin{proof}
Note that, by construction, the total completion times and abnormality points of $\cP,\cP_1,\ldots,\cP_{T-1}$ are the same.
Also, thanks to Step~3 and the fact that $u<T$, $d_l$ and $d_{l+1}$ are the only two jobs executed in $[\complTime{\cP_u}{d_{l-1}},\complTime{\cP_u}{d_l}]$ for each $u\in\{1,\ldots,T-1\}$.
Finally, by Lemma~\ref{lem:alt_chain:normal}, we have $\complTime{\cP_u}{d_{l+1}}>\complTime{\cP_u}{d_l}$ for $u=1$, and by \ref{it:I5} for schedule $\cP_{u-1}$ we have $\complTime{\cP_{u-1}}{d_{l}}<\min\{r(x),r(y)\}<\complTime{\cP_{u-1}}{d_{l+1}}$ for $u>1$.
By construction, $\complTime{\cP_{u}}{d_{l}} <\lceil \complTime{\cP_{u-1}}{d_{l}} \rceil$ and $j_u<\tau_{u}(d_{l+1})$ such that $\jobs(\xi^u_{j_u})=\{x,y\}$.
Thus, $\complTime{\cP_{u}}{d_{l}}<\min\{r(x),r(y)\}<\complTime{\cP_{u}}{d_{l+1}}$.
Therefore, $\cP_u$ satisfies \ref{it:I1}.

By Lemma \ref{lem:shifting}, there is no idle time in block $\tau_{\cP}(d_l)+1$ in $\cP$. Thus the choice of $a_1$ ensures that it starts or resumes at $\complTime{\cP_1}{d_l}$ for $u=1$.
For $u>1$, the job $a_u$ always exists because there is no idle time in block $\tau_{u}(d_l)+1$ in $\cP_u$.
This follows from the fact that otherwise, by construction, there would be idle time in $[\complTime{\cP_{u-1}}{d_{l}}, \complTime{\cP_{u-1}}{d_{l+1}}]$ in $\cP_{u-1}$, which, since  \ref{it:I3} and \ref{it:I4} hold for $\cP_{u-1}$, implies that $a_{u-1}$ can be completed earlier in $\cP_{u-1}$, which contradicts its optimality.
Hence, $\cP_u$ satisfies \ref{it:I2}.

Finally, Steps~2 and 3 and $u<T$ ensure that \ref{it:I3} holds for $\cP_u$.
\end{proof}

\begin{claim} \label{claim:I4}
Let $T>1$.
If $u=1$, or $u\in\{2,\ldots,T-1\}$ and $\cP_{u-1}$ is $d_{l+1}$-preempted, then $\cP_{u}$ satisfies condition \ref{it:I4}.
\end{claim}

\begin{proof}
It suffices to argue that

\begin{equation} \label{eq:a_completes_later}
k_u<\tau_u(a_u).
\end{equation}
Denote for brevity $k=k_u$.
Suppose for a contradiction that $k=\tau_u(a_u)$.
By Lemma~\ref{lem:how_jobs_finish}, $\xi_k^u(a_u)=e^u_{k+1}-e^u_k$. By Claim \ref{claim:I1-I2-I3}, \ref{it:I3} holds for $\cP_u$.
Thus, $d_{l+1}$ covers $a_u$ in $I_u$ and, by Lemma~\ref{lem:interlace}, $\xi_k^u(d_{l+1})=e^u_{k+1}-e^u_k$.
Hence, there exists $a'\in\jobs$ such that $\complTime{\cP_{u}}{a'}=e^u_{k}$ because $e^u_k$ is an event in $\cP_u$.

If $a'\neq d_l$, then $\tau_u(d_l)<k$ and $\xi^u_{k-1}(a_u)=0$ because, again, $d_{l+1}$ covers $a_u$ in $I_u$.
Then, let $j<k-1$ be such that $\xi^u_j(a_u)>0$ and $\xi^u_{j'}(a_u)=0$ for each $j'\in\{j+1,\ldots,k-1\}$.
By \eqref{eq:a_d_finish0} such a $j$ exists.
By Lemma~\ref{lem:one_dominates} and the fact that $d_{l+1}$ covers $a_u$ in $I_u$, $\xi^u_j(a_u)=\xi^u_j(d_{l+1})=e^u_{j+1}-e^u_j$ which implies $\xi^u_j(a')=0$.
Lemma~\ref{lem:between_preemptions} applied to $a=a_u$, $a'$, $j$ and $j'=k$, leads to a contradiction.

It remains to consider the case when $a'=d_l$.
We have that $u>1$ because otherwise the jobs $d_l$ and $a_1$ form an $\A$-configuration in $\cP$ --- a contradiction since $\cP$ is $\A$-free.
We have
\begin{align} \label{eq:from_u-1_to_u}
 \begin{split}
 \complTime{\cP_{u-1}}{a_{u-1}} & >  \complTime{\cP_{u-1}}{d_{l+1}} \hspace*{\fill} \\
                                & \geq  \complTime{\cP_u}{d_{l+1}}\geq\complTime{\cP_u}{a_u}.
 \end{split}
\end{align}
The first inequality follows from \ref{it:I4} for $\cP_{u-1}$, which is $d_{l+1}$-preempted; the second inequality follows by construction of $\cP_u$, while the last one holds by assumption that $k=\tau_u(a_u)$.
Moreover, the construction ensures that the inequality $\complTime{\cP_u}{d_{l+1}}\geq\complTime{\cP_u}{a_u}$ implies
\begin{equation} \label{eq:from_u-1_to_u:2}
\complTime{\cP_{u-1}}{d_{l+1}}\geq\complTime{\cP_{u-1}}{a_u}.
\end{equation}
Thus, by \eqref{eq:from_u-1_to_u}, we have $\complTime{\cP_{u-1}}{a_{u-1}}>\complTime{\cP_{u-1}}{a_u}$.
Therefore, \eqref{eq:from_u-1_to_u:2} and \ref{it:I2}, \ref{it:I3} for $\cP_{u-1}$ imply that $a_{u}$ and $a_{u-1}$ interlace in $\cP_{u-1}$.
Finally, by \ref{it:I1} applied to $\cP_{u-1}$, $\cP_{u-1}$ is optimal and hence we arrive at a contradiction with Lemma~\ref{lem:interlace}.
Hence, \eqref{eq:a_completes_later} follows, which completes the proof of the lemma.
\end{proof}

\begin{claim} \label{claim:A-free}
If the job $d_{l+1}$ executes with no preemption in interval $[\complTime{\cP_1}{d_{l-1}},\complTime{\cP_1}{d_{l+1}}]$ in $\cP_1$, then $\cP_1$ is $\A$-free in $[\complTime{\cP_1}{d_{l}},\infty)$.
\end{claim}

\begin{proof}
In view of Lemma~\ref{lem:partially-free}, it suffices to prove that $\cP''_{1}$ at the end of Step 2 is $\A$-free in $[\complTime{\cP'}{d_{l}},\infty)$, if the step is not vacuous.
Let  $\tau(a_1)\equiv\tau_{\cP''_1}$ and $k=\min\{\tau(a_1),\tau(d_{l+1})\}$ for convenience.
Let $\vect{e}''$ and $\vect{\xi}''$ be the events and the partition of $\cP_{1}''$, respectively.
We start with an observation concerning the construction of $\cP''_1$, namely, the sequence of events, that is the start and completion times of jobs, is the same in $\cP=\cP'_1$ and $\cP''_{1}$.
More precisely,
\begin{enumerate}[itemsep=0pt,topsep=0pt,label={\normalfont(B\arabic*)},leftmargin=*]
\item\label{it:ll1} $\jobs(\xi_j)=\jobs(\xi''_j)$ for $j\in\{1,\ldots,\tau(d_{l})-1\}$ and $j\geq k+1$;
\item\label{it:ll2} $\jobs(\xi_{\tau_{\cP}(d_{l})})=\{d_l,d_{l+1}\}$ and $\jobs(\xi''_{\tau(d_{l})})=\{d_l,d_{l+1},a_1\}$;
\item\label{it:ll3} $\xi''_j(x)=\xi_j(x)$ for $x\notin \{a_1,d_{l+1}\}$, $\xi''_j(a_1)\leq \xi_j(a_1)$, $\xi''_j(d_{l+1})\geq \xi_j(d_{l+1})$, and $\xi''_j(a_1)+\xi''_j(d_{l+1})=\xi_j(a_1)+\xi_j(d_{l+1})$ for each $j\in\{\tau(d_l)+1,\ldots,k\}$.
\end{enumerate}

Suppose for a contradiction that $\cP''_{1}$ is not $\A$-free in interval $[\complTime{\cP''_{1}}{d_{l}},\infty)$.
Then, by \ref{it:ll1}-\ref{it:ll3} and since the schedule $\cP$ is $\A$-free in $[\complTime{\cP}{d_{l-1}},\infty)$, $a=a_1$ must be one of the two jobs that form an $\A$-configuration in $\cP''_{1}$.

Let $a$ and a job $x$ form an $\A$-configuration at $e^{''}_{\tau(a)}$ in $\cP''_{1}$.
By definition of $\A$-configuration, $x$ is not preempted in $(e''_j, e''_{\tau(a)}]$ for some $j<\tau(a)$ in $\cP''_{1}$ and $\xi''_{j-1}(x)<e''_{j-1}-e''_j$.
By \ref{it:ll2}, $\tau(d_l)+1\leq j$.
Thus, $\tau(d_l)+1\leq j \leq k$ for otherwise, by \ref{it:ll1}, $a$ and $x$ or form an $\A$-configuration in $\cP$ --- a contradiction.
Moreover, there is a block $t\in \{j,\ldots,k\}$ in $\cP$ with $\{a,x\}\subseteq \jobs(\xi_t)$ for otherwise again $a$ and $x$ form an $\A$-configuration in $\cP$ --- a contradiction.
Therefore, if $x=d_{l+1}$,
then $\xi_t(a)>0$ and $\xi_t(d_{l+1})=e_{t+1}-e_t$ in $\cP$ and it remains so in $\cP''_{1}$ which follows from the transformation in Step~1.
Thus, $a$ and $d_{l+1}$ do not form an $\A$-configuration in $\cP''_{1}$ which contradicts our assumption that $x=d_{l+1}$.
If $x\neq d_{l+1}$, then $\xi_t(d_{l+1})<e_{t+1}-e_t$ and hence a job $a'\in\{a,x\}$ with $\xi_{\tau_{\cP}(d_{l+1})}(d_{l+1})<e_{\tau_{\cP}(d_{l+1})+1}-e_{\tau_{\cP}(d_{l+1})}$ and $d_{l+1}$ interlace in $\cP$ when $\tau(a)>\tau(d_{l+1})=k$ --- a contradiction with Lemma~\ref{lem:interlace}.
Thus it remains to consider $\tau(a)\leq \tau(d_{l+1})$.
By Step 3, no preemption of $d_{l+1}$ in $[\complTime{\cP_1}{d_{l-1}},\complTime{\cP_1}{d_{l+1}}]$ in $\cP_1$ implies no preemption of $d_{l+1}$ in $[\complTime{\cP''_1}{d_{l}},\complTime{\cP''_1}{d_{l+1}}]$ in $\cP''_1$.
Therefore, the jobs $a$ and $x$ interlace in $\cP''_1$ if $\tau(d_l)\geq\tau(a)$ --- a contradiction by Lemma~\ref{lem:interlace}.

Now, suppose that $a$ and a job $x$ such that $a$ completes at $e''_{\tau(x)}$ form an $\A$-configuration in $\cP''_{1}$.
Since $d_{l+1}$ covers $a$ in $I_1'$ in $\cP''_{1}$, we have $x\neq d_{l+1}$.
By definition of $\A$-configuration, $a$ is not preempted in $(e''_j, e''_{\tau(x)}]$ for some $j\leq\tau(a)$ in $\cP''_{1}$ and $\xi''_{j-1}(a)<e''_{j-1}-e''_j$. Also, $\tau(d_l)+1\leq j \leq k$ for otherwise, by \ref{it:ll1}, $a$ and $x$ form an $\A$-configuration in $\cP$ --- a contradiction.
Moreover, there is a block $t\in \{j,\ldots,k-1\}$ in $\cP$ with $\{a,x\}\subseteq \jobs(\xi_t)$ for otherwise again $a$ and $x$ form an $\A$-configuration in $\cP$ --- a contradiction. Therefore, $\{d_{l+1},x\}\subseteq \jobs(\xi''_t)$, and
$a$ and $x$ interlace in $\cP''_{1}$ --- a contradiction by Lemma~\ref{lem:interlace}.
\end{proof}

\begin{claim} \label{claim:cP_1-ok}
If $T=1$, then $\cP_1$ is maximal, $\A$-free in $[\complTime{\cP}{d_{l}},\infty)$ and $(d_1,\ldots,d_{l+1})$ is an alternating chain in $\cP_1$. If $T>1$, then $\cP_1$ satisfies condition \ref{it:I5}.
\end{claim}

\begin{proof}
If $T=1$, then by the maximality of $\cP$ we have that $\cP_1$ is maximal and $(d_1,\ldots,d_{l+1})$ is an alternating chain in $\cP_1$.
By Claim~\ref{claim:A-free}, $\cP_1$ is $(\A,[\complTime{\cP}{d_{l}},\infty))$-free, which completes the proof in this case.

Suppose now that $T>1$.
Let for brevity $k=k_1$ and $\tau_1=\tau_{\cP_1}$ in the proof of Claim~\ref{claim:cP_1-ok}.
By Claim~\ref{claim:I4}, $\cP_1$ satisfies \ref{it:I4}.
Hence, $k=\tau_1(d_{l+1})$, and $\xi^1_k(d_{l+1})=e^1_{k+1}-e^1_k$ by Lemma~\ref{lem:how_jobs_finish}.
Note that $\xi^1_j(d_{l+1})=e^1_{j+1}-e^1_j$ for each $j\in\{\tau_1(d_l)+1,\ldots,k\}$ is not possible because then $(d_1,\ldots,d_{l+1})$ is an alternating chain in $\cP_1$ and hence the iterative process would stop with $T=1$ in Step~3.
We show that otherwise we can find the desired jobs $x$ and $y$ in \ref{it:I5}.
The key to finding the jobs is the existence of a block $j\in \{\tau_1(d_1)+1,\ldots,k-1\}$ such that $\xi^1_j(d_{l+1})<e^1_{j+1}-e^1_j$.
Take the smallest such $j$.
Let $\{x,y\}\subseteq \jobs(\xi^1_j)\setminus\{a_1, d_{l+1}\}$.
Such two jobs exist since there is no idle time in block $j$ and $a_1\notin\jobs(\xi^1_j)$ because, by  Claim  \ref{claim:I1-I2-I3}, $d_{l+1}$ covers $a_1$ in $I_1$.
We now prove that
\begin{equation} \label{eq:release_of_x}
r(x)> \complTime{\cP_1}{d_l} \quad\textup{ and }\quad r(y)> \complTime{\cP_1}{d_l}.
\end{equation}
To that end we first argue that no predecessor of $x$ or $y$ is executed in $(\complTime{\cP_1}{d_l},e^1_j]$.
By contradiction, suppose $z$ is a predecessor of $x$ or $y$ that completes in $(\complTime{\cP_1}{d_l},e^1_j]$.
Then, $z$ must also start in $[\complTime{\cP_1}{d_l},e^1_j]$ for otherwise by \eqref{eq:a_d_finish0} and the fact that $d_{l+1}$ covers $a_1$ in $I_1$ (by Claim \ref{claim:I1-I2-I3}) we get that $z$ interlaces either $a_1$ or $d_{l+1}$ in $\cP_1$ --- a contradiction by Lemma ~\ref{lem:interlace}.
Therefore, $z$ starts in $[\complTime{\cP_1}{d_l},e^1_j]$ and thus there is a block $j'\in\{\tau_1(d_l)+1,\ldots,j-1\}$ such that $\xi^1_{j'}(z)>\xi^1_{j'}(d_{l+1})$.
The latter is guaranteed by the fact that $z$ executes in $[\complTime{\cP_1}{d_l},e_j^1]$.
Since, by Claim \ref{claim:I1-I2-I3}, $d_{l+1}$ covers $a_1$ in $I_1$, we have $\xi^1_{j'}(a_1)=0$.
Therefore, there is a job $w$, $w\notin \{a_1,d_{l+1}\}$, such that $\xi^1_{j'}(w)>0$.
Thus, we get a contradiction by our definition of $j$ since $j'<j$.

Second, arguing by contradiction, suppose without loss of generality that $r(x)\leq \complTime{\cP_1}{d_l}$.
If $\xi^1_k(a_1)<e^1_{k+1}-e^1_k$, then take
\begin{eqnarray*}
 \varepsilon &= \min\big\{ & \xi_{\tau_1(d_l)+1}^{1}(a_1),e^1_{k+1}-e^1_k-\xi_k'(a_1),e^1_{j+1}-e^1_j-\xi^1_j(d_{l+1}), \\
             & & \xi^1_j(x),e^1_{\tau_1(d_l)+2}-e^1_{\tau_1(d_l)+1}-\xi^1_{\tau_1(d_l)+1}(x)\big\}
\end{eqnarray*}
and let $\cP''$ be a schedule with events $\vect{e}''$ and partition $\vect{\xi}''$, where
\[(\vect{e}'',\vect{\xi}'')=\cyclicshift{ \vect{e}',\vect{\xi}',\varepsilon, (\tau_1(d_l)+1\cyclic{a_1} k\cyclic{d_{l+1}} j\cyclic{x}\tau_1(d_l)+1) }.\]
Note that $j>\tau_1(d_l)+1$ because $\xi^1_j(a_1)=0$.
The assumption $r(x)\leq\complTime{\cP_1}{d_l}$ implies that $\cP''$ is feasible.
Thus, we get a contradiction since the total completion time of $\cP''$ is smaller than that of $\cP_1$.

On the other hand, if $\xi^1_k(a_1)=e^1_{k+1}-e^1_k$, then $x\notin\jobs(\xi^1_k)$, and \eqref{eq:a_d_finish0} and the fact that $d_{l+1}$ covers $a_1$ in $I_1$ (by Claim \ref{claim:I1-I2-I3}) imply that $\jobs(\xi^1_{\tau_1(d_l)+1})\cap\{x,y\}=\emptyset$.
Thus, $\complTime{\cP_1}{x}\leq e^1_k$ or $\complTime{\cP_1}{x}>e^1_{k+1}$.
In the former case $a_1$ and $x$ interlace, and in the latter case $x$ and $d_{l+1}$ interlace --- contradiction with Lemma~\ref{lem:interlace}.
Therefore, \eqref{eq:release_of_x} holds and $\cP_1$ satisfies \ref{it:I5} as required.
\end{proof}

\begin{claim} \label{claim:I5}
Let $T>1$. Then, $\cP_{u}$ satisfies condition \ref{it:I5} for $u\in\{1,\ldots,T-1\}$.
\end{claim}
\begin{proof}
By Claim \ref{claim:cP_1-ok}, $\cP_1$ satisfies \ref{it:I5}.
By induction on $u=1,\ldots,T-1$, we have $\min\{r(x),r(y)\}>\complTime{\cP_u}{d_l}$. Let $j_u$ be the earliest $j$ such that $\jobs(\xi^u_j)=\{x,y\}$ in $\cP_u$.
Again by induction on  $u=1,\ldots,T-1$ we have $j_u<\tau_u(d_{l+1})$. Therefore, $\cP_{u}$ satisfies condition  \ref{it:I5} for $u\in\{1,\ldots,T-1\}$.
\end{proof}

\begin{claim} \label{claim:epsilon}
For each $u\in\{1,\ldots,T\}$, if $\cP_{u}$ is $d_{l+1}$-preempted, then $\xi^{u}_{k_{u}}(a_{u})<e^{u}_{k_{u}+1}-e^{u}_{k_{u}}$.
\end{claim}
\begin{proof}
Suppose that $\xi^{u}_{k_{u}}(a_{u})=e^{u}_{k_{u}+1}-e^{u}_{k_{u}}$ and $\cP_u$ is $d_{l+1}$-preempted.
By \ref{it:I4}, it holds $k_u=\tau_u(d_{l+1})$ and by Lemma~\ref{lem:how_jobs_finish}, $\xi^{u}_{k_{u}}(d_{l+1})=e^{u}_{k_{u}+1}-e^{u}_{k_{u}}$.
Then, there is a job $a'$ that completes at $e_{k_{u}}$ because $e_{k_u}$ is an event in $\cP_u$.
The job $a'$ starts after $\complTime{\cP_{u}}{d_l}$ for otherwise $a'$ interlaces with either $a_u$ or $d_{l+1}$ in an optimal $\cP_{u}$ --- contradiction with with Lemma~\ref{lem:interlace}.
Perform swapping of $a'$ and $d_{l+1}$.
By Lemma \ref{lem:swapping}, this leads to a feasible schedule $\cP'$.
If $\xi_{k_u-1}(d_{l+1})<e_{k_u}-e_{k_u-1}$, then the total completion time of the new schedule is smaller than than that of $\cP_{u}$ --- a  contradiction again.
If  $\xi_{k_u-1}(d_{l+1})=e_{k_u}-e_{k_u-1}$, then there is a job $a''$ that completes at $e_{k_{u}-1}$ in $\cP_u$ because $e_{k_u-1}$ is an event in $\cP_u$.
Again, the job $a''$ starts after $\complTime{\cP_{u}}{d_l}$ for otherwise $a''$ interlaces with either $a_u$ or $d_{l+1}$ in an optimal $\cP_{u}$ --- contradiction with with Lemma~\ref{lem:interlace}.
Clearly, $\complTime{\cP_u}{a''}=\complTime{\cP'}{a''}$ and $\startTime{\cP_u}{a''}=\startTime{\cP'}{a''}$.
Take $\cP_u:=\cP'$ and $a':=a''$ and repeat the above swapping.

After a finite number of the above `swappings' we obtain a feasible schedule with lower total completion time than the initial one --- a  contradiction.  This completes the proof of Claim~\ref{claim:epsilon}.
\end{proof}

Now we return to the proof of Proposition~\ref{pro:infinite_chain}.
If $T=1$, then the lemma holds by Claim \ref{claim:cP_1-ok}.
To complete the proof we show that $T>1$ leads to a contradiction.
First, for $T>1$, Claims~\ref{claim:I1-I2-I3},~\ref{claim:I4}, and  \ref{claim:cP_1-ok} imply that the schedule $\cP_1$ is $d_{l+1}$-preempted. Thus, Claims~\ref{claim:I1-I2-I3}, \ref{claim:I4} and~\ref{claim:I5} and an induction on $u\in\{1,\ldots,T-1\}$ give that each of the schedules $\cP_1,\ldots,\cP_{T-1}$ is $d_{l+1}$-preempted.
Since the iterative process exits in iteration $u=T$, producing either $\cP'_T$ or $\cP''_T$, we get a contradiction since either these two exit schedules have smaller total completion times than $\cP_{T-1}$ or they have the same total completion times but their abnormality points are greater than $i$.
Hence, $\cP_{T-1}$ is not maximal, contradicting the fact that it is $d_{l+1}$-preempted.

It remains to show that $T$ exists, i.e., that the number of iterations is finite.
To that end let $c_u$ and $C_u$ be the numbers of jobs executed in $\cP_u$ which are not covered by $d_{l+1}$, and the number of blocks, respectively, in $[\complTime{\cP_u}{d_{l}},\complTime{\cP_u}{d_{l+1}}]$.

By Claim~\ref{claim:epsilon}, $\varepsilon>0$ in Step~1 of the construction of $\cP_u'$ for each $u\in\{2,\ldots,T\}$.
If $\varepsilon=\xi^{u-1}_{\tau_{u-1}(d_l)+1}(a_{u-1})$ in Step~1, then the block $\tau_{u-1}(d_l)+1$ disappears.
Otherwise, $\varepsilon=e^{u-1}_{k_{u-1}+1}-e^{u-1}_{k_{u-1}}-\xi^{u-1}_{k_{u-1}}(a_{u-1})$ in Step~1.
If $\xi^{u-1}_{k_{u-1}}(a_{u-1})>0$, then $a_u=a_{u-1}$ at the end of Step~1 and $\xi^{u}_{k_{u}}(a_{u})=e^{u}_{k_{u}+1}-e^{u}_{k_{u}}$.
Therefore, both Steps~2 and~3 are vacuous and thus $\cP'_u=\cP_u$ is $d_{l+1}$-preempted.
However, $\xi^{u}_{k_{u}}(a_{u})=e^{u}_{k_{u}+1}-e^{u}_{k_{u}}$ contradicts Claim \ref{claim:epsilon}.
Thus, we have $\xi^{u-1}_{k_{u-1}}(a_{u-1})=0$ in Step~1, and therefore the block $k_{u-1}$ disappears.

Next, Step 2 does not change the number of blocks.
Finally, Step~3, if not vacuous, may increase the number of blocks by at most one.
Thus, we have $C_u=C_{u-1}-1$ and $c_u=c_{u-1}$, if Step 3 is vacuous, and $C_u\leq C_{u-1}$ and $c_u=c_{u-1}-1$, if Step 3 is not vacuous.
Consequently $C_u+c_u<C_{u-1}+c_{u-1}$
and $T\leq C_1+c_1\leq 3n$.
Thus, the iterative process indeed stops with some schedule $\cP'_T$ or $\cP''_T$.
This, by Claim~\ref{claim:cP_1-ok}, completes the proof of the proposition.

\section{How many preemptions of a job is required?}
\label{sec:lower_bound}

In this section we show that, for any given number $n$ of jobs, it is sometimes necessary to preempt a job $p=\Omega(\log n)$ times.
Let $A_i$, $i\geq 0$, be a set of four jobs $a_1^i,a_2^i,a_3^i,a_4^i$ such that $r(a_j^i)=2i$ for each $j\in\{1,2,3\}$, $r(a_4^i)=2i+1$ and $a_j^i\prec a_4^i$ for each $j\in\{1,2,3\}$.
Then, define
\[
\jobs_p=\bigcup_{i=0}^{p}A_{i},
\]
where $a_4^{i}\prec a_4^{i+1}$ for each $i\in\{0,\ldots,p-1\}$. (See Figure~\ref{fig:precedence_constraints}.)

\begin{figure}[htb]
\begin{center}
\includegraphics[scale=1]{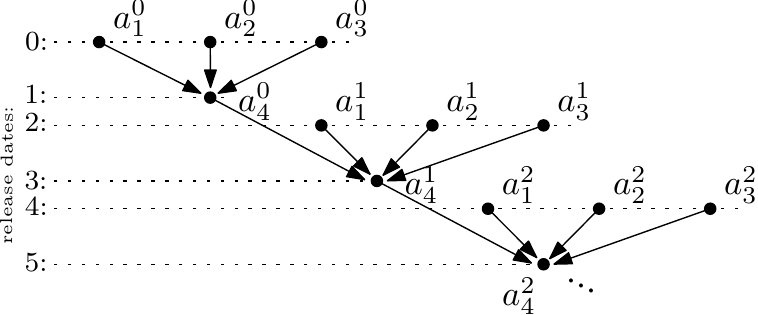}
\caption{The precedence constraints and release dates for $\jobs_p$, $p\geq 2$}
\label{fig:precedence_constraints}
\end{center}
\end{figure}

We prove that the job $a_4^p$ should complete exactly at $2p+3-1/2^{p+1}$ in any optimal schedule.
This is done by first proving that no valid schedule (optimal or not) can complete $a_4^p$ earlier
(cf. Claim~\ref{clm:lower-bound1}), and then by proving that staring $a_4^p$ later leads to a schedule
that cannot be optimal (cf.\ Claim~\ref{clm:lower-bound3}).

\begin{claim} \label{clm:lower-bound1}
Let $p\geq 0$ be any integer.
If $\cP$ is a preemptive schedule for $\jobs_p$, then the total length of the job $a_4^p$ executing in $[2p+2,+\infty)$ is at least $1-1/2^{p+1}$.
\end{claim}

\begin{proof}
We prove the lemma by induction on $p$.
Let $p=0$. (Note that $\jobs_p=A_0$.)
Executing less than $1/2$ units of $a_4^0$ in $[2,+\infty)$ implies that $a_4^0$ completes at $5/2-\varepsilon$, for some $\varepsilon>0$
This, however, requires completing each job in $\A_0\setminus\{a_4^0\}$ at $3/2-\varepsilon$ or earlier, which is not possible.

Suppose that the lemma holds for integers smaller than $p$ and we prove it for $p$.
Let $\cP$ be a preemptive schedule for $\jobs_p$.
We consider all jobs that must execute in time interval $I=[2p,\startTime{\cP}{a_4^p}]$.
Each job in $A_{p}\setminus\{a_4^{p}\}$ executes in this interval, because $r(a)=2p$ and $a\prec a_4^{p}$ for each $a\in A_{p}\setminus\{a_4^{p}\}$.
By induction hypothesis and by the facts that $a\prec a_4^p$ for each $a\in\jobs_{p-1}$, we obtain that a part of $a_4^{p-1}$ that executes in $I$ is of length at least $1-1/2^{p}$.
Thus, the total length of all jobs that execute in $I$ in $\cP$ is at least $4-1/2^{p}$.
Therefore,
\[\startTime{\cP}{a_4^{p}}\geq 2p+|I|/2=2p+2-1/2^{p+1}.\]
Thus, at least $1-1/2^{p+1}$ units of $a_4^p$ execute in $[2p+2,+\infty)$ as required.
\end{proof}

We iteratively construct a schedule $\cP_p$.
Let $\cP_0$ be such that the jobs $a_1^0,a_2^0,a_3^0$ form a $3/2$ schedule in interval $[0,3/2]$ and $a_4^0$ executes in $[3/2,5/2]$.
For $p>0$, first take $\cP_{p-1}$ and then execute the jobs in $A_p$ as follows:
\begin{eqnarray*}
 a_1^p & \textup{in} & [2p,2p+1],   \\
 a_2^p & \textup{in} & [2p+1-1/2^{p},2p+1-1/2^{p+1}]\\
       &             & \quad\cup [2p+1,2p+2-(1/2^{p}-1/2^{p+1})]  \\
       &             & = [2p+1-1/2^{p},2p+1-1/2^{p+1}] \\
       &             & \quad\cup [2p+1,2p+2-1/2^{p+1}],  \\
 a_3^p & \textup{in} & [2p+1-1/2^{p+1},2p+2-1/2^{p+1}],   \\
 a_4^p & \textup{in} & [2p+2-1/2^{p+1},2p+3-1/2^{p+1}].  
\end{eqnarray*}
(See Figure~\ref{fig:many_preemptions}.)

\begin{figure*}[htb]
\begin{center}
\includegraphics[scale=0.7]{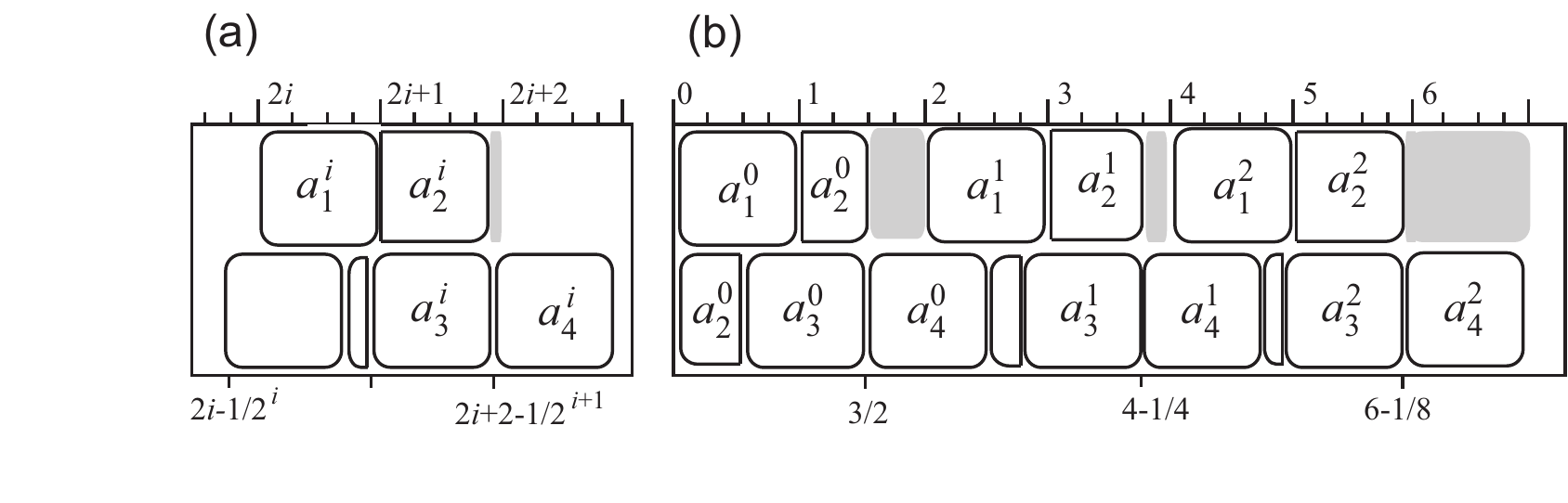}
\caption{(a) the execution of the jobs in $A_i$ in $\cP_p$;
         (b) the schedule $\cP_p$ for $\jobs_2$}
\label{fig:many_preemptions}
\end{center}
\end{figure*}

Note that, when $p>0$, $\complTime{\cP_{p}}{a_4^{p-1}}=\complTime{\cP_{p-1}}{a_4^{p-1}}=2p+1-1/2^{p}$ and hence $\cP_p$ is valid.
This gives the following.

\begin{claim} \label{clm:lower-bound2}
Let $p\geq 0$ be any integer.
There exists a schedule for $\jobs_p$ that completes $a_4^p$ at $2p+3-1/2^{p+1}$. \qed
\end{claim}

\begin{claim} \label{clm:lower-bound3}
Let $p\geq 0$ be any integer.
Each optimal schedule for $\jobs_p$ completes $a_4^p$ at $2p+3-1/2^{p+1}$ and satisfies the following: the total length of idle time in $[0,2(j+1)]$ is $1-2^{j+1}$ for each $j\in\{0,\ldots,p\}$, and there is no time interval contained in $[0,2(p+1)]$ in which both processors are idle.
\end{claim}

\begin{proof}
We prove the lemma by contradiction, i.e., suppose that, for some $p\geq 0$, there exists an optimal schedule $\cP_p'$ such that $\complTime{\cP_p'}{a_4^p}\neq 2p+3-1/2^{p+1}$.
Let, without loss of generality, $p$ be the minimum integer for which this holds.
One can verify the lemma for $p=0$ and hence $p>0$.
By Claim~\ref{clm:lower-bound1},
\begin{equation} \label{eq:a_4-late}
 \complTime{\cP_p'}{a_4^p}>2p+3-1/2^{p+1}.
\end{equation}
By Claim~\ref{clm:lower-bound2} and by the minimality of $p$, there exists an optimal schedule $\cP_{p-1}$ for $\jobs_{p-1}$ that executes at least $1-1/2^{p}$ units of $a_4^{p-1}$ in $[2p,+\infty)$.

We argue that the total length of $a_4^{p-1}$, denoted by $x$, that executes in $\cP_p'$ in interval $[2p,+\infty)$ equals exactly $1-1/2^{p}$.
By Claim~\ref{clm:lower-bound1},
\begin{equation} \label{eq:x-big}
 x\geq 1-1/2^{p}.
\end{equation}
Define a schedule $\cP'$ that equals $\cP_{p-1}$ in the interval $[0,2p]$ and equals $\cP_p'$ in the interval $[2p,+\infty)$.
Note that $\cP'$ is not valid only if the total length of $a_4^p$ executing in $\cP'$ (that equals $1/2^{p}+x$) is greater than $1$.
However,
\begin{equation} \label{eq:bound-on-cp}
\begin{split}
\complTime{\cP'}{a}=\complTime{\cP_p'}{a}\textup{ for each }a\in A_p \quad\textup{and}\\
\sum_{a\in\jobs_{p-1}\setminus\{a_4^{p-1}\}}\complTime{\cP'}{a}\leq\sum_{a\in\jobs_{p-1}\setminus\{a_4^{p-1}\}}\complTime{\cP_p'}{a}.
\end{split}
\end{equation}
We then obtain a schedule $\cP$ by removing the total length of $x-1+1/2^{p}$ of $a_4^{p-1}$ from $\cP'$ in a way that minimizes the completion time of $a_4^{p-1}$ in $\cP$, which gives $\complTime{\cP}{a_4^{p-1}}\leq\complTime{\cP'}{a_4^{p-1}}-(x-1+1/2^{p})$.
The schedule $\cP$ is valid and $\complTime{\cP}{a}=\complTime{\cP'}{a}$ for each $a\in\jobs_p\setminus\{a_4^{p-1}\}$.
Hence, by \eqref{eq:bound-on-cp},
\[\sum_{a\in\jobs_p}\complTime{\cP}{a}\leq-(x-1+1/2^{p})+\sum_{a\in\jobs_p}\complTime{\cP}{a}.\]
Thus, by the optimality of $\cP_p'$, $x-1+1/2^{p}=0$, i.e., $x=1-1/2^{p}$ as required.

In the schedule $\cP_p'$, the jobs that are executed in time interval $[2p,\complTime{\cP_p'}{a_4^p}]$ are the ones in $A_p$ and $x$ units of $a_4^{p-1}$.
By a case analysis one can prove that
\[\sum_{a\in (A_p\cup\{a_4^{p-1}\})\setminus\{a_4^p\}}\complTime{\cP_p'}{a}\geq 4\cdot 2p + 4 + 2(1-1/2^{p}).\]
Note that, by construction,
\[\sum_{a\in (A_p\cup\{a_4^{p-1}\})\setminus\{a_4^p\}}\complTime{\cP_p}{a} = 4\cdot 2p + 4 + 2(1-1/2^{p}).\]
Moreover, $x=1-1/2^{p}$ and the minimality of $p$ imply that
\[\sum_{a\in\jobs_{p-1}\setminus\{a_4^{p-1}\}}\complTime{\cP_p}{a}\leq\sum_{a\in\jobs_{p-1}\setminus\{a_4^{p-1}\}}\complTime{\cP_p'}{a}.\]
This, \eqref{eq:a_4-late} and $\complTime{\cP_p}{a_4^p}=2p+3-1/2^{p+1}$ imply that $\cP_u'$ is not optimal.
This gives the desired contradiction.

\medskip
Note that is follows that in each optimal schedule $\cP$ for $\jobs_p$, $\complTime{\cP}{a_4^j}=2p+3-1/2^{j+1}$.
Since $a_4^j$ is a successor of all jobs in $\jobs_j$, we obtain that there is, in $\cP$, idle time on exactly one processor in interval $(2(j+1)-1/2^{j+1},2(j+1))$.
Also note that, for each $j\in\{0,\ldots,p\}$, no idle time is possible in the interval $I=(2j,2(j+1)-1/2^{j+1})$ because, for $j>0$, $a_1^,a_2^j,a_3^j$ and a part of $a_4^{j-1}$ of length $1-2^j$ must execute in $I$, and for $j=0$, $a_1^0,a_2^0$ and $a_3^0$ must execute in $I$.
This completes the proof of the claim.
\end{proof}

Note that for any given $p\geq 0$, the number of jobs in $\jobs_p$ equals $4(p+1)$.
Claim~\ref{clm:lower-bound3} proved above gives us that each optimal schedule for $\jobs_p$ has a job that completes at a time point that is a multiple of $1/2^{p+1}$ but is not a multiple of $1/2^{p}$.
This gives that exists a set of $n$ jobs $\jobs$ such that there exists no optimal solution to $P2|pmtn,in\textup{-}tree,r_j,p_j=1|\sum C_j$ for $\jobs$ in which each job start, completion and preemption occurs at a time point that is a multiple of $1/2^{n/4-1}$.
We note that this upper bound on the resolution of this problem is slightly weaker than the bound $2^{-(n-1)/3}$ proved in \cite{CNT13}.
Our main result of this section is the following lower bound on the number of preemption of one job in an optimal schedule.

\begin{theorem} \label{arbitrary}
Given any positive integer $p$, there exists an instance of the problem $P2|pmtn,in\textup{-}tree,r_j,p_j=1|\sum C_j$, such
that in any optimal schedule for the instance, this is a job that is preempted at least $p=\Omega(\log |\jobs|)$ times, where
$\jobs$ is the job set of the instance.
\end{theorem}

\begin{proof}
Let $\cP_p$ be an optimal schedule for $\jobs$. Thus $\cP_p$ satisfies the conditions in Claim~\ref{clm:lower-bound3}.
By Theorem~\ref{thm:normal_exist}, we may assume that $\cP_p$ is normal.
Define $l=\complTime{\cP_p}{a_4^p}\cdot 2^{c}$, where $c=2|\jobs_p|+3$.

Take $\jobs=\{a\}\cup\jobs_p\cup\{b_1,\ldots,b_l\}$.
The precedence constraints between the jobs in $\jobs_p$ are as in Figure \ref{fig:precedence_constraints}. We extend the precedence relation to $\jobs$ by additionally enforcing:
\[a_4^p\prec b_1 \prec b_2 \prec \cdots \prec b_l \quad\textup{and}\quad a\prec b_1.\]

We first construct a schedule $\cP'$ as follows.
Take $\cP_p$ and extend it by executing $a$ so that $\complTime{\cP'}{a}=2(p+1)+1/2^{p+1}$ and executing $b_1,\ldots,b_l$ so that $\complTime{\cP'}{b_i}=\complTime{\cP'}{b_{i-1}}+1$ for each $i\in\{1,\ldots,l\}$, where $b_0=a_4^p$.
By Claim~\ref{clm:lower-bound3}, such a schedule $\cP'$ exists and $a$ is preempted $p$ times in $\cP'$.

\smallskip
Let $\cP$ be an optimal normal schedule for $\jobs$.
By Theorem~\ref{thm:normal_exist}, such a schedule exists.
Suppose for a contradiction that $a$ is preempted at most $p-1$ times in $\cP$.
By Claim~\ref{clm:lower-bound3}, $\complTime{\cP}{a}<\complTime{\cP'}{a}$ and $\complTime{\cP}{a_4^p}>\complTime{\cP'}{a_4^p}$.
The number of events in $\cP$ (respectively in $\cP'$) in interval $[0,\complTime{\cP}{a_4^p}]$ (respectively, $[0,\complTime{\cP'}{a_4^p}]$) is at most $2|\jobs_p|+3$ because each event either equals $0$ or is the start or completion time of a job.
Since both schedules are normal,
\[
 \complTime{\cP}{a_4^p}-\complTime{\cP'}{a_4^p}\geq 1/2^{2|\jobs_p|+3}.
\]
Thus, by definition of $l$,
\begin{align} \label{eq:comple-times:difference}
\begin{split}
\sum_{i=1}^l\complTime{\cP}{b_i} & \geq  \sum_{i=1}^l\complTime{\cP'}{b_i} + l/2^{2|\jobs_p|+3} \\
                                 &  =    \sum_{i=1}^l\complTime{\cP'}{b_i} + \complTime{\cP_p}{a_4^p}.
\end{split}
\end{align}
By Claim~\ref{clm:lower-bound3}, $\cP$ restricted to the jobs in $\jobs_p$ is not optimal for $\jobs_p$.
This in particular implies
\[\sum_{x\in\jobs_p}\complTime{\cP}{x} > \sum_{x\in\jobs_p}\complTime{\cP'}{x}.\]
This, together with \eqref{eq:comple-times:difference}, implies
\begin{eqnarray*}
\sum_{x\in\jobs}\complTime{\cP}{x} & = & \complTime{\cP}{a} + \sum_{x\in\jobs_p}\complTime{\cP}{x} + \sum_{i=1}^l\complTime{\cP}{b_i} \\
     & > & \sum_{x\in\jobs_p}\complTime{\cP'}{x} + \sum_{i=1}^l\complTime{\cP'}{b_i} + \complTime{\cP_p}{a_4^p}
\end{eqnarray*}
Since, by construction of $\cP'$,
\[\complTime{\cP_p}{a_4^p}=\complTime{\cP'}{a_4^p}\geq\complTime{\cP'}{a},\]
we obtain that the total completion time of $\cP'$ is strictly smaller than that of $\cP$, which gives a required contradiction.

Finally, note that $\complTime{\cP_p}{a_4^p}\leq|\jobs_p|\leq c$ and hence $l\leq  2^{2c}$.
Thus, $|\jobs|=|\jobs_p|+1+l=2^{O(p)}$ because $c=O(p)$ and $|\jobs_p|=O(p)$.
This implies that $p=\Omega(\log|\jobs|)$ as required.
\end{proof}

\section{Summary}
In this paper we have provided some structural characterization of the preemptions in optimal schedules for the problem
$P2|pmtn,in\textup{-}tree,r_j,p_j|\sum C_j$. The advantages of our characterization are as follows:
\begin{itemize}
 \item It narrows down a search space of optimal solutions from an infinite one to a finite one.
 \item The understanding of the possible structure of preemptions is a step towards determining the complexity of the problem
       $P2|pmtn,in\textup{-}tree,r_j,p_j|\sum C_j$. On the one hand, the normality of an optimal schedule may lead us to a
       polynomial-time algorithm. On the other hand, the fact that a single job may need many (of the order of $\log n$) preemptions
       as stated in Theorem~\ref{arbitrary} could be useful in proving NP-completeness, although the complexity of the problem $P2|pmtn,in\textup{-}tree,r_j,p_j|\sum C_j$ is left as an interesting and challenging open problem.
 \item It significantly improves the lower bound on the resolution for the problem $P2|pmtn,in\textup{-}tree,r_j,p_j|
       \sum C_j$.
\end{itemize}

Note that we rely on the in-trees precedence constraints between the jobs in our proof.
This assumption is crucial when proving in Section~\ref{subsec:A-configurations} that a maximal $\A$-free schedule exists.
The generalization of our result to arbitrary precedence constraints or providing an example that an analogous statement as the one in Corollary~\ref{cor:main} is false for more general precedence constraints is left as an open problem.
 
\begin{center}
\textbf{Acknowledgements}
\end{center}

This research has been supported by the Natural Sciences and Engineering Research
Council of Canada (NSERC) Grant OPG0105675. Wieslaw Kubiak was also supported by the Polish National Science Center
research grant.
Dariusz Dereniowski was partially supported by Polish National Science Center under contract DEC-2011/02/A/ST6/00201 and a scholarship for outstanding young researchers founded by the Polish Ministry of Science and Higher Education.

\bibliographystyle{plain}
\bibliography{scheduling}

\end{document}